\documentclass[11pt,letter]{article}

\usepackage{amsmath,amsthm,amsfonts,float,graphicx}
\usepackage{nccmath}
\usepackage{multirow} %tabbing
\usepackage{hyperref}
\usepackage{enumitem}
\usepackage{float}
\usepackage[capitalize]{cleveref} % apparently must be loaded before algorithm
\usepackage{tikz}
\usetikzlibrary{arrows}

\usepackage{mleftright}
\usepackage{fullpage}
\usepackage{microtype}

\usepackage{thm-restate}
\usepackage{framed}

\newcommand{\eps}{\epsilon}
\newcommand{\dtv}{d_{\mathrm TV}}
\newcommand{\norm}[1]{\lVert#1{\rVert}}
\newcommand{\normone}[1]{{\norm{#1}}_1}
\newcommand{\normtwo}[1]{{\norm{#1}}_2}
\newcommand{\norminf}[1]{{\norm{#1}}_\infty}

\newcommand{\eqdef}{\stackrel{\rm def}{=}}

\newcommand{\parent}[1]{\mathrm{Par}(#1)}
\newcommand{\structure}{\mathcal{S}}
\newcommand{\poly}{\mathrm{poly}}

\newcommand{\poisson}[1]{\ensuremath{\operatorname{Poi}\mleft( #1 \mright)}}
\newcommand{\bernoulli}[1]{\ensuremath{\operatorname{Bern}\mleft( #1 \mright)}}

\newtheorem{theorem}{Theorem}
\newtheorem{corollary}{Corollary}
\newtheorem{question}{Question}
\newtheorem{proposition}{Proposition}
\newtheorem{fact}{Fact}
\newtheorem{claim}{Claim}
\newtheorem{lemma}{Lemma}
\newtheorem*{unnumberedfact}{Fact}
\theoremstyle{definition}
\newtheorem{definition}{Definition}
\theoremstyle{remark}
\newtheorem{remark}{Remark}

\newcommand{\Var}{\mathop{\textnormal{Var}}\nolimits}
\newcommand{\expect}[1]{\mathbb{E}\!\left[#1\right]}
\newcommand{\shortexpect}{\mathbb{E}}

\newcommand{\dyes}{{\cal Y}}
\newcommand{\dno}{{\cal N}}
\newcommand{\R}{\mathbb{R}}

\newcommand{\setOfSuchThat}[2]{ \left\{\; #1 \;\colon\; #2\; \right\} } 			% sets such as "{ elems | condition }"
\newcommand{\bigabs}[1]{\left\lvert#1\right\rvert}
\newcommand{\probaCond}[2]{\proba\!\left[\, #1 \;\middle\vert\; #2\, \right]}
\newcommand{\indic}[1]{\textbf 1_{#1}}
\newcommand{\lp}[1][1]{\ell_{#1}}
\newcommand{\binomial}[2]{\ensuremath{\operatorname{Bin}\!\left( #1, #2 \right)}}
\newcommand{\bigTheta}[1]{{\Theta\left( #1 \right)}}
\newcommand{\bigOmega}[1]{{\Omega\left( #1 \right)}}
\newcommand{\indicSet}[1]{\indic{\left\{#1\right\}}}
\newcommand{\totalvardist}[2]{\totalvardistrestr[]{#1}{#2}}
\newcommand{\totalvardistrestr}[3][]{{\dtv^{#1}\!\left({#2, #3}\right)}}
\newcommand{\expectCond}[2]{\mathbb{E}\!\left[\, #1 \;\middle\vert\; #2\, \right]}
\newcommand{\flr}[1]{\left\lfloor #1 \right\rfloor}
\newcommand{\Cov}{\mathrm{Cov}}
\newcommand{\tildeO}[1]{\tilde{O}\left( #1 \right)}

\newcommand{\probaDistrOf}[2]{\proba_{#1}\left[\, #2\, \right]}
\newcommand{\proba}{\Pr}
\newcommand{\probaOf}[1]{\proba\!\left[\, #1\, \right]}

\newcommand{\bigO}[1]{{O\mleft( #1 \mright)}}
\newcommand{\tildeOmega}[1]{\operatorname{\tilde{\Omega}}\left( #1 \right)}
\newcommand{\Poi}{\mathop{\textnormal{Poi}}\nolimits}

\newcommand{\E}{\mathbb{E}}
\newcommand{\abs}[1]{\mleft\lvert#1\mright\rvert}

\newcommand{\dkl}[2]{\operatorname{D}( #1\| #2 )}
\newcommand{\hellinger}[2]{\operatorname{d_{H}}( #1, #2 )}

\newcommand{\mutualinfo}[2]{ I\left(#1; #2\right) }
\newcommand{\entropy}[1]{ H\left(#1\right) }
\newcommand{\condentropy}[2]{ H\left(#1 \mid #2\right) }
\begin{document}

\title{Testing Bayesian Networks}

\author{Cl\'ement~L.~Canonne
        \and Ilias~Diakonikolas
         \and Daniel~M.~Kane
        \and Alistair~Stewart% <-this % stops a space
\thanks{An extended abstract of this work was presented at the 30th Conference on Learning Theory, {COLT} 2017~\cite{CDKS:17}.}
\thanks{C.~L. Canonne is with IBM Research; part of this work was performed while he was a postdoctoral fellow at Stanford University, and before that a graduate student at Columbia University.}%
\thanks{I. Diakonikolas is with the University of Wisconsin–Madison; this work was performed while he was at the University of Southern California.}
\thanks{A. Stewart is with Web3 Foundation; this work was performed while he was a postdoctoral researcher at the University of Southern California.}% <-this % stops a space
\thanks{D.~M.~Kane is with the University of California, San Diego.}
}
\maketitle

\begin{abstract}
This work initiates a systematic investigation of testing \emph{high-dimensional} structured 
distributions by focusing on testing \emph{Bayesian networks}~--~the prototypical family of directed graphical models. A Bayesian network 
is defined by a directed acyclic graph, where we associate a random variable with each node. 
The value at any particular node is conditionally independent of all the other non-descendant nodes once its parents are fixed. 
Specifically, we study the properties of identity testing and closeness testing of Bayesian networks. Our main contribution is 
the first non-trivial efficient testing algorithms for these problems and corresponding information-theoretic lower bounds. 
For a wide range of parameter settings, our testing algorithms have sample complexity \emph{sublinear} in the dimension
and are sample-optimal, up to constant factors.
\end{abstract}

%%%%%%%%%%%%%%%%%%%%%%%%%%%%%%%%%%%%%%%%%%%%%%%%%%%%%%%%%%%%%%%%%%%%%%%%%%%%%%%%
%%%%%%%%%%%%%%%%%%%%%%%%%%%%%%%%%%%%%%%%%%%%%%%%%%%%%%%%%%%%%%%%%%%%%%%%%%%%%%%%
\clearpage
\tableofcontents
\clearpage

\section{Introduction} \label{sec:intro}

\subsection{Background} \label{ssec:background}

Distribution testing has its roots in statistical hypothesis testing~\cite{NeymanP,lehmann2005testing}
and was initiated in~\cite{GRexp:00, BFR+:00}.
The paradigmatic problem in this area is the following: given sample access to an arbitrary distribution
$P$ over a domain of size $N$, determine whether $P$ has some global property or is ``far''
from any distribution having the property. A natural way to solve this problem would be to
learn the distribution in question to good accuracy, and then check if the
corresponding hypothesis is close to one with the desired property.
However, this testing-via-learning approach requires $\Omega(N)$ samples and is typically suboptimal.
The main goal in this area is to obtain {\em sample-optimal} testers~--~ideally, testers that draw
$o(N)$ samples from the underlying distribution.
During the past two decades, a wide range of properties have been studied,
and we now have sample-optimal testers for many of these properties~\cite{Paninski:08, CDVV14, VV14, DK:16, DiakonikolasGPP16}.

We remark that even for the simplest properties, e.g., identity testing,
at least $\Omega(\sqrt{N})$ many samples are required for  %properties of
arbitrary distributions over $N$ atoms.
While this is an improvement over the $\Omega(N)$ samples required
to learn the distribution, a sample upper bound of $O(\sqrt{N})$ is still impractical if $N$ is very large.
For example, suppose that the unknown distribution is supported on $\{0, 1\}^n$. For this high-dimensional setting,
a sample complexity bound of $\Theta(2^{n/2})$ quickly becomes prohibitive, when the dimension increases.
Notably, the aforementioned $\Omega(\sqrt{N})$ sample lower bound
characterizes worst-case instances, which in many cases are unlikely to arise in real-world data.
This observation motivates the study of testing {\em structured} distribution families,
where significantly improved testers may be possible. Hence, the following natural question arises:
%\begin{question} \label{q:struc}
{\em Can we exploit the structure of the data to perform the desired testing task more efficiently?}
%\end{question}

A natural formalization of this question involves
viewing the data as samples from a  {\em probabilistic model}~--~a model
that we believe represents the random process generating the samples.
The usual assumption is that there exists a known
family of probabilistic models~--~describing a set of probability distributions~--~and that the data are random samples drawn from an unknown distribution in the family.
In this context, the distribution testing problem is the following:
Let $\mathcal{C}$ be a family of {probabilistic models}.
The {\em testing algorithm} has access to independent samples from an unknown $P \in \mathcal{C}$,
and its goal is to output ``yes'' if $P$ has some property $\mathcal{P}$,
and output ``no''  if the total variation distance, $\dtv(P,Q) \eqdef (1/2) \normone{P-Q}$, 
{where $\| \cdot \|_1$ denotes the $L_1$-norm,} is at least $\eps$ to {\em every} $Q \in \mathcal{C}$ that has property $\mathcal{P}$.
The sample complexity of this structured testing problem depends on the underlying family
$\mathcal{C}$, and we are interested in obtaining efficient algorithms that are {\em sample optimal} for $\mathcal{C}$.

More than a decade ago, Batu, Kumar, and Rubinfeld~\cite{BKR:04} considered a specific instantiation of this broad question~--~testing the equivalence between two unknown discrete monotone distributions~--~and obtained
a tester whose sample complexity is poly-logarithmic in the domain size. A recent sequence of works~\cite{DDSVV13, DKN:15, DKN:15:FOCS}
developed a framework to obtain sample-optimal estimators
for testing the identity of structured distributions {\em over total orders} (e.g., univariate multi-modal or log-concave distributions).
The main lesson of these works is that, under reasonable structural assumptions,
the sample complexity of testing may dramatically improve~--~becoming sub-logarithmic or even independent of the support size. Moreover, in all studied cases,
one obtains testers with {\em sub-learning} sample complexity.

\subsection{This Work: Testing High-Dimensional Structured Distributions} \label{ssec:this-work}
This paper initiates a systematic investigation of testing properties of {\em high-dimensional} structured
distributions. One of the most general formalisms to succinctly represent such distributions
is provided by probabilistic graphical models~\cite{Wainwright:2008, Koller:2009}.
Graphical models compactly encode joint probability distributions in high dimensions.
Formally, a graphical model is a graph where we associate a random variable with each node.
The key property is that the edge-structure of the graph determines the dependence relation
between the nodes.

The general problem of inference in graphical models is of fundamental importance
and arises in many applications across several scientific disciplines,
see~\cite{Wainwright:2008} and references therein.
In particular, the task of learning graphical models
has been extensively studied~\cite{Neapolitan:2003, RQS11}.
A range of information-theoretic and algorithmic results have been developed
during the past five decades in various settings,
see, e.g.,~\cite{Chow68, Dasgupta97, Friedman96, Friedman1997, Friedman00, Cheng02,
Chickering02, Margaritis03, Abbeel:2006, WainwrightRL06, AnandkumarHHK12,
SanthanamW12, LohW12, DiakonikolasKS16b} for a few references.
In contrast,
the general question of {\em testing} graphical models has received less attention.
We propose the following broad set of questions:\footnote{In what follows, by \emph{learning} we refer to the task of, given i.i.d. samples from an arbitrary distribution $P\in\mathcal{C}$, outputting a hypothesis distribution $\hat{P}$ such that, with high probability, $P$ and $\hat{P}$ are close. The number of samples required for this task is then the \emph{sample complexity of learning $\mathcal{C}$}.}

\begin{question} \label{q:gm}
{\em Let $\mathcal{C}$ be a family of high-dimensional graphical models and $\mathcal{P}$ be a property of
$\mathcal{C}$. What is the {\em sample complexity} of testing whether an unknown $P \in \mathcal{C}$ has property
$\mathcal{P}$? Can we develop testers for $\mathcal{P}$ with {\em sub-learning} sample complexity?
Can we design {\em sample-optimal} and {\em computationally efficient} testers?}
\end{question}

We believe that Question~\ref{q:gm} points to a fundamental research direction that warrants study for its own sake.
Moreover, as we explain in the following paragraphs, such estimation tasks
arise directly in various practical applications across the data sciences,
where sample efficiency is of critical importance. Hence, improved estimators for these tasks
may have implications for the analysis of datasets in these areas.

For concreteness, Question~\ref{q:gm} refers to a single unknown distribution that we have sample access to.
We are also naturally interested in the broader setting of testing properties for {\em collections} of distributions in
$\mathcal{C}$. Before we proceed to describe our contributions,
a few comments are in order: As previously mentioned, for all global properties of interest (e.g., identity,
independence, etc.), the sample complexity of testing the property is bounded from above by the sample complexity
of learning an arbitrary distribution from $\mathcal{C}$. Hence, the overarching goal is to obtain
testers that use fewer samples than {are required to actually learn} the model~--~or to prove that this is impossible.
On a related note, in the well-studied setting of testing arbitrary discrete distributions,
the main challenge has been to devise sample-optimal testers; the algorithmic aspects are typically straightforward.
This is no longer the case in the high-dimensional setting, where the combinatorial structure of the underlying model
may pose non-trivial algorithmic challenges.

In this work, we start this line of inquiry by focusing on testing
{\em Bayesian networks}~\cite{Pearl88} ({\em Bayes nets} or BN for brevity),
the prototypical family of {\em directed} graphical models. Bayesian networks are used for modeling beliefs
in many fields including robotics, computer vision, computational biology, natural language processing,
and medicine~\cite{Jensen:2007, Koller:2009}. Formally, a Bayesian network
is defined by a directed acyclic graph (DAG) $\structure = (V, E)$,
where we associate a random variable with each node. Moreover, the value
at any particular node is conditionally independent of all the other non-descendant nodes once its parents are fixed. 
Hence, for a fixed topology, it suffices to specify each node's conditional distribution, 
for each configuration of its parents' values.

The main problems that we study in this setting
are the related tasks of testing {\em identity} and {\em closeness}:
In identity testing, we are given samples from an unknown Bayes net $P$
and we want to distinguish between the case that it is identical to
versus significantly different from an explicitly given Bayes net $Q$.
In closeness testing, we want to test whether two unknown Bayes nets $P, Q$
are identical versus significantly different. {We believe that our techniques can be naturally
adapted to test other related properties (e.g., independence), 
but we have not pursued this direction in the current paper.}
A related testing problem that we consider is that of {\em structure testing}:
given samples from an unknown Bayes net $P$, we want to
test whether it can be represented with a given graph structure $\mathcal{S}$
or is far from any Bayes net with this structure.

In the prior work on testing {\em unstructured} discrete distributions,
the natural complexity measure was the domain size of the unknown distributions.
For the case of Bayes nets, the natural complexity measures are the number of variables
(nodes of the DAG)~--~denoted by $n$~--~the maximum in-degree of the DAG~--~denoted by $d$~--~and the alphabet size of the discrete distributions on the nodes. To avoid clutter in the expressions,
we focus on the natural setting that the random variables associated with each node are Bernoulli's,
i.e., the domain of the underlying distributions is $\{0, 1\}^n$. (As we will point out, our bounds
straightforwardly extend to the case of general alphabets with a necessary
polynomial dependence on the alphabet size.)

We note that Bayes nets are a universal representation scheme:
{\em Any} distribution over $\{0, 1\}^n$ can be presented as a BN,
if the maximum in-degree $d$ of the graph is unbounded.
(Indeed, for $d = n-1$, {one} can encode all distributions over $\{0, 1\}^n$.)
In fact, as we will see, the sample complexity of testing scales exponentially with $d$.
Therefore, an upper bound on the maximum in-degree is {\em necessary} to obtain non-trivial upper bounds.
Indeed, the most interesting regime is the settting where the number of nodes $n$ is large
and the degree $d$ is small. {In applications of interest, this assumption will be automatically satisfied.
In fact, {as we explain in the following subsection,} 
in many relevant applications the maximum in-degree is either 
$1$ (i.e., the underlying graph is a tree) or bounded by a small constant.

\subsection{Related Work} \label{ssec:related}
We partition the related work intro three groups corresponding to research efforts by different communities.

\paragraph{Computer Science}
A large body of work in computer science has focused on designing 
statistically and computationally efficient algorithms for learning structured distributions
in both low and high dimensions~\cite{Dasgupta:99, FreundMansour:99short, AroraKannan:01, VempalaWang:02, CGG:02, 
MosselRoch:05, MoitraValiant:10, BelkinSinha:10, DDS12soda, DDS12stoc, CDSS13, 
DDOST13focs, CDSS14, CDSS14b, HardtP15, ADLS15, DDS15, DDKT15, DKS15, DKS16}. On the other hand, the vast majority of the literature in distribution property testing during the past two decades focused on
arbitrary discrete distributions, where the main complexity measure was the domain size.
See~\cite{BFR+:00, BFFKRW:01, Batu01, BDKR:02, BKR:04, Paninski:08, ValiantValiant:11,
DDSVV13, DJOP11, LRR11, ILR12, CDVV14, VV14, ADK15, CDGR16, DK:16} for a sample of works, or~\cite{Rub12,Canonne15} for surveys.

A line of work~\cite{BKR:04, DDSVV13, DKN:15, DKN:15:FOCS} studied properties of one-dimensional
structured distribution families under various ``shape restrictions'' on the underlying density.
In the high-dimensional setting, Rubinfeld and Servedio~\cite{RubinfeldServedio:05}
studied the identity testing problem for monotone distributions over $\{0, 1\}^n$.
It was shown in~\cite{RubinfeldServedio:05}  that $\poly(n)$ samples suffice for the case of uniformity testing,
but the more general problems of identity testing and independence testing  require $2^{\Omega(n)}$ samples.
Subsequently, Adamaszek, Cjumaj, and Sohler~\cite{AdamaszekCS10}
generalized these results to continuous monotone distributions over $[0, 1]^n$.
A related, yet distinct, line of work studied the problem of testing {\em whether} a probability distribution
has a certain structure~\cite{BKR:04, BFRV11, ADK15, CDGR16}.
The sample complexity bounds in these works scale exponentially with the dimension.

\paragraph{Statistics}
The area of hypothesis testing for high-dimensional models has a long history in statistics
and is currently an active topic of study. A sequence of early and recent works,
starting with~\cite{Weiss60, Bickel69, LS93}, {has} studied the problem
of testing the equivalence between two nonparametric high-dimensional
distributions in the asymptotic regime. In the parametric setting, Hotelling's
T-squared statistic~\cite{Hotelling1931} is the classical test for the equivalence
of two high-dimensional Gaussians (with known and identical covariance).
However, Hotelling's test has the serious defect that it fails when the
sample size is smaller than the dimension of the data~\cite{BaiS96}.
Recent work has obtained testers that, under a high-dimensional Gaussian model (with known covariance),
succeed in the sub-linear regime for testing identity~\cite{Sriv08}
and closeness~\cite{chen2010}. A number of more recent works study properties of
covariance matrices~\cite{Cai2013}, regression~\cite{Jav14}, and linear independence testing~\cite{RamdasISW16}.

\paragraph{Applications}
The problems of testing identity and closeness of Bayesian networks arise in a number of applications
where sample efficiency is critical~\cite{Friedman00, GWJ03, Sobel03, Almudevar10, Nguyen2011, SEG14, StadM15, Yin2015}.
In bioinformatics applications (e.g., gene set analysis), each sample corresponds to an experiment that
may be costly or ethically questionable~\cite{Yin2015}. Specifically,~\cite{Yin2015} emphasizes
the need of making accurate inferences on {\em tree structured} Bayesian networks, using an extremely small
sample size~--~significantly smaller than the number of variables (nodes). \cite{Almudevar10} studies
the problem of testing closeness between two unknown Bayesian network models
in the context of a biology application, where Bayes nets are used to model gene expression data.
The motivation in~\cite{Almudevar10}  comes from the need to compare network models for a common set of genes
under varying phenotypes, which can be formulated as the problem of testing closeness between two unknown Bayes nets.
As argued in~\cite{Almudevar10}, due to the small sample size available, it is not feasible to directly
learn each BN separately.

\subsection{Basic Notation and Definitions}
Consider a directed acyclic graph (DAG), $\structure$, with $n$ vertices that are {topologically sorted, i.e.,} labelled {from the set}
$[n] \eqdef \{1,2,\ldots,n\}$ so that all directed edges {of $\structure$} point from vertices with smaller {label} to vertices with larger {label}.
A probability distribution $P$ over $\{0,1\}^n$ is defined to be
a \emph{Bayesian network} {(or Bayes net)} with {dependency} graph $\structure$
if for each $i \in [n]$, we have that $\Pr_{X\sim P}\left[X_i = 1 \mid X_1,\ldots,X_{i-1}\right]$
depends only on the values $X_j$,
where $j$ is a parent of $i$ in $\structure$. Such a distribution $P$ can be specified by
its {\em conditional probability table}, i.e., the vector of conditional probabilities of $X_i=1$
conditioned on every possible combination of values {to} the coordinates of $X$
at the parents of $i$.

To formalize the above description, we use the following terminology.
We will denote by $\parent{i}$ the set of parents of node $i$ in $\structure$.
For a vector $X = (X_1, \ldots, X_n)$ and a subset $A \subseteq [n]$, we use $X_A$ to denote the 
vector $(X_i)_{i \in A}$. We can now give the following definition:

\begin{definition} \label{def:BN-terminology}
Let $S$ be the set $\{(i,a): i \in [n], a\in \{0,1\}^{{|}\parent{i}{|}}\}$ and $s=|S|$. 
For $(i,a)\in S$, the \emph{parental configuration} $\Pi_{i,a}$
is defined to be the event that $X_{\parent{i}} =a$.
Once $\structure$ is fixed, we may associate with a Bayesian network $P$
the conditional probability table $p\in [0,1]^{S}$ given by
$p_{i,a}=\Pr_{X\sim P}\left[X_i=1 \mid \Pi_{i,a}\right]$, {for $(i, a) \in S$.}
We note that the distribution $P$ is determined by $p$.

We will frequently index $p$ as a vector.
That is, we will use the notation $p_k$, for $1 \leq k \leq s$,
and the associated events $\Pi_k$, where each $k$ stands for an $(i,a) \in S$ lexicographically ordered.
\end{definition}
%%%%%%%%%%%%%%%%%%%%%%%%%%%%%%%%%%%%%%%%%%%%%%%%%%%%%%%%%%%%%%%%%%%%%%%%%%%%%%%%%%%%%%%%%%%%%%%%%%%%%%%%%%%
%%%%%%%%%%%%%%%%%%%%%%%%%%%%%%%%%%%%%%%%%%%%%%%%%%%%%%%%%%%%%%%%%%%%%%%%%%%%%%%%%%%%%%%%%%%%%%%%%%%%%%%%%%%
%%%%%%%%%%%%%%%%%%%%%%%%%%%%%%%%%%%%%%%%%%%%%%%%%%%%%%%%%%%%%%%%%%%%%%%%%%%%%%%%%%%%%%%%%%%%%%%%%%%%%%%%%%%

\section{Our Results and Techniques} \label{sec:results-techniques}

The structure of this section is as follows: In~\cref{ssec:results}, we provide the statements
of our main results in tandem with a brief explanation of their context and the relations between them. We outline the organization o the paper in~\cref{sec:struc}, before discussion concurrent work in~\cref{sec:related}.
% In~\cref{ssec:techniques}, we give a detailed outline of our algorithmic and lower bound
% techniques.

\subsection{Main Results} \label{ssec:results}
The focus of this paper is on the properties of \emph{identity testing} and \emph{closeness testing} 
of Bayes nets. We give the first non-trivial efficient
testing algorithms and matching information-theoretic lower bounds for these problems.
For a wide range of parameter settings, our algorithms achieve sub-learning sample complexity and 
are sample-optimal (up to constant factors).

For concreteness, we consider Bayes nets over Bernoulli random variables.
We note that our upper bounds straightforwardly extend to general alphabets with a
polynomial dependence on the alphabet size {(see Remark~\ref{rem:alph})}.
Let $\mathcal{BN}_{n, d}$ denote the family of Bernoulli
Bayes nets on $n$ variables such that the corresponding DAG
has maximum in-degree at most $d$. For most of our results, we will think
of the dimension $n$ as being large and the maximum degree
$d$ as being comparably small (say, bounded from above by a constant
or at most logarithmic in $n$).

For the inference problems of learning and testing Bayes nets, there are two versions
of the problem: The first version corresponds to the setting
where the structure of the graph is fixed (and known a priori to the algorithm).
In the second version, both the graph and the parameters are unknown to the algorithm.
We note that both versions of the problem are interesting, based on the application.
The {\em unknown structure} setting is clearly at least as hard, and typically includes an algorithm
for the {\em fixed structure} case plus additional algorithmic ingredients.

Before we give the statements of our main testing results, we state a nearly tight
bound on the sample complexity of learning $\mathcal{BN}_{n, d}$. This bound
will be used as a baseline to compare against our efficient testers:

\begin{fact} \label{fact:learning-sample}
The sample complexity of learning $\mathcal{BN}_{n, d}$, within total
variation distance $\eps$, with confidence probability $9/10$, is:
(i) $\widetilde{\Theta}(2^d \cdot n/\eps^2)$, for all $d \leq n/2$, in the fixed structure setting,
and (ii) $\widetilde{\Theta}(2^d \cdot n/\eps^2)$ in the unknown structure setting.
\end{fact}

We give a proof of this fact in~\cref{sec:learn}.
\cref{fact:learning-sample} characterizes the sample complexity of learning Bayes nets
(up to logarithmic factors). We remark that our information-theoretic upper bound
for the fixed structure case also yields a simple computationally efficient algorithm.
The unknown structure regime is much more challenging computationally.
For this setting, we provide a nearly tight information-theoretic upper bound
that is non-constructive. (The corresponding algorithm runs in exponential time.)
In fact, we note that no sample-optimal computationally efficient algorithm is known
for unknown structure Bayes nets.

\medskip

Our first main result concerns the fixed structure regime.
For technical reasons, we focus on Bayes nets that satisfy
a natural balancedness condition. Roughly speaking, our balancedness condition
ensures that the conditional probabilities are bounded away from $0$ and $1$,
and that each parental configuration happens with some minimum probability.
Formally, we have:

\begin{definition} \label{def:balanced-net}
A Bayes net $P$ over $\{0,1\}^n$ with structure $\structure$ is called
\emph{$(c,C)$-balanced} if, for all $k$, we have that (i) $p_k\in[c,1-c]$,
and (ii) $\probaDistrOf{P}{\Pi_k} \geq C$.
\end{definition}

Under a mild condition on the balancedness, we give sample-optimal and
computationally efficient algorithms for testing identity and closeness of Bayes nets.
Specifically, for the problem of identity testing against an explicit distribution,
we require that the explicit distribution be balanced (no assumption is needed for the unknown Bayes net).
For the problem of closeness testing, we require that one of the two unknown distributions be balanced.
We are now ready to state our first main theorem:

\begin{theorem}[Testing Identity and Closeness of Fixed--Structure Bayes Nets]\label{thm:informal-identity-closeness-known}
For testing identity and closeness of fixed structure Bayes nets $P, Q$
with $n$ nodes and maximum in-degree $d$, there is an efficient
algorithm that uses $\bigO{2^{d/2}\sqrt{n}/\eps^2}$ samples and, assuming that one of $P, Q$
is $(c,C)$-balanced with $c=\tildeOmega{1/\sqrt{n}}$ and $C=\tildeOmega{d\eps^2/\sqrt{n}}$,
correctly distinguishes between the cases that $P=Q$ versus $\normone{P-Q} > \eps$,
with probability at least $2/3$.
Moreover, this sample size is information-theoretically optimal, up to constant factors,
for all $d < n/2$, even for the case of uniformity testing.
\end{theorem}

\noindent The conceptual message of~\cref{thm:informal-identity-closeness-known}
is that, for the case of fixed structure, testing is information-theoretically easier than learning.
Specifically, our result establishes a quadratic gap between learning and identity testing,
reminiscent of the analogous gap in the setting of unstructured discrete distributions.
We remark here that the information-theoretic lower bounds of Fact~\ref{fact:learning-sample} (i)
hold even for Bayes nets with constant balancedness.

We now turn our attention to the case of unknown structure.
Motivated by~\cref{thm:informal-identity-closeness-known}, it would be tempting
to conjecture that one can obtain testers with sub-learning sample complexity
in this setting as well. Our first main result for unknown structure testing
is an information-theoretic lower bound, showing that this is not the case.
Specifically, even for the most basic case of tree-structured Bays Nets ($d=1$)
with unknown structure, uniformity testing requires $\Omega(n/\eps^2)$ samples.
It should be noted that our lower bound applies even for Bayes nets with {\em constant}
balancedness. Formally, we have:

\begin{theorem}[Sample Lower Bound for Uniformity Testing of Unknown Tree-Structured Bayes Nets]
\label{thm:informal-uniformity-lower-unknown}
Any algorithm that, given sample access to a balanced tree-structured Bayes net $P$ over $\{0, 1\}^n$,
distinguishes between the cases $P=U$ and $\normone{P-U} > \eps$ {(where $U$ denotes the uniform distribution over $\{0,1\}^n$)},
with probability $2/3$, requires $\Omega(n/\eps^2)$ samples from $P$.
\end{theorem}

At the conceptual level, our above lower bound implies
that in the unknown topology case~--~even for the simplest non-trivial case of degree-$1$ Bayes nets~--~identity testing is information-theoretically essentially as hard as learning.
That is, in some cases, no tester with sub-learning sample complexity exists.
We view this fact as an interesting phenomenon that is absent from the previously
studied setting of testing unstructured discrete distributions.

\cref{thm:informal-uniformity-lower-unknown}  shows that
testing Bayes nets can be as hard as learning. However, it is still possible that
testing is easier than learning in \emph{most} natural situations. For the sake of intuition,
let us examine our aforementioned lower bound more carefully.
We note that the difficulty of the problem originates from the fact that
the explicit distribution is the uniform distribution, which can be thought of
as having any of a large number of possible structures. We claim that
this impediment can be circumvented if the explicit distribution satisfies
some non-degeneracy conditions. Intuitively, we want these conditions
to ensure \emph{robust identifiability} of the structure: that is, that any (unknown)
Bayes net sufficiently close to a non-degenerate Bayes net $Q$
must also share the same structure.

For tree structures, there is a very simple non-degeneracy condition.
Namely, that for each node, the two conditional probabilities for that node
(depending on the value of its parent) are non-trivially far from each other.
For Bayes nets of degree more than one, our non-degeneracy condition
is somewhat more complicated to state, but the intuition is still simple:
By definition, non-equivalent Bayesian network structures satisfy
different conditional independence constraints. Our
non-degeneracy condition rules out some of these possible
new conditional independence constraints,
as \emph{far} from being satisfied
by the non-degenerate Bayesian network.
Let $\gamma >0$ be a parameter quantifying non-degeneracy.
Under our non-degeneracy condition, we can design a {\em structure tester}
with the following performance guarantee:

\begin{theorem}[Structure Testing for Non-Degenerate Bayes Nets]\label{thm:informal-struct-test}
Let $\structure$ be a structure of degree at most $d$ and $P$
be a degree at most $d$ Bayes net over $\{0, 1\}^n$ with structure $\structure'$
whose underlying undirected graph has no more edges than $\structure$. There is an algorithm that uses
 $\bigO{(2^d+ d \log n)/\gamma^2}$ samples from $P$, runs in time $\bigO{n^{d+3}/\gamma^2}$,
 and distinguishes between the following two cases with probability at least $2/3$:
(i) $P$ can be expressed as a degree-$d$ Bayes net
with structure $\structure$ that is $\gamma$-non-degenerate; 
or (ii) $P$ cannot be expressed as a Bayes net with structure $\structure$.
\end{theorem}

By invoking the structure test of the above theorem, we can reduce the identity testing with unknown structure
to the case of known structure, obtaining the following:

\begin{theorem}[Testing Identity of Non-Degenerate Unknown Structure Bayes Nets]
\label{thm:informal-upper-identity-unknown-nondegen}
There exists an algorithm with the following guarantees. Let $Q$ be a degree-$d$ Bayes net $Q$ over $\{0, 1\}^n$,
which is $(c,C)$ balanced and $\gamma$-non-degenerate for
$c=\tildeOmega{1/\sqrt{n}}$ and $C=\tildeOmega{d\eps^2/\sqrt{n}}$.  
Given the description of $Q$, $\eps >0$,
and sample access to a distribution $P$ promised to be a degree-$d$ Bayes net
with no more edges than $Q$, the algorithm takes
$\bigO{2^{d/2}\sqrt{n}/\eps^2+(2^d+ d \log n)/\gamma^2}$ samples from $P$,
runs in time $\bigO{n}^{d+3}(1/\gamma^2+1/\eps^2)$,
and distinguishes with probability at least $2/3$ between (i) $P=Q$ and (ii) $\normone{P-Q}  > \eps$.
\end{theorem}

\noindent We remark that we can obtain an analogous result for the problem of testing closeness.
See~\cref{sec:closeness}.

We have shown that, without any assumptions, testing is almost as hard as learning for the case of trees.
An interesting question is whether this holds for high degrees as well. We show that for the case
of high degree sub-learning sample complexity is possible.
We give an identity testing algorithm for degree-$d$ Bayes nets with unknown structure, 
\emph{without} balancedness or degeneracy assumptions.
While the dependence on the number of nodes $n$ of this tester is suboptimal, 
it does essentially achieve the ``right'' dependence on the degree $d$, that is $2^{d/2}$:

\begin{theorem}[Sample Complexity Upper Bound of Identity Testing] \label{thm:informal-sample-identity-general}
Given the description of a degree-$d$ Bayes net $Q$ over $\{0, 1\}^n$, $\eps>0$,
and sample access to a degree-$d$ Bayes net $P$, we can distinguish between the cases
that $P=Q$ and $\normone{P-Q}  > \eps$, with probability at least $2/3$, using
$2^{d/2}\poly(n,1/\eps)$ samples from $P$.
\end{theorem}

(See~\cref{thm:unknown-structure-identity:informationtheoretic} for a more detailed statement
handling closeness testing as well.)
The message of this result is that when the degree $d$ increases, specifically for $d = \Omega(\log n)$,
the sample complexity of testing becomes lower than the sample complexity of learning.
We also show an analogue of~\cref{thm:informal-sample-identity-general}
for closeness testing of two unknown Bayes nets, under the additional assumption
that we know the topological ordering of the unknown DAGs.

\subsection{Organization} \label{sec:struc}

%%%%%%%%%%%%%%%%%%%%%%%%%%%%%%%%%%%%%%%%%%%%%%%

This paper is organized as follows:
In \cref{sec:prelim}, we give the necessary definitions and tools we will require. 
\cref{sec:product:identity} contains our matching upper and lower bounds for identity testing of product distributions, and is followed in~\cref{sec:product} by our matching upper and lower bounds for \emph{closeness} testing of product distributions. \cref{sec:product:to:bn:overview} then provides an overview of what is required to generalize these techniques from product distributions to Bayes nets, discussing at a high-level the following sections. 
In~\cref{sec:identity-known} we study the identity testing for Bayes nets with known structure:
We give an identity tester that works under a mild balancedness condition on the explicit Bayes net distribution, 
and also show that the sample complexity of our algorithm is optimal, up to constant factors.
In~\cref{sec:identity-uknown}, we study the identity testing for unknown structure Bayes nets:
We start by proving a sample complexity lower bound showing that, for the unknown structure regime,
uniformity testing is information-theoretically as hard as learning~--~even for the case of trees.
We then show that this lower bound can be circumvented under a natural non-degeneracy condition
on the explicit Bayes net distribution. Specifically, we give an identity tester with 
sub-learning sample complexity for all low-degree non-degenerate Bayes nets.
Our identity tester for unknown structure non-degenerate Bayes nets relies on a novel structure tester
that may be of interest in its own right.  \cref{sec:closeness} studies the corresponding closeness testing problems 
for both known and unknown structure Bayes nets.
Finally, in~\cref{sec:it:ub} we consider the case of high-degree Bayes nets and
obtain testers for identity and closeness of unknown-structure Bayes nets.
Our testers in this section have optimal (and sub-learning) sample complexity as a function of 
the maximum in-degree $d$ 
and polynomial dependence in the dimension $n$.

\subsection{Concurrent and Independent Work}\label{sec:related}
Contemporaneous work by \cite{DaskalakisP16} studies the identity testing problem
for Bayes nets with {\em the same known} graph structure. Using different arguments, 
they obtain a tester with sample complexity $\tilde{O}(2^{(3/4)d} \cdot n/\eps^2)$ and running time $O_{\eps}(n^{d+1})$ for this problem.
This sample bound is comparable to that of our~\cref{thm:unknown-structure-identity:informationtheoretic}
(that works without assumptions on the parameters), having the right dependence on $n, 1/\eps$ (as follows from our 
Theorem~\ref{thm:informal-uniformity-lower-unknown} and Fact~\ref{fact:learning-sample}) and a sub-optimal dependence on the degree $d$.
As previously mentioned, a sample complexity of $\Omega(n/\eps^2)$ is relevant for high-degree Bayes nets.
For the case of low-degree (which is the main focus of our paper), 
one can straightforwardly obtain the same sample bound just by learning the distribution (Fact~\ref{fact:learning-sample}).
\cite{DaskalakisP16} also obtain an $O(n^{1/2}/\eps^2)$ upper bound for testing identity
against a known product, matching our~\cref{theo:identity:product:ub}. 
(This sample bound is optimal by our~\cref{theo:lb:product:uniform}.)

\section{Preliminaries} \label{sec:prelim}
In this section, we provide the basic definitions and technical tools we shall use throughout this paper.

\paragraph{Basic Notation and Definitions}
The $L_1$-distance between two discrete probability distributions $P, Q$ supported on a set $A$ is defined as 
$\normone{P-Q} = \sum_{x \in A} \abs{P(x)-Q(x)}$. Our arguments will make essential use 
of related distance measures, specifically the Kullback--Leibler (KL) divergence, defined as 
$\dkl{P}{Q} = \sum_{x \in A} P(x) \log \frac{P(x)}{Q(x)}$, and the Hellinger distance, defined as
$\hellinger{P}{Q} = (1/\sqrt{2}) \cdot \sqrt{ \sum_{x \in A} (\sqrt{P(x)} - \sqrt{Q(x)})^2}$.

We write $\log$ and $\ln$ for the binary and natural logarithms, respectively, and by $H(X)$ the (Shannon) 
entropy of a discrete random variable $X$ (as well as, by extension, $H(P)$ for the entropy of a discrete distribution $P$). 
We denote by $\mutualinfo{X}{Y}$ the mutual information between two random variables $X$ and $Y$, 
defined as $\mutualinfo{X}{Y} = \sum_{x,y} \probaOf{(X,Y) = (x,y) } \log \frac{\probaOf{(X,Y) = (x,y)}}{\probaOf{X=x}\probaOf{Y=y}}$. 
For a probability distribution $P$, we write $X\sim P$ to indicate that $X$ is distributed according to $P$.
For probability distributions $P, Q$, we will use $P\otimes Q$ to denote the product distribution with marginals $P$ and $Q$.

\paragraph{Identity and Closeness Testing}
We now formally define the testing problems that we study.
 
\begin{definition}[Identity testing]
An \emph{{identity} testing algorithm of distributions belonging to a class $\mathcal{C}$} 
is a randomized algorithm which satisfies the following. Given a parameter $0< \eps <1$ 
and the explicit description of a reference distribution $Q\in \mathcal{C}$, as well as 
access to independent samples from an unknown distribution $P\in \mathcal{C}$, 
the algorithm outputs either $\textsf{accept}$ or $\textsf{reject}$ such that the following holds:
\begin{itemize}
\item(Completeness) if $P=Q$, then the algorithm outputs $\textsf{accept}$ with probability at least $2/3$;
\item(Soundness) if $\normone{P-Q} \geq \eps$, then the algorithm outputs $\textsf{reject}$ with probability at least $2/3$.
\end{itemize}
\end{definition}
Note that by the above definition the algorithm is allowed to answer arbitrarily if neither the completeness nor the soundness cases hold. 
The closeness testing problem is similar, except that now both $P,Q$ are unknown 
and are only available through independent samples.
\begin{definition}[Closeness testing]
A \emph{{closeness} testing algorithm of distributions belonging to a class $\mathcal{C}$} 
is a randomized algorithm which satisfies the following. Given a parameter $0< \eps <1$ 
and access to independent samples from two unknown distributions $P, Q\in \mathcal{C}$, 
the algorithm outputs either $\textsf{accept}$ or $\textsf{reject}$ such that the following holds:
\begin{itemize}
\item(Completeness) if $P=Q$, then the algorithm outputs $\textsf{accept}$ with probability at least $2/3$;
\item(Soundness) if $\normone{P-Q} \geq \eps$, then the algorithm outputs $\textsf{reject}$ with probability at least $2/3$.
\end{itemize}
\end{definition}

Finally, we also consider a third related question, that of \emph{structure testing}:
\begin{definition}[Structure testing]
{Let $\mathcal{C}$ be a family of Bayes nets.}
A \emph{{structure} testing algorithm of Bayes nets belonging to $\mathcal{C}$} 
is a randomized algorithm which satisfies the following. Given a parameter $0< \eps <1$ 
and the explicit description of a DAG $\structure$, as well as access to independent samples 
from an unknown $P\in \mathcal{C}$, 
the algorithm outputs either $\textsf{accept}$ or $\textsf{reject}$ such that the following holds:
\begin{itemize}
\item(Completeness) if $P$ can be expressed as a Bayes net with structure $\structure$, 
then the algorithm outputs $\textsf{accept}$ with probability at least $2/3$;
\item(Soundness) if $\normone{P-Q} > \eps$ for every $Q\in\mathcal{C}$ with structure $\structure$, 
then the algorithm outputs $\textsf{reject}$ with probability at least $2/3$.
\end{itemize}
\end{definition}

In all cases the two relevant complexity measures are the \em{sample complexity}, i.e., the number of samples drawn by the algorithm,
and the \em {time complexity} of the algorithm. The golden standard is to achieve sample complexity 
that is information-theoretically optimal 
and time-complexity linear in the sample complexity. 

{
In this work, the family $\mathcal{C}$ will correspond to the family of Bayes nets over $\{0,1\}^n$,
where we will impose an upper bound $d$ on the maximum in-degree of each node.
For $d=0$, i.e., when the underlying graph has no edges, 
we obtain the family of product distributions over $\{0,1\}^n$.}

\paragraph{Relations between Distances} We will require a number of 
inequalities relating the $L_1$-distance, the KL-divergence, and the Hellinger distance between distributions.
We state a number of inequalities relating these quantities that we will use extensively in our arguments.
The simple proofs are deferred to \cref{sec:prelim:distances:proofs}. 

Recall that a binary product distribution is a distribution over $\{0, 1\}^n$ whose coordinates are independent; and that such a distribution is determined by its mean vector. 
We have the following:

\begin{lemma}\label{lem:prod:kl}
Let $P, Q$ be binary product distributions with mean vectors $p, q \in (0, 1)^n.$
We have that 
\begin{equation}\label{eq:prod:kl}
 2\sum_{i=1}^n (p_i-q_i)^2 \leq \dkl{P}{Q} \leq \sum_{i=1}^n \frac{(p_i-q_i)^2}{q_i(1-q_i)} \;.
\end{equation}
In particular, if there exists $\alpha > 0$ such that $q\in[\alpha,1-\alpha]^n$, we obtain
\begin{equation}\label{eq:prod:kl:balanced}
 2\normtwo{p-q}^2 \leq \dkl{P}{Q} \leq \frac{1}{\alpha(1-\alpha)}\normtwo{p-q}^2 \;.
\end{equation}
\end{lemma}

\noindent Recall that for any pair of distributions $P, Q$, 
Pinsker's inequality states that $\normone{P-Q}^2 \leq  2\dkl{P}{Q}$.
This directly implies the following:

\begin{corollary}\label{lem:prod:dtv:kl}
Let $P, Q$ be binary product distributions with mean vectors $p, q \in (0, 1)^n.$
We have that 
\[
\normone{P-Q}^2 \leq 2\sum_{i=1}^n \frac{(p_i-q_i)^2}{q_i(1-q_i)} \;.
\]
\end{corollary}

\noindent The following lemma states an incomparable and symmetric upper bound on the $L_1$-distance, 
as well as a lower bound.
\begin{lemma}\label{lem:prod:dtv:hellinger}
Let $P, Q$ be binary product distributions with mean vectors $p, q \in (0, 1)^n.$
Then it holds that
\begin{align*}
\min\Big( c, \normtwo{p-q}^4 \Big) &\leq \normone{P-Q}^2 \leq 8 \sum_{i=1}^n \frac{(p_i-q_i)^2}{(p_i+q_i)(2-p_i-q_i)} \;.
\end{align*}
for some absolute constant $c>0$. (Moreover, one can take $c = 4(1-e^{-3/2}) \simeq 3.11$.)
\end{lemma}

\noindent While the above is specific to product distributions, 
we will require analogous inequalities for Bayes nets.
We start with the following simple lemma:

\begin{lemma}\label{lemma:kl:bn}
Let $P$ and $Q$ be Bayes nets with the same dependency graph.
In terms of the conditional probability tables $p$ and $q$ of $P$ and $Q$, we have:
\begin{align*}
2\sum_{k=1}^s \probaDistrOf{P}{ \Pi_k }  (p_k-q_k)^2 &\leq \dkl{P}{Q} \leq \sum_{k=1}^s \probaDistrOf{P}{ \Pi_k }  \frac{(p_k-q_k)^2}{q_k(1-q_k)} \;.
\end{align*}
\end{lemma}
Finally, we state an alternative bound, expressed with respect to the Hellinger distance between two Bayes nets:
\begin{lemma}[{\cite[Lemma 4]{DiakonikolasKS16b}}]\label{lemma:hellinger:bn} 
Let $P$ and $Q$ be Bayes nets with the same dependency graph.
In terms of the conditional probability tables $p$ and $q$ of $P$ and $Q$, we have:
\begin{align*}
&\hellinger{P}{Q}^2 \leq 2\sum_{k=1}^s \sqrt{\probaDistrOf{P}{ \Pi_k }\probaDistrOf{Q}{ \Pi_k }} \frac{(p_k-q_k)^2}{(p_k+q_k)(2-p_k-q_k)} \;.
\end{align*}
\end{lemma}

%%%%%%%%%%%%%%%%%%%%%%%%%%%%%%%%%%%%%%%%%%%%%%%%%%%%%%%%%%%%%%%%%%%%%%%%%%%%%%%%%%%%%%%%%%%%%%%%%%%%%%%%%%%
%%%%%%%%%%%%%%%%%%%%%%%%%%%%%%%%%%%%%%%%%%%%%%%%%%%%%%%%%%%%%%%%%%%%%%%%%%%%%%%%%%%%%%%%%%%%%%%%%%%%%%%%%%%

\section{Testing Identity of Product Distributions} \label{sec:product:identity}

Before considering the general case of testing properties of Bayesian nets, we first provide in this section a sample-optimal efficient algorithm and a matching information-theoretic lower bound for the related, but simpler question of testing identity of \emph{product distributions} over $\{0, 1\}^n$ (the next section will focus on the harder problem of closeness testing, again for product distributions). Our results for this setting can be viewed as discrete analogues
of testing identity and closeness of high-dimensional spherical Gaussians,
that have been studied in the statistics literature~\cite{Hotelling1931, BaiS96, Sriv08, chen2010}.
We note that the Gaussian setting is simpler since the total variation distance can be bounded
by the Euclidean distance between the mean vectors, instead of the chi-squared distance.

The structure of this section is as follows: In~\cref{ssec:product-upper}, we give an identity testing
algorithm for $n$-dimensional binary product distributions with sample complexity $O(\sqrt{n}/\eps^2)$. 
In~\cref{ssec:product-lower}, we show that this sample bound is information-theoretically optimal.

\subsection{Identity Testing Algorithm}  \label{ssec:product-upper}

As mentioned above, here we are concerned with the problem of testing the identity of an unknown product $P$
with mean vector $p$ against an explicit product distribution $Q$ with mean vector $q$.
Our tester relies on a statistic providing an unbiased estimator of $\sum_i
(p_i-q_i)^2/(q_i(1-q_i))$. Essentially, every draw from $P$ gives us an
independent sample from each of the coordinate random variables.
In order to relate our tester more easily to the analogous testers for
unstructured distributions over finite domains, we consider $\Poi(m)$
samples from each of these coordinate distributions. From there, we
construct a random variable $Z$ that provides an unbiased estimator of
our chi-squared statistic, and a careful analysis of the variance of $Z$
shows that with $O(\sqrt{n}/\eps^2)$ samples we can
distinguish between $P=Q$ and $P$ being $\eps$-far from $Q$; leading to the following theorem:

\begin{restatable}{theorem}{identityproductub}\label{theo:identity:product:ub}
There exists a  computationally efficient algorithm\footnote{Throughout this paper, we say an algorithm is computationally efficient if its running time is polynomial in the number of samples and the relevant parameters of the problem (i.e., $n$ and $\eps$).}{} which, given an {explicit} product distribution $Q$ {(via its mean vector)}, 
and sample access to an unknown product distribution $P$ over $\{0,1\}^n$, has the following guarantees: 
For any $\eps >0$, the algorithm takes $\bigO{\sqrt{n}/\eps^2}$ samples from $P$, 
and distinguishes with probability $2/3$ 
between the cases that $P=Q$ versus $\normone{P-Q} > \eps$.
\end{restatable}
\begin{proof} %% Proof of Theorem (Product, identity)
Let $Q=Q_1\otimes\cdots\otimes Q_n$ be a known product distribution over $\{0,1\}^n$ with {mean vector $q$}, 
and $P=P_1\otimes\cdots\otimes P_n$ be an unknown product distribution on $\{0,1\}^n$ {with unknown mean vector $p$}. 
The goal is to distinguish, given independent samples from $P$, between $P=Q$, and $\normone{P-Q} > \eps$.

{Let $0 < \gamma < 1/2$. We say that a product distribution $P$ over $\{0,1\}^n$ is \emph{$\gamma$-balanced} if its mean vector 
$p$ satisfies  $p_i \in [\gamma, 1-\gamma]$ for all $i \in [n]$.} To prove~\cref{theo:identity:product:ub}, 
we can assume without loss of generality that $P, Q$ are $\gamma_0$-balanced for $\gamma_0 \eqdef \frac{\eps}{16n}$. 
Indeed, given sample access to a product distribution $P$, we can simulate access to the $\gamma_0$-balanced product distribution $P'$ 
by re-randomizing independently each coordinate with probability $2\gamma_0$, 
choosing it then to be uniform in $\{0,1\}$. That is, from each draw from $P$ (and the algorithm's own randomness) we can simulate a draw from $P' = P'_1\otimes\cdots\otimes P'_n$, with
\[
    P'_i(0) = (1-2\gamma_0)P_i(0) + \gamma_0\,,\quad P'_i(1) = (1-2\gamma_0)P_i(1) + \gamma_0
\]
for each $i\in[n]$. 
The resulting product distribution $P'$ is $\gamma_0$-balanced, 
and satisfies $\normone{P-P'} \leq n\cdot 2\gamma_0 = \frac{\eps}{8}$. 
Therefore, to test the identity of a product distribution $P$ against a product distribution $Q$ with parameter $\eps$, 
it is sufficient to test the identity of the $\gamma_0$-balanced product distributions $P', Q'$ (with parameter $\frac{\eps}{2}$).

\paragraph{Preprocessing} {We also note that by flipping the coordinates $i$ such that $q_i > 1/2$, 
we can assume that  $q_i \in [\gamma_0, 1/2]$ for all $i \in [n]$.}
This can be done without loss of generality, as $q$ is explicitly given. 
For any $i$ such that $q_i>\frac{1}{2}$, we replace $q_i$ by $1-q_i$ 
and work with the corresponding distribution $Q^\prime$ instead. 
By flipping the $i$-th bit of all samples we receive from $P$, 
it only remains to test identity of the resulting distribution 
$P^\prime$ to $Q^\prime$, as all distances are preserved.

\paragraph{Proof of Correctness} 
Let $m\geq {2716}\frac{\sqrt{n}}{\eps^2}$, and let $M_1,\dots,M_n$ be i.i.d. $\poisson{m}$ random variables. 
We set $M=\max_{i\in[n]} M_i$ and note that $M \leq 2m$ with probability $1-e^{-\Omega(m)}$ (by a union bound). 
We condition hereafter on $M \leq 2m$ (our tester will reject otherwise) 
and take $M$ samples $X^{(1)},\dots, X^{(M)}$ drawn from $P$. We define the following statistic:
\[
    W = \sum_{i=1}^n \frac{(W_i - mq_i)^2- W_i}{q_i(1-q_i)} \;,
\]
where we write $W_i \eqdef \sum_{j=1}^{M_i} X^{(j)}_i$ for all $i\in[n]$. 
We note that the $W_i$'s are independent, as $P$ is a product distribution and the $M_i$'s are independent.
{The pseudocode for our algorithm is given in Figure~\ref{algo:identity:product}.}
Our identity tester is reminiscent of the ``chi-squared type'' testers that have been designed for the unstructured 
univariate discrete setting~\cite{CDVV14, DKN:15, ADK15}.

\begin{figure}
  \begin{framed}
    \begin{description}
      \item[Input] Error tolerance $\eps$, dimension $n$, balancedness parameter $\gamma \geq {\frac{\eps}{16n}}$, {mean vector} 
      $q = (q_1,\dots,q_n)\in [\gamma, 1/2]^n$ of an explicit product distribution $Q$ over $\{0,1\}^n$, 
      and sampling access to a product distribution $P$ over $\{0,1\}^n$.
      
      \item[-] Set $\tau \gets \frac{1}{4} \eps^2$, $m \gets \left\lceil \frac{{2716}\sqrt{n}}{\eps^2}\right\rceil$.
      \item[-] Draw $M_1,\dots,M_n\sim\poisson{m}$ independently, and let $M\gets \max_{i\in[n]} M_i$.
      \item[If] $M > 2m$ set
        $
              W = \tau m^2
          $
      \item[Else] Take $M$ samples $X^{(1)},\dots,X^{(M)}$ from $P$. For $i\in[n]$, let $W_i \gets \sum_{j=1}^{M_i} X^{(j)}_i$, and define
          \[
              W = \sum_{i=1}^n \frac{(W_i - mq_i)^2- W_i}{q_i(1-q_i)}\,.
          \]
      \item[If] $W \geq \tau m^2$ return $\textsf{reject}$.
      \item[Otherwise] return $\textsf{accept}$.
    \end{description}
  \end{framed}
  \caption{Identity testing: unknown product distribution $P$ against given product distribution~$Q$.}
  \label{algo:identity:product}
\end{figure}
We start with a simple formula for the expected value of our statistic:
\begin{lemma}
  $\expect{W} = m^2 \sum_{i=1}^n\frac{(p_i-q_i)^2}{q_i(1-q_i)}$.
\end{lemma}
\begin{proof}
Since $W_i \sim \poisson{mp_i}$ for all $i$, we can write
\begin{align*}
  \expect{(W_i - mq_i)^2} &= \expect{W_i^2}-2mq_i\expect{W_i} + m^2q_i^2
      %= \left( mp_i + m^2p_i^2 \right) - 2m^2p_iq_i + m^2q_i^2 
      = mp_i + m^2(p_i - q_i)^2 \;,
\end{align*}
and therefore
\[
  \expect{W} = \sum_{i=1}^n \frac{\expect{(W_i-mq_i)^2} - \expect{W_i}}{q_i(1-q_i)}
      = m^2\sum_{i=1}^n\frac{(p_i-q_i)^2}{q_i(1-q_i)}\,.
\]
\end{proof}
As a corollary we obtain:
\begin{claim}\label{claim:toy:2:distance:means}
If $P=Q$ then $\expect{W}=0$. Moreover, whenever $\normone{P-Q} > \eps$ 
we have $\expect{W} > \frac{1}{2}m^2\eps^2$.
\end{claim}
\begin{proof}
The first part is immediate from the expression of $\expect{W}$. 
The second follows from~\cref{lem:prod:dtv:kl}, as $m^2\normone{P-Q}^2 \leq 2m^2\sum_{i=1}^n \frac{(p_i-q_i)^2}{q_i(1-q_i)} = 2\expect{W}$.
\end{proof}
We now proceed to bound from above the variance of our statistic. The completeness case is quite simple:
\begin{claim}\label{claim:toy:2:variance:completeness}
  If $P=Q$, then $\Var[W] \leq 8m^2n$.
\end{claim}
\begin{proof}
Suppose that $P=Q$, i.e., $p=q$. 
From independence, we have that $\Var[W] = \sum_{i=1}^n \frac{\Var[(W_i-mq_i)^2 - W_i]}{q_i^2(1-q_i)^2}$. 
Using the fact that $\expect{(W_i-mq_i)^2 - W_i} = 0$, we get
$\Var[(W_i-mq_i)^2 - W_i] = \expect{((W_i-mq_i)^2 - W_i)^2} = 2 m^2 q_i^2$,  
where the last equality follows from standard computations involving the moments of a Poisson random variable. 
From there, recalling that $q_i \in (0,1/2]$ for all $i\in[n]$, we obtain
$\Var[W] = 2m^2\sum_{i=1}^n \frac{1}{(1-q_i)^2} \leq 8m^2n.$
\end{proof}
For the soundness case, the following lemma bounds the variance of our statistic from above. We note that the upper bound depends
on the balancedness parameter $\gamma$.
\begin{lemma}\label{claim:toy:2:variance:soundness}
We have that  $\Var[W] \leq 16nm^2  + \left(\frac{32}{\gamma} + 16\sqrt{2n}m\right) \expect{W} + \frac{32}{\sqrt{\gamma}}\expect{W}^{3/2}$.
% \leq 16nm^2  + 32\sqrt{2n}m\expect{W} + 128\sqrt{\frac{n}{\eps}}\expect{W}^{3/2}$.
\end{lemma}
\begin{proof}
For general $p,q$, we have that
\begin{align*}
  \Var[(W_i-mq_i)^2 - W_i] 
    &= \expect{((W_i-mq_i)^2 - W_i)^2} - m^4(p_i-q_i)^4 \\
    &= 2 m^2 p_i^2 + 4 m^3 p_i (p_i-q_i)^2 \;,
\end{align*}
where as before the last equality follows from standard computations involving the moments of a Poisson random variable. 
This leads to
\begin{align*}
  \Var[W] 
  &= 2m^2\sum_{i=1}^n \frac{p_i^2}{q_i^2(1-q_i)^2} + 4m^3\sum_{i=1}^n \frac{p_i (p_i-q_i)^2}{q_i^2(1-q_i)^2} \\
  &\leq 8m^2\sum_{i=1}^n \frac{p_i^2}{q_i^2} + 16m^3\sum_{i=1}^n \frac{p_i (p_i-q_i)^2}{q_i^2}.
\end{align*}
We handle the two terms separately, in a fashion similar to~\cite[Lemma 2]{ADK15}. For the first term, we can write:
\begin{align*}
  \sum_{i=1}^n \frac{p_i^2}{q_i^2} &= \sum_{i=1}^n \frac{(p_i-q_i)^2}{q_i^2} + \sum_{i=1}^n \frac{2p_iq_i-q_i^2}{q_i^2}  
  = \sum_{i=1}^n \frac{(p_i-q_i)^2}{q_i^2} + \sum_{i=1}^n \frac{2q_i(p_i - q_i)+q_i^2}{q_i^2} \\
  &= n+\sum_{i=1}^n \frac{(p_i-q_i)^2}{q_i^2} + \sum_{i=1}^n \frac{2(p_i - q_i)}{q_i} 
  = n+\sum_{i=1}^n \frac{(p_i-q_i)^2}{q_i^2} + \sum_{i=1}^n \frac{2(p_i - q_i)}{q_i} \\
  &\operatorname*{\leq}_{\text{(AM-GM)}} n+\sum_{i=1}^n \frac{(p_i-q_i)^2}{q_i^2} + \sum_{i=1}^n \left( 1+\frac{(p_i - q_i)^2}{q_i^2} \right) 
  = 2n+2\sum_{i=1}^n \frac{(p_i-q_i)^2}{q_i^2} \\
  &\leq 2n+\frac{2}{\gamma}\sum_{i=1}^n \frac{(p_i-q_i)^2}{q_i}
  \leq 2n+\frac{4}{m^2\gamma}\expect{W}.
\end{align*}
{We bound the second term from above as follows:}
\begin{align*}
  \sum_{i=1}^n \frac{p_i (p_i-q_i)^2}{q_i^2}  
  &\leq \sum_{i=1}^n \frac{p_i}{q_i}\cdot \frac{p_i (p_i-q_i)^2}{q_i}   \\
  &\leq \sqrt{ \sum_{i=1}^n \frac{p_i^2}{q_i^2} }\sqrt{ \sum_{i=1}^n \frac{(p_i-q_i)^4}{q_i^2} } \tag{Cauchy--Schwarz} \\
  &\leq \left( \sqrt{2n}+\frac{2}{m\sqrt{\gamma}}\sqrt{\expect{W}} \right){ \sum_{i=1}^n \frac{(p_i-q_i)^2}{q_i} } \tag{monotonicity of $\ell_p$-norms} \\
  &= \frac{1}{m^2}\left( \sqrt{2n}+\frac{2}{m\sqrt{\gamma}}\sqrt{\expect{W}} \right)\cdot\expect{W}.
\end{align*}
%in a similar fashion as in~\cite[Lemma 2]{ADK15}
Overall, we obtain
\begin{align*}
  \Var[W] &\leq 16nm^2  + \frac{32}{\gamma} \expect{W}  + 16m \left( \sqrt{2n}+\frac{2}{m\sqrt{\gamma}}\sqrt{\expect{W}} \right)\cdot\expect{W} \\
  &= 16nm^2  + \left(\frac{32}{\gamma} + 16\sqrt{2n}m\right) \expect{W} + \frac{32}{\sqrt{\gamma}}\expect{W}^{3/2}.
\end{align*}
\end{proof}
\noindent We are now ready to prove correctness.
\begin{lemma}\label{lemma:product:correctness}
Set $\tau \eqdef \frac{\eps^2}{4}$. 
Then we have the following:
  \begin{itemize}
    \item If $\normone{P-Q} = 0$, then $\probaOf{ W \geq \tau m^2 } \leq \frac{1}{3}$.
    \item If $\normone{P-Q} > \eps$, then $\probaOf{ W < \tau m^2 } \leq \frac{1}{3}$.
  \end{itemize}
\end{lemma}
\begin{proof}
We start with the soundness case, i.e., assuming $\normone{P-Q} > \eps$.
In this case, \cref{claim:toy:2:distance:means} implies $\expect{W} > 2\tau m^2$. 
Since $\gamma \geq \frac{\eps}{16n}$ and for $m \geq \frac{16}{\eps}\sqrt{2n}$, Lemma~\ref{claim:toy:2:variance:soundness} implies that
\[
\Var[W] \leq 16nm^2  + 32\sqrt{2n}m\expect{W} + 32\cdot 4\sqrt{\frac{n}{\eps}}\expect{W}^{3/2}.
\]
By Chebyshev's inequality, we have that
\begin{align*}
  \probaOf{ W < \tau m^2 } &\leq \probaOf{ \expect{W} - W > \frac{1}{2}\expect{W} } 
  \leq \frac{4\Var[W]}{\expect{W}^2} \\
  &\leq \frac{64nm^2}{\expect{W}^2}  + \frac{128\sqrt{2n}m}{\expect{W}^{\vphantom{1}}} 
  + \frac{{4\cdot 128\sqrt{n/\eps}}}{\expect{W}^{1/2}}  \\
  &\leq 
  \frac{4\cdot 64n}{m^2\eps^4} 
  + \frac{2\cdot 128\sqrt{2n}}{m\eps^2} 
  + \frac{{4}\sqrt{2}\cdot 128\sqrt{n}}{m\eps^{{3/2}}} \\
  &\leq 128\left(\frac{2}{C^2} + \frac{5\sqrt{2}}{C} \right) \;,
\end{align*}
which is at most $1/3$ as long as $C\geq 2716$, that is $m \geq 2716\frac{\sqrt{n}}{\eps^2}$.

Turning to the completeness, we suppose $\normone{P-Q} = 0$. Then, again by Chebyshev's inequality 
and~\cref{claim:toy:2:variance:completeness} we have that
\begin{align*}
  \probaOf{ W \geq \tau m^2 } &= \probaOf{ W \geq \expect{W} + \tau m^2 }
  \leq \frac{\Var[W]}{\tau^2 m^4} 
  \leq \frac{128n}{\eps^4 m^2} \;,
\end{align*}
which is no more than $1/3$ as long as $m \geq 8\sqrt{6}\frac{\sqrt{n}}{\eps^2}$.
\end{proof}
\end{proof} %% Proof of Theorem (Product, identity)
\begin{remark}\label{remark:tolerance:tradeoff}
\emph{
We observe that the aforementioned analysis~--~specifically~\cref{claim:toy:2:distance:means} and~\cref{lemma:product:correctness}~--~can be adapted to provide some tolerance guarantees in the completeness case, that is it implies a tester that distinguishes 
between $\normone{P-Q} \leq \eps'$ and $\normone{P-Q} > \eps$, where $\eps' = O(\eps^2)$. 
This extension, however, requires the assumption that $Q$ be balanced: indeed, 
the exact dependence between $\eps'$ and $\eps^2$ will depend on this balancedness parameter, 
leading to a tradeoff between tolerance and balancedness. 
Further, as shown in~\cref{ssec:product-lower:tolerance}, this tradeoff is in fact necessary, 
as tolerant testing of arbitrary product distributions requires $\Omega(n/\log n)$ samples.}
\end{remark}

%%%%%%%%%%%%%%%%%%%%%%%%%%%%%%%%%%%%%%%%%%%%%%%%%%%%%%%%%%%%%%%%%%%%%%%%%%%%%%%%%%

\subsection{Sample Complexity Lower Bound for Identity Testing} \label{ssec:product-lower}

In this section, we prove our matching information-theoretic lower bound for identity testing.
In~\cref{theo:lb:product:uniform}, we give a lower bound for uniformity testing
of a product distribution, while~\cref{theo:lb:product:identity:unbalanced} shows a quantitatively similar lower bound for identity testing
against the product distribution with mean vector $q = (1/n, \ldots, 1/n)$. 
To establish these lower bounds, we use the information-theoretic technique from~\cite{DK:16}:
Given a  candidate hard instance, we proceed by bounding from above
the mutual information between appropriate random variables.
More specifically, we construct an appropriate family of
hard instances (distributions)
and show that  a set of $k$ samples taken from a distribution
belonging to the chosen family has small shared
information with whether or not the distributions are the same.

\begin{restatable}{theorem}{uniformityproductlb} \label{theo:lb:product:uniform}
There exists an absolute constant $\eps_0 > 0$ such that, for any $0 < \eps \leq \eps_0$, the following holds:
Any algorithm that has sample access to an unknown product distribution $P$ over $\{0,1\}^n$
and distinguishes between the cases that $P=U$ and $\normone{P-U} > \eps$ {with probability $2/3$} 
requires $\Omega(\sqrt{n}/\eps^2)$ samples.
\end{restatable}

%\uniformityproductlb*

\begin{proof}
As previously mentioned, we first define two distributions over product distributions $\dyes,\dno$:
\begin{itemize}
  \item $\dyes$ is the distribution that puts probability mass $1$ on the uniform distribution, $U=\bernoulli{1/2}^{\otimes n}$;
  \item $\dno$ is the uniform distribution over the set
    \[
        \setOfSuchThat{ \bigotimes_{j=1}^n \bernoulli{\frac{1}{2} + \frac{(-1)^{b_j}\eps}{\sqrt{n}}}  }{ b\in\{0,1\}^n } \;.
    \]
\end{itemize}

\begin{lemma}\label{lemma:noinstances:far}
  $\dno$ is supported on distributions that are $\Omega(\eps)$-far from $U$.
\end{lemma}
\begin{proof}
The proof, deferred to~\cref{ssec:product-lower:proofs}, proceeds by considering directly the quantity $\normone{P-U}$ in order to obtain a lower bound, where $P\eqdef \bigotimes_{j=1}^n \bernoulli{\frac{1}{2} + \frac{\eps}{\sqrt{n}}}$; specifically, by focusing on the contribution to the distance from the points in the ``middle layers'' of the Boolean hypercube. (We note that an argument relying on the more convenient properties of the Hellinger distance with regard to product distributions, while much simpler, would only give a lower bound of $\Omega(\eps^2)$~--~losing a quadratic factor.)
\end{proof}

We will make a further simplification, namely that instead of drawing $k$ samples from $P=P_1\otimes\dots\otimes P_n$, the algorithm is given $k_i$ samples from each $P_i$, where $k_1,\dots,k_n$ are independent $\poisson{k}$ random variables. This does not affect the lower bound, as this implies a lower bound on algorithms taking $k^\ast\eqdef\max(k_1,\dots, k_n)$ samples from $P$ (where the $k_i$'s are as above), and $k^\ast \geq \frac{k}{2}$ with probability $1-2^{-\Omega(n)}$.
We now consider the following process: letting $X\sim\bernoulli{1/2}$ be a uniformly random bit, we choose a distribution $P$ over $\{0,1\}^n$ by
\begin{itemize}
  \item Drawing $P\sim\dyes$ if $X=0$, and;
  \item Drawing $P\sim\dno$ if $X=1$;
  \item Drawing $k_1,\dots,k_n\sim \poisson{k}$, and returning $k_1$ samples from $P_1$, \dots, $k_n$ samples from $P_n$.
\end{itemize}
For $i\in[n]$, we write $N_i$ for the number of $1$'s among the $k_i$ samples drawn from $P_i$, and let $N=(N_1,\dots, N_n)\in\mathbb{N}^n$. We will invoke this standard fact, as stated in~\cite{DK:16}:
\begin{fact}\label{fact:main:fano}
Let $X$ be a uniform random bit and $Y$ a random variable taking value in some set $\mathcal{S}$. 
If there exists a function $f\colon \mathcal{S}\to \{0,1\}$ such that $\probaOf{ f(Y)=X }\geq 0.51$, then $\mutualinfo{X}{Y} = \Omega(1)$. 
\end{fact}
\begin{proof}
By Fano's inequality, letting $q=\probaOf{ f(Y) \neq X }$, we have $h(q) = h(q)+ q \log (\abs{\{0,1\}}-1) \geq \condentropy{X}{Y}$. This implies
$\mutualinfo{X}{Y} = \entropy{X} - \condentropy{X}{Y} = 1 - \condentropy{X}{Y} \geq 1-h(q) \geq 1 - h(0.49) \geq 2\cdot 10^{-4}$.
\end{proof}
The next step is then to bound from above $\mutualinfo{X}{N}$, in order to conclude that it will be $o(1)$ unless $k$ is taken big enough and invoke~\cref{fact:main:fano}. By the foregoing discussion and the relaxation on the $k_i$'s, we have that the conditioned on $X$ the $N_i$ are independent (with $N_i \sim \poisson{kp_i}$). Recall now that if $X,Y_1,Y_2$ are random variables such that $Y_1$ and $Y_2$ are independent conditioned on $X$, by the chain rule we have that
\begin{align*}
  \condentropy{(Y_1,Y_2)}{X} 
  &= \condentropy{Y_1}{X} + \condentropy{Y_2}{X,Y_1} 
  = \condentropy{Y_1}{X} + \condentropy{Y_2}{X} \;,
\end{align*}
where the second equality follows from conditional independence, and therefore
\begin{align*}
\mutualinfo{X}{(Y_1,Y_2)} 
&= \entropy{(Y_1,Y_2)} - \condentropy{(Y_1,Y_2)}{X} \\
&= \entropy{Y_1} + \condentropy{Y_1}{Y_2} - (\condentropy{Y_1}{X} + \condentropy{Y_2}{X}) \\
&\leq \entropy{Y_1} + \entropy{Y_1} - (\condentropy{Y_1}{X} + \condentropy{Y_2}{X}) \\
&= (\entropy{Y_1}-\condentropy{Y_1}{X}) + (\entropy{Y_2}-\condentropy{Y_2}{X})\\
&= \mutualinfo{X}{Y_1}+\mutualinfo{X}{Y_2}.
\end{align*}
This implies that
\begin{equation}\label{eq:bound:mutual:info:sum}
  \mutualinfo{X}{N} \leq \sum_{i=1}^n \mutualinfo{X}{N_i} \;,
\end{equation}
so that it suffices to bound each $\mutualinfo{X}{N_i}$ separately.

\begin{lemma}\label{lemma:lb:key}
Fix any $i\in [n]$, and let $X, N_i$ be as above. Then $\mutualinfo{X}{N_i} = O( k^2\eps^4/n^2 )$.
\end{lemma}
\begin{proof}
The proof of this rather technical result can be found in~\cref{ssec:product-lower:proofs}, and broadly proceeds as follows. The first step is to upper bound $\mutualinfo{X}{N_i}$ by a more manageable quantity, $\sum_{a=0}^\infty \probaOf{ N_i=a }\left(1-\frac{ \probaCond{N_i=a }{ X=1}  }{ \probaCond{N_i=a }{ X=0} }\right)^2$. After this, giving an upper bound on each summand can be done by performing a Taylor series expansion (in $\eps/\sqrt{n}$), relying on the expression of the moment-generating function of the Poisson distribution to obtain cancellations of many low-order terms.
\end{proof}
This lemma, along with~\cref{eq:bound:mutual:info:sum}, gives the desired result, that is
\begin{equation}
\mutualinfo{X}{N} \leq \sum_{i=1}^n O\!\left( \frac{\eps^4 k^2}{n^2} \right) = O\!\left(\frac{\eps^4 k^2}{n} \right) \;,
\end{equation}
which is $o(1)$ unless $k = \Omega(\sqrt{n}/{\eps^2})$.
\end{proof}
%%%%%%%%%%%%%%%%%%%%%%%%%%%%%%%%%%%%%%%%%%%%%%%%%%%%%%%%%%%%%%%%%%%%%%%%%%%%%%%%
%%%%%%%%%%%%%%%%%%%%%%%%%%%%%%%%%%%%%%%%%%%%%%%%%%%%%%%%%%%%%%%%%%%%%%%%%%%%%%%%
\begin{theorem}\label{theo:lb:product:identity:unbalanced}
There exists an absolute constant $\eps_0 > 0$ such that, for any $\eps \in(0,\eps_0)$, 
distinguishing $P=P^\ast$ and $\normone{P-P^\ast} > \eps$ {with probability $2/3$}
requires $\Omega(\sqrt{n}/\eps^2)$ samples, where $P^\ast \eqdef \bernoulli{1/n}^{\otimes n}$.
\end{theorem}
\begin{proof}
The proof will follow the same outline as that of~\cref{theo:lb:product:uniform} first defining two distributions over product distributions $\dyes,\dno$:
\begin{itemize}
  \item $\dyes$ is the distribution that puts probability mass $1$ on $P^\ast$;
  \item $\dno$ is the uniform distribution over the set
    \[
        \setOfSuchThat{ \bigotimes_{j=1}^n \bernoulli{\frac{1 + (-1)^{b_j}\eps}{n}  }}{ b\in\{0,1\}^n } \;.
    \]
\end{itemize}

\begin{lemma}\label{lemma:noinstances:far:unbalanced}
  With probability $1-2^{-\Omega(n)}$, $\dno$ is supported on distributions $\Omega(\eps)$-far from $P^\ast$.
\end{lemma}
\begin{proof}[Proof of Lemma~\ref{lemma:noinstances:far:unbalanced}]
As for Lemma~\ref{lemma:noinstances:far}, using Hellinger distance as a proxy would only result in an $\Omega(\eps^2)$ lower bound on the distance, 
so we will compute it explicitly instead. The proof can be found in~\cref{ssec:product-lower:proofs}.
\end{proof}
The only ingredient missing to conclude the proof is the analogue of~\cref{lemma:lb:key}:
\begin{lemma}\label{lemma:lb:key:unbalanced}
Suppose $\frac{k\eps^2}{n}\leq 1$. Fix any $i\in [n]$, and let $X, N_i$ be as above. Then $\mutualinfo{X}{N_i} = O( k^2\eps^4/n^2 )$.
\end{lemma}
\begin{proof}
The proof is similar as that of~\cite[Lemma 3.3]{DK:16}, replacing (their) $mn$ by (our) $n$. 
For completeness, we provide an alternative proof in~\cref{sec:misc:proofs}.
\end{proof}
\end{proof}
%%%%%%%%%%%%%%%%%%%%%%%%%%%%%%%%%%%%%%%%%%%%%%%%%%%%%%%%%%%%%%%%%%%%%%%%%%%%%%%%%%%%%%%%%%%%%%%%%%%%%%%%%%%%%%%%%%%%%%%%%%%%%%%%%%%%%%%%%%%%%%%%

%%%%%%%%%%%%%%%%%%%%%%%%%%%%%%%%%%%%%%%%%%%%%%%%%%%%%%%%%%%%%%%%%%%%%%%%%%%%%%%%%%%%%%%%%%%%%%%%%%%%%%%%%%%

\section{Testing Closeness of Product Distributions} \label{sec:product}

Our first set of results involves sample-optimal testers and matching
information-theoretic lower bounds for testing closeness of product distributions
over $\{0, 1\}^n$ (recall that in~\cref{sec:product:identity}, we settled the related question of identity testing of such product distributions). Specifically, in~\cref{sec:closeness:product:ub}, we give a closeness testing
algorithm for $n$-dimensional binary product distributions with sample complexity $O(\sqrt{n}/\eps^2)$; then, we show in~\cref{sec:closeness:product:lb} that this sample bound is information-theoretically optimal.

\subsection{Closeness Testing Algorithm}\label{sec:closeness:product:ub}

Compared to identity, testing closeness between two unknown product distributions is
somewhat more complicated and requires additional ideas. As is the case when comparing unknown
discrete distributions on $[n]$, we have the difficulty that we do not
know how to scale our approximations to the $(p_i-q_i)^2$ terms. We are
forced to end up rescaling using the total number of samples drawn
with $x_i=1$ as a proxy for $1/(q_i)$. This leaves us with a statistic
reminiscent of that used in~\cite{CDVV14}, which can be shown to
work with a related but more subtle analysis. First, in our
setting, it is no longer the case that the sum of the $q_i$'s is
$O(1)$, and this ends up affecting the analysis, making our sample complexity depend on $n^{3/4}$ instead of $n^{2/3}$ as in the unstructured
case. Second, to obtain the optimal sample complexity as a function of both $n$ and $\eps$, we need to 
partition the coordinates into two groups based on the value of their marginals
and apply a different statistic to each group. It turns out that the sample complexity
of our closeness testing algorithm is $O(\max(n^{1/2}/\eps^2, n^{3/4}/\eps))$:

\begin{restatable}{theorem}{closenessproductub}\label{theo:closeness:product:ub}
There exists an efficient algorithm which, given sample access to two unknown 
product distributions $P, Q$ over $\{0,1\}^n$, has the following guarantees. 
For any $\eps\in(0,1)$, the algorithm takes 
$
\bigO{ \max\left(\sqrt{n}/\eps^2, n^{3/4}/{\eps} \right) }
$
 samples from $P$ and $Q$, and distinguishes with probability $2/3$ 
 between (i)~$\normone{P-Q} = 0$ and (ii)~$\normone{P-Q} > \eps$.
\end{restatable}
\noindent The rest of this section is devoted to the proof of the above theorem.
\begin{proof} % Proof of theorem, closeness
 Let $P,Q$ be two product distributions on $\{0,1\}^n$ with mean vectors $p,q\in[0,1]^n$. 
 For $S\subseteq [n]$, we denote by $P_S$ and $Q_S$ the product distributions on $\{0,1\}^{\abs{S}}$ 
 obtained by restricting $P$ and $Q$ to the coordinates in $S$. 
 Similarly, we write $p_S, q_S\in[0,1]^{\abs{S}}$ for the vectors obtained by restricting $p,q$ to the coordinates in $S$, 
 so that $P_S$ has mean vector $p_S$.

\paragraph{High-level Idea} The basic idea of the algorithm is to divide the coordinates in two bins $U,V$: 
one containing the indices where both distributions have marginals very close to $0$ 
(specifically, at most $1/m$, where $m$ is our eventual sample complexity), 
and one containing the remaining indices, on which at least one of the two distributions is roughly balanced. 
Since $P$ and $Q$ can only be far from each other if at least one of $\normone{P_U-Q_U}$, $\normone{P_V-Q_V}$ is big, 
we will test separately each case. Specifically, we will apply two different testers: 
one ``$\chi^2$-based tester'' (with sample complexity $\bigO{\sqrt{n}/\eps^2}$) to the ``heavy bin'' $U$~--~
which relies on the fact that the marginals of $P,Q$ on $U$ are balanced by construction~--~
and one ``$\lp[2]$-tester'' (with sample complexity $\bigO{{n}^{3/4}/\eps}$) to the ``light bin'' $V$~--~
relying on the fact that $\normtwo{p_V}$, $\normtwo{q_V}$ are small.
{The pseudocode of our algorithm is given in Figure~\ref{algo:product:closeness}.}

\paragraph{Sample Complexity} Hereafter, we let
\[
m \eqdef C\max\left(\frac{\sqrt{n}}{\eps^2}, \frac{n^{3/4}}{\eps} \right) \;,
\]
for some absolute constant $C>0$ to be determined in the course of the analysis. 
We let $M_1,\dots,M_n$ and $M'_1,\dots,M'_n$ be i.i.d. $\poisson{m}$ random variables,
 set $M=\max_{i\in[n]} M_i$, $M'=\max_{i\in[n]} M'_i$; and note that $M,M' \leq 2m$ with probability $1-e^{-\Omega(m)}$ (by a union bound). 
We will condition hereafter on the event that $M,M' \leq 2m$ and our tester will reject otherwise.

\noindent Without loss of generality, as in the previous sections, 
we will assume that $\frac{\eps}{16n}\leq p_i,q_i \leq \frac{3}{4}$ for every $i\in[n]$. 
Indeed, this can be ensured by the simple preprocessing step below.

\paragraph{Preprocessing} 
Using $O(\log n)$ samples from $P$ and $Q$, we can ensure without loss of generality that all $p_i,q_i$ 
are at most $3/4$ (with probability $9/10$). 
Namely, we estimate every $p_i, q_i$ to an additive $1/64$, and proceed as follows:
\begin{itemize}
  \item If the estimate of $q_i$ is not within an additive $\pm \frac{1}{32}$ of that of $p_i$, we output $\textsf{reject}$ and stop;
  \item If the estimate of $p_i$ is more than $43/64$, mark $i$ as ``swapped'' and replace $X_i$ by $1-X_i$ (for $P$) 
  and $Y_i$ by $1-Y_i$ (for $Q$) in all future samples.
\end{itemize}
Assuming correctness of the estimates (which holds with probability at least $9/10$), 
if we pass this step then $\abs{p_i - q_i} < \frac{1}{16}$ for all $i$. 
Moreover, if $i$ was not swapped, then it means that we had $p_i \leq 43/64+1/64 < 3/4$, 
and therefore $q_i < 43/64+1/64+1/16 = 3/4$. 
Now, if we had $q_i > 3/4$, then $p_i > 3/4-1/16$ and the estimate of $p_i$ would be more than $3/4-1/16-1/64 = 43/64$.

\begin{figure}[h!]\small
  \begin{framed}
    \begin{description}
      \item[Input] Error tolerance $\eps \in (0,1)$, dimension $n$, and sampling access to two product distributions $P,Q$ over $\{0,1\}^n$.
     
      \item[-] Preprocess $P,Q$ so that $q_{i} \leq \frac{3}{4}$ for all $i\in[n]$, return $\textsf{reject}$ if a discrepancy appears.
      \item[-] Set $m \eqdef C\max\left(\frac{\sqrt{n}}{\eps^2}, \frac{n^{3/4}}{\eps} \right)$.
      \item[-] Define $M, M'$ as follows:
      Draw $M_1,\dots,M_n$, $M'_1,\dots,M'_n$ i.i.d. $\poisson{m}$ random variables, and set $M=\max_{i\in[n]} M_i$, $M'=\max_{i\in[n]} M'_i$.
      \item[-] Take $m$ samples from both $P$ and $Q$, and let $U',V'\subseteq[n]$ be respectively the set of coordinates $i$ such that $X_i=1$ 
      for at least one sample, and its complement.
      \item[If] $\max(M,M') > 2m$, return $\textsf{reject}$.
      \item[-] Take $M$ (resp. $M'$) samples  $X^{(1)},\dots, X^{(M)}$ from $P_{U'}$ (resp. $Y^{(1)},\dots, Y^{(M')}$ from $Q_{U'}$), and define
      \[
          W_{\rm heavy} = \sum_{i\in U'} \frac{ (W_i - V_i)^2- (W_i+V_i) }{W_i+V_i} \;,
      \]
      for $V_i,W_i$ defined as $W_i = \sum_{j=1}^{M_i} X^{(j)}_i$ and $V_i = \sum_{j=1}^{M'_i} Y^{(j)}_i$ for all $i\in U'$.      
      
      \item[If] $W_{\rm heavy} \geq \frac{m\eps^2}{12000}$ return $\textsf{reject}$.
      \item[-] Take $M$ (resp. $M'$) samples $X^{'(1)},\dots, X^{'(M)}$  from $P_{V'}$ (resp. $Y^{'(1)},\dots, Y^{'(M')}$ from $Q_{V'}$), and define
      \[
          W_{\rm light} = \sum_{i\in V'} \left( (W'_i - V'_i)^2- (W'_i+V'_i) \right) \;,
      \]
      for $V'_i,W'_i$ defined as $W'_i = \sum_{j=1}^{M_i} X^{'(j)}_i$, $V'_i = \sum_{j=1}^{M'_i} Y^{'(j)}_i$ for all $i\in V'$.
      \item[If] $W_{\rm light} \geq \frac{\eps^2}{600n}$ return $\textsf{reject}$.
      \item return $\textsf{accept}$.
    \end{description}
  \end{framed}
  \caption{Closeness testing between two unknown product distributions $P,Q$ over $\{0,1\}^n$.}\label{algo:product:closeness}
\end{figure}

\paragraph{Proof of Correctness}
For $m$ as above, define $U,V\subseteq [n]$ by $V \eqdef \setOfSuchThat{ i \in [n]}{ \max(p_i,q_i) < \frac{1}{m} }$ and $U\eqdef [n]\setminus V$.
We start with the following simple claim:

\begin{claim}\label{claim:closeness:twobuckets:UV}
 Assume $\normone{P-Q} > \eps$. Then, at least one of the following must hold: 
 (i) $\normtwo{p_V-q_V}^2 > \frac{\eps^2}{16n}$, or (ii) $\sum_{i\in U} \frac{(p_i-q_i)^2}{p_i+q_i} > \frac{\eps^2}{64}$.
\end{claim}
\begin{proof}
Since $\eps < \normone{P-Q} \leq \normone{P_U-Q_U}+\normone{P_V-Q_V}$, 
at least one of the two terms in the RHS must exceed $\frac{\eps}{2}$. 
We now recall that, by~\cref{lem:prod:dtv:hellinger}, it holds that
$\normone{P_U-Q_U}^2 \leq 8 \sum_{i\in U} \frac{(p_i-q_i)^2}{(p_i+q_i)(2-p_i-q_i)}$ 
and from the further assumption that $p_i,q_i \leq \frac{3}{4}$ that
$\normone{P_U-Q_U}^2 \leq 16 \sum_{i\in U} \frac{(p_i-q_i)^2}{p_i+q_i}$.

\noindent Using subadditivity and the Cauchy--Schwarz inequality, we also have
\begin{align*}
\normone{P_V-Q_V} 
&\leq \sum_{i\in V} \normone{P_{i}-Q_{i}} 
= 2\sum_{i\in V} \abs{p_i-q_i} \\
&= 2\normone{p_V-q_V} 
\leq 2\sqrt{\abs{V}} \normtwo{p_V-q_V} \\
&\leq 2\sqrt{n} \normtwo{p_V-q_V} \;,
\end{align*}
from where we derive that 
$\normtwo{p_V-q_V}^2 \geq \frac{1}{4n}\normone{P_V-Q_V}^2 .$
This completes the proof.
\end{proof}

We now define $U',V'\subseteq [n]$ (our ``proxies'' for $U,V$) as follows: 
Taking $m$ samples from both $P$ and $Q$, 
we let $V'$ be the set of indices which were never seen set to one in any sample, 
and $U'$ be its complement. We have the following:

\begin{claim}\label{claim:closeness:twobuckets:UVprime}
 Assume $\normone{P-Q} > \eps$. Then, at least one of the following must hold: 
 (i)~$\expect{\normtwo{p_{V'}-q_{V'}}^2} > \frac{\eps^2}{150n}$, 
 or (ii)~$\expect{ \sum_{i\in U'\cap U} \frac{(p_i-q_i)^2}{p_i+q_i} } > \frac{\eps^2}{128}$.
\end{claim}
\begin{proof}
By definition, any fixed $i$ belongs to $V'$ with probability $(1-p_i)^m(1-q_i)^m$, and so
\begin{align*}
\expect{\normtwo{p_{V'}-q_{V'}}^2} 
&= \sum_{i=1}^n (p_i-q_i)^2 \cdot (1-p_i)^m(1-q_i)^m 
\geq \sum_{i\in V} (p_i-q_i)^2 \cdot (1-p_i)^m(1-q_i)^m \\
&\geq \mleft(1-\frac{1}{m}\mright)^{2m} \sum_{i\in V} (p_i-q_i)^2 
= \mleft(1-\frac{1}{m}\mright)^{2m}\normtwo{p_V-q_V}^2 \\
&\geq \frac{1}{9} \normtwo{p_V-q_V}^2 \;,
\end{align*}
for $m\geq 10$. Similarly,
\begin{align*}
\shortexpect\Bigg[ \sum_{i\in U'\cap U} \frac{(p_i-q_i)^2}{p_i+q_i} \Bigg]  
&=  \sum_{i\in U} \frac{(p_i-q_i)^2}{p_i+q_i}  \cdot (1-(1-p_i)^m(1-q_i)^m) \\
&\geq \left(1-\left(1-\frac{1}{m}\right)^{2m}\right) \sum_{i\in U} \frac{(p_i-q_i)^2}{p_i+q_i}  \\
&\geq \frac{1}{2} \sum_{i\in U} \frac{(p_i-q_i)^2}{p_i+q_i} \;,
\end{align*}
and in both cases the proof follows by~\cref{claim:closeness:twobuckets:UV}.
\end{proof}

We will require the following implication:
\begin{claim}\label{claim:closeness:twobuckets:UVprime:whp}
Assume $\normone{P-Q} > \eps$. Then, at least one of the following must hold with probability at least $4/5$ (over the choice of $U',V'$): 
(i) $\normtwo{p_{V'}-q_{V'}}^2 > \frac{\eps^2}{300n}$, or (ii) $\sum_{i\in U'\cap U} \frac{(p_i-q_i)^2}{p_i+q_i} > \frac{\eps^2}{2000}$.
\end{claim}
\begin{proof}
First, assume that $\normtwo{p_V-q_V}^2 > \frac{\eps^2}{16n}$, 
and let $V''$ denote the random variable $V'\cap V$. 
By (the proof of)~\cref{claim:closeness:twobuckets:UVprime}, we have 
$\expect{\normtwo{p_{V''}-q_{V''}}^2} \geq \frac{1}{9} \normtwo{p_V-q_V}^2 > \frac{\eps^2}{150n}$. 
Writing $m^2\normtwo{p_{V''}-q_{V''}}^2 = \sum_{i=1}^n m^2(p_i-q_i)^2\indic{i\in V''}$ (note that each summand is in $[0,1]$), we then get by a Chernoff bound that
\[
\probaOf{ \normtwo{p_{V''}-q_{V''}}^2 < \frac{\eps^2}{300n} } < e^{-\frac{1}{8}\frac{m^2 \eps^2}{150n}} < e^{-\frac{C}{1200\eps^2}}\;,
\]
which is less than $1/5$ using our setting of $m$ (for an appropriate choice of the constant $C>0$).
  
Suppose now that $\sum_{i\in U} \frac{(p_i-q_i)^2}{p_i+q_i} > \frac{\eps^2}{64}$. We divide the proof in two cases.
  \begin{itemize}
    \item Case 1: there exists $i^\ast\in U$ such that $\frac{(p_i-q_i)^2}{p_i+q_i} > \frac{\eps^2}{2000}$.
        Then $\probaOf{ \sum_{i\in U'\cap U} \frac{(p_i-q_i)^2}{p_i+q_i} > \frac{\eps^2}{2000} } \geq \probaOf{ i^\ast \in U' } \geq 1-\left(1-\frac{1}{m}\right)^{2m} > \frac{4}{5}$.
    \item Case 2: $\frac{(p_i-q_i)^2}{p_i+q_i} \leq \frac{\eps^2}{2000}$ for all $i\in U$.
      Then, writing $X_i \eqdef \frac{2000}{\eps^2}\frac{(p_i-q_i)^2}{p_i+q_i} \indic{i\in U'\cap U} \in [0,1]$ for all $i\in[n]$, we have
      $\expect{\sum_{i=1}^n X_i} \geq \frac{2000}{128}$ by~\cref{claim:closeness:twobuckets:UVprime}, and a multiplicative Chernoff bound ensures that 
      \begin{align*}
          \probaOf{ \sum_{i\in U'\cap U} \frac{(p_i-q_i)^2}{p_i+q_i} < \frac{\eps^2}{2000} }
          &\leq \probaOf{ \sum_{i=1}^n X_i < 1 } 
          \leq e^{-\frac{2000}{8\cdot 128}} < \frac{1}{5} \;,
      \end{align*}
  \end{itemize}
concluding the proof.
\end{proof}

\noindent Finally, we will need to bound the expected $\lp[2]$-norm of $p_{V'}$ and $q_{V'}$.
\begin{claim}\label{claim:closeness:bound:2norm:vprime}
For $U',V'$ defined as above, we have $\expect{\normtwo{p_{V'}}^2},\expect{\normtwo{q_{V'}}^2} \leq \frac{n}{m^2}$.
\end{claim}
\begin{proof}
By symmetry, it is sufficient to bound $\expect{\normtwo{p_{V'}}^2}$. We have
\begin{align*}
\expect{\normtwo{p_{V'}^2}} &= \sum_{i=1}^n p_i^2 \cdot (1-p_i)^m(1-q_i)^m
\leq \sum_{i=1}^n p_i^2 \cdot (1-p_i)^m.
\end{align*}
Studying the auxiliary function $f\colon x\in[0,1]\mapsto x^2(1-x)^m$, 
we see that it achieves a maximum at $\frac{2}{m+2}$. 
We can then bound
\[
  \expect{\normtwo{p_{V'}^2}} \leq n\cdot f\left(\frac{2}{m+2}\right) \operatorname*{\sim}_{m\to\infty} \frac{4n}{e^2 m^2} \;,
\]
and so $\expect{\normtwo{p_{V'}^2}} \leq \frac{n}{m^2}$ for $m$ large enough 
(and this actually holds for any $m\geq 1$).
\end{proof}

\noindent{In what follows, we analyze our statistics $W_{\rm heavy}$ and $W_{\rm light}$, conditioning on $U',V'$.}
\paragraph*{Case 1: discrepancy in $U'$}
We assume that Algorithm~\ref{algo:product:closeness} reached the line where $W_{\rm heavy}$ is computed, 
and show the following:
\begin{lemma} \label{lem:case1-close}
If $P=Q$, then with probability at least $9/10$ we have $W_{\rm heavy} \leq \frac{m\eps^2}{12000}$. 
Conversely, if $\sum_{i\in U\cap U'} \frac{(p_i-q_i)^2}{p_i+q_i} > \frac{\eps^2}{2000}$, 
then $W_{\rm heavy} \geq \frac{m\eps^2}{12000}$ with probability at least $9/10$.
\end{lemma}
\begin{proof}
Recall that the $W_i$'s are independent, as $P$ is a product distribution and the $M_i$'s are independent. Similarly for the $V_i$'s.
We have:

\begin{claim}\label{claim:product:closeness:expectation}
If $P=Q$, then $\expect{W_{\rm heavy}} = 0$. 
Moreover, if $\sum_{i\in U\cap U'} \frac{(p_i-q_i)^2}{p_i+q_i} > \frac{\eps^2}{2000}$, 
then $\expect{W_{\rm heavy}} > \frac{m\eps^2}{6000}$.
\end{claim}
\begin{proof}
Note that $W_i \sim \poisson{mp_i}$ and $V_i \sim \poisson{mq_i}$ for all $i\in U'$. From there, we can compute (as in~\cite{CDVV14})
\begin{align*}
  \expect{\frac{ (W_i - V_i)^2- (W_i+V_i) }{W_i+V_i}}  =  m\frac{(p_i-q_i)^2}{p_i+q_i}\mleft( 1-\frac{1-e^{-m(p_i+q_i)}}{m(p_i+q_i)} \mright) \;,
\end{align*}
by first conditioning on $W_i+V_i$. 
This immediately gives the first part of the claim. 
As for the second, observing that $1-\frac{1-e^{-x}}{x} \geq \frac{1}{3}\min(1,x)$ for $x\geq 0$, 
and that $p_i+q_i \geq \frac{1}{m}$ for all $i\in U$, by definition we get
\begin{align*}
  \expect{W_{\rm heavy}} 
  &= m\sum_{i\in U'} \frac{(p_i-q_i)^2}{p_i+q_i}\left( 1-\frac{1-e^{-m(p_i+q_i)}}{m(p_i+q_i)} \right) 
  \geq \frac{1}{3} m\sum_{i\in U\cap U'} \frac{(p_i-q_i)^2}{p_i+q_i}
  \geq \frac{m\eps^2}{6000} \;.
\end{align*}
\end{proof}
\noindent We can now bound the variance of our estimator:
\begin{claim}\label{claim:product:closeness:variance}
$\Var[W_{\rm heavy}] \leq 2n + 5m\sum_{i\in U'} \frac{(p_i-q_i)^2}{p_i+q_i} \leq 7n + \frac{3}{5}\expect{W_{\rm heavy}}$. 
In particular, if $P=Q$ then $\Var[W_{\rm heavy}] \leq 2n$.
\end{claim}
\begin{proof}
The proof of the first inequality is similar to that in \cite[Lemma 5]{CDVV14}, 
with a difference in the final bound due to the fact that the $p_i$'s and $q_i$'s 
no longer sum to one. 
For completeness, we give the proof below.

First, note that by independence of the $V_i$'s and $W_i$'s, we have
$\Var[W_{\rm heavy}] = \sum_{i\in U'} \Var\left[ \frac{ (W_i - V_i)^2- (W_i+V_i) }{W_i+V_i} \right]$, 
so it is sufficient to bound each summand individually. 
In order to do so, we split the variance calculation into two parts: 
the variance conditioned on $W_i+V_i=j$, 
and the component of the variance due to the variation in $j$.  
Writing for convenience
\[
f(W_i,V_i)\eqdef \frac{(W_i-V_i)^2-W_i-V_i}{W_i+V_i} \;,
\]
we have that 
\begin{align*}
\Var[f(X,Y)] 
\leq \max_j\left( \Var[f(X,Y) \mid X+Y=j]\right) +\Var[\expect{f(X,Y) \mid X+Y=j}] \;.
\end{align*}

We now bound the first term. 
Since $(W_i - V_i)^2 = (j - 2V_i)^2$, and $V_i$ is distributed as $\binomial{j}{\alpha}$ 
(where for conciseness we let $\alpha\eqdef\frac{q_i}{p_i+q_i}$), 
we can compute the variance of $(j -2V_i)^2$ from standard expressions for the moments 
of the Binomial distribution as 
$
\Var[(j - 2V_i)^2] =16j(j -1)\alpha(1 -\alpha)\left( (j-\frac{3}{2})(1-2\alpha)^2+\frac{1}{2}\right).
$
Since $\alpha(1-\alpha) \leq \frac{1}{4}$ and $j-\frac{3}{2} < j-1 < j$, 
this in turn is at most $j^2(2+4j(1-2\alpha)^2)$.  
Because the denominator is $W_i + V_i$ which equals $j$, 
we must divide this by $j^2$, make it $0$ when $j = 0$, 
and take its expectation as $j$ is distributed as $\operatorname{Poi}(m(p_i + q_i))$. 
This leads to
\begin{align*}
  \Var[f(W_i,V_i) \mid W_i+V_i=j] 
  &\leq 2(1 - e^{-m(p_i+q_i)})  + 4m\frac{(p_i-q_i)^2}{p_i+q_i} \;.
\end{align*}
We now consider the second component of the variance--the contribution to the variance due to the variation 
in the sum $W_i + V_i$. Since for fixed $j$, as noted above, we have $V_i$ distributed as $\binomial{j}{\alpha}$, 
we have
\begin{align*}
  \shortexpect[(W_i - V_i)^2] &= \shortexpect[j^2-4jV_i+4V_i^2]\\
  &=j^2-4j^2\alpha+4(j\alpha-j\alpha^2+j^2\alpha^2)\\
  &=j^2(1-2\alpha)^2+4j\alpha(1-\alpha) \;.
\end{align*}
We finally subtract $W_i+V_i=j$ and divide by $j$ to yield $(j - 1)(1 - 2 \alpha)^2$, 
except with a value of $0$ when $j = 0$ by definition. 
However, note that replacing the value at $j=0$ with $0$ can only lower the variance. 
Since  the sum $j=W_i+V_i$ is drawn from a Poisson distribution with parameter $m(p_i + q_i)$, 
we thus have:
\begin{align*}
  \Var\mleft[ \shortexpect[f(W_i,V_i)|W_i+V_i=j] \mright] &
  \leq m(p_i + q_i)(1 - 2\alpha)^4 
  \leq  m(p_i + q_i)(1 - 2\alpha)^2 
  = m\frac{(p_i-q_i)^2}{p_i+q_i} \;.
\end{align*}
\noindent Summing the final expressions of the previous two paragraphs 
yields a bound on the variance of $f(W_iV_i)$ of
\[
2(1 - e^{-m(p_i+q_i)}) + 5m\frac{(p_i-q_i)^2}{p_i+q_i} \leq 2 + 5m\frac{(p_i-q_i)^2}{p_i+q_i} \;,
\]
as $1 - e^{-x}\leq 1$ for all $x$. This shows that
\begin{align*}
  \Var[W_{\rm heavy}] 
  &\leq 2n + 5m\sum_{i\in U'} \frac{(p_i-q_i)^2}{p_i+q_i} 
  = 2n + 5m\sum_{i\in U'\cap U} \frac{(p_i-q_i)^2}{p_i+q_i} 
  + 5m\sum_{i\in U'\cap V} \frac{(p_i-q_i)^2}{p_i+q_i}\\
  &\leq 2n + \frac{3}{5}\expect{W_{\rm heavy}} + 5m\sum_{i\in U'\cap V} \frac{(p_i-q_i)^2}{p_i+q_i} \;,
\end{align*}
so it only remains to bound the last term. 
But by definition, $i\in V$ implies $0\leq p_i,q_i < \frac{1}{m}$, 
from which
\[
5m\sum_{i\in U'\cap V} \frac{(p_i-q_i)^2}{p_i+q_i}
\leq 5\sum_{i\in U'\cap V} \frac{\abs{p_i-q_i}}{p_i+q_i}
\leq 5\abs{U'\cap V}
\]
which is itself at most $5n$. This completes the proof.
\end{proof}

\noindent With these two claims in hand, we are ready to conclude the proof of Lemma~\ref{lem:case1-close}. 
We start with the soundness case, i.e. assuming $\sum_{i\in U\cap U'} \frac{(p_i-q_i)^2}{p_i+q_i} > \frac{\eps^2}{2000}$. 
Then, by Chebyshev's inequality and~\cref{claim:product:closeness:expectation} we have that
\begin{align}
  \probaOf{ W_{\rm heavy} < \frac{m\eps^2}{12000} } 
  &\leq \probaOf{ \expect{W_{\rm heavy}} - W_{\rm heavy} > \frac{1}{2}\expect{W_{\rm heavy}} } \notag\\
  &\leq \frac{4\Var[W_{\rm heavy}]}{\expect{W_{\rm heavy}}^2} \notag\\
  &\leq \frac{28n}{\expect{W_{\rm heavy}}^2}  + \frac{12}{5\expect{W_{\rm heavy}}} \tag{by \cref{claim:product:closeness:variance}} \\
  &\leq \frac{9\cdot 2000^2\cdot 28n}{m^2\eps^4} + \frac{36\cdot 2000}{5\eps^2 m} \notag \\
  &= \bigO{ \frac{n}{m^2\eps^4}+\frac{1}{\eps^2 m} } \notag\;.
\end{align}
We want to bound this quantity by $1/10$, 
for which it suffices to have $m > C\frac{\sqrt{n}}{\eps^2}$ 
for an appropriate choice of the absolute constant $C>0$ %(e.g. $C = 350$)
in our setting of $m$.

Turning to the completeness, assume that $\normone{P-Q} = 0$. 
Then, by Chebyshev's inequality, and invoking~\cref{claim:product:closeness:variance} we have:
\begin{align*}
  \probaOf{ W \geq \frac{m\eps^2}{12000} } &= \probaOf{ W \geq \expect{W} + \frac{m\eps^2}{12000} }
  \leq \frac{36\cdot 2000^2\Var[W]}{\eps^4 m^2} = \bigO{\frac{n}{\eps^4 m^2}} \;,
\end{align*}
which is no more than $1/10$ for the same choice of $m$. This establishes Lemma~\ref{lem:case1-close}.
\end{proof}

\paragraph*{Case 2: discrepancy in $V'$}
We now assume that Algorithm~\ref{algo:product:closeness} reached the line where $W_{\rm light}$ is computed, 
and show the following:
\begin{lemma}\label{lemma:case2}
If $P=Q$, then with probability at least $9/10$ 
we have $W_{\rm light} \leq \frac{\eps^2}{600n}$. 
Conversely, if $\normtwo{p_{V'}-q_{V'}}^2 > \frac{\eps^2}{300n}$, 
then $W_{\rm light} \geq \frac{\eps^2}{600n}$ with probability at least $9/10$.
\end{lemma}
\begin{proof}
We condition on $\normtwo{p_V'}^2,\normtwo{q_V'}^2 \leq \frac{20n}{m^2}$, 
which by~\cref{claim:closeness:bound:2norm:vprime}, a union bound, and Markov's inequality 
happens with probability at least $19/20$. 
The analysis is similar to~\cite[Section 3]{CDVV14}, 
observing that the $(V'_i)_{i\in V'},(W'_i)_{i\in V'}$'s are mutually independent Poisson random variables, 
$V'_i$ (resp. $W'_i$) having mean $mp_i$ (resp. $mq_i$). 
Namely, following their analysis, the statistic $W_{\rm light}$ 
is an unbiased estimator for $m^2\normtwo{p_{V'}-q_{V'}}^2$ with variance
\[\Var[W_{\rm light}] \leq 8m^3\sqrt{b}\normtwo{p_{V'}-q_{V'}}^2+8m^2b \;,\]
where $b\eqdef \frac{20n}{m^2}$ is our upper bound on $\normtwo{p_V'}^2,\normtwo{q_V'}^2$. 
From there, setting $\eps'\eqdef\frac{\eps}{\sqrt{n}}$ and applying Chebyshev's inequality, 
we get that there exists an absolute constant $C'>0$ 
such the completeness and soundness guarantees from the lemma holds with probability at least $19/20$, 
provided that $m> C'\frac{\sqrt{b}}{{\eps'}^2}$, i.e.,
\[m > C' \frac{n}{\eps^2}\cdot\frac{\sqrt{20n}}{m^2} = \sqrt{20}C'\frac{n^{3/2}}{m\eps^2} \;.\]
Solving for $m$ shows that choosing $m\geq C\frac{n^{3/4}}{\eps}$ 
for some absolute constant $C>0$ is enough. 
A union bound then allows us to conclude the proof of the lemma, 
guaranteeing correctness with probability at least $1-\frac{1}{20}-\frac{1}{20}=\frac{9}{10}$ and concluding the proof of Lemma~\ref{lemma:case2}.
\end{proof}
\noindent This establishes Theorem~\ref{theo:closeness:product:ub}.
 \end{proof} % Proof of theorem, closeness

%%%%%%%%%%%%%%%%%%%%%%%%%%%%%%%%%%%%%%%%%%%%%
\subsection{Sample Complexity Lower Bound for Closeness Testing}\label{sec:closeness:product:lb}
In this section, we prove a matching information-theoretic lower bound 
for testing closeness of two unknown arbitrary product distributions; showing that, perhaps surprisingly, the sample complexity of~\cref{theo:closeness:product:ub} is in fact optimal, up to constant factors:

\begin{restatable}{theorem}{closenessproductlb} \label{theo:lb:product:closeness}
There exists an absolute constant $\eps_0 > 0$ such that, for any $0 < \eps \leq \eps_0$, the following holds:
Any algorithm that has sample access to two unknown \emph{product} distribution $P,Q$ over $\{0,1\}^n$
and distinguishes between the cases that $P=Q$ and $\normone{P-Q} > \eps$ requires $\Omega(\max(\sqrt{n}/\eps^2,n^{3/4}/\eps))$ samples.
\end{restatable}

Before delving into the proof, we give some intuition for the $n^{3/4}$ term of the lower bound lower bound. Recall that the hard family of instances for distinguishing discrete distributions over $[n]$
had (a) many ``light'' bins (domain elements) of probability mass
approximately $1/n$, where either $p_i=q_i$ on each bin
or $p_i = q_i(1 \pm \eps)$ in each bin, and (b) a number of ``heavy'' bins
where $p_i=q_i \approx 1/k$ (where $k$ was the number of
samples taken). The goal of the heavy bins, when designing these hard instances, was to ``add noise''
and hide the signal from the light bins.
In the case of discrete distributions over $[n]$,
we could only have $k$ such heavy bins.
In the case of product distributions,
there is no such restriction,
and we can have $n/2$ of them in our hard instance.
The added noise leads to an increased sample complexity of testing closeness
in the high-dimensional setting.

\begin{proof}
The first part of the lower bound, $\Omega(\sqrt{n}/\eps^2)$, follows from~\cref{theo:lb:product:uniform}; we focus here on the second term, $\Omega(n^{3/4}/\eps)$, {and consequently assume hereafter that $\sqrt{n}/\eps^2 < n^{3/4}/\eps$. Let $k\geq 1$ be fixed, and suppose we have a tester that takes $k=o(n^{3/4}/\eps)$ samples: we will show that it can only be correct with vanishing probability.}  We will again follow the information-theoretic framework of~\cite{DK:16} for proving distribution testing lower bounds, first defining two distributions over pairs of product distributions $\dyes,\dno$:
\begin{itemize}
  \item $\dyes$: for every $i\in[n]$, independently choose $(p_i,q_i)$ to be either $p_i=q_i = \frac{1}{{k}}$ with probability $1/2$, and $p_i=q_i = \frac{1}{n}$ otherwise; and set $P\eqdef \bigotimes_{j=1}^n \bernoulli{p_i}$, $Q\eqdef \bigotimes_{j=1}^n \bernoulli{q_i}$.
  \item $\dno$: for every $i\in[n]$, independently choose $(p_i,q_i)$ to be either $p_i=q_i = \frac{1}{{k}}$ with probability $1/2$, and $(\frac{1+\eps}{n},\frac{1-\eps}{n})$ or $(\frac{1-\eps}{n},\frac{1+\eps}{n})$ uniformly at random otherwise; and set $P\eqdef \bigotimes_{j=1}^n \bernoulli{p_i}$, $Q\eqdef \bigotimes_{j=1}^n \bernoulli{q_i}$.
\end{itemize}
Note that in both $\dyes$ and $\dno$, with overwhelming probability the pairs $(P,Q)$ have roughly $n/2$ marginals with (equal) parameter $1/{k}$, and roughly $n/2$ marginals with parameter $\bigTheta{1}/n$.

\noindent The following lemma is shown similarly to Lemma~\ref{lemma:noinstances:far:unbalanced}:
\begin{lemma}\label{lemma:noinstances:far:again}
  With probability $1-2^{-\Omega(n)}$, a uniformly chosen pair $(P,Q)\sim \dno$ satisfies $\normone{P-Q} = \Omega(\eps)$.
\end{lemma}
We will as before make a further simplification, namely that instead of drawing $k$ samples from $P=P_1\otimes\dots\otimes P_n$ and $Q=Q_1\otimes\dots\otimes Q_n$, the algorithm is given $k_i$ samples from each $P_i$ (resp. $k'_i$ from $Q_i$), where $k_1,\dots,k_n,k'_1,\dots,k'_n$ are independent $\poisson{k}$ random variables. We now consider the following process: letting $X\sim\bernoulli{1/2}$ be a uniformly random bit, we choose a pair of distributions $(P,Q)$ (both $P$ and $Q$ being probability distributions over $\{0,1\}^n$) by
\begin{itemize}
  \item Drawing $(P,Q)\sim\dyes$ if $X=0$, and;
  \item Drawing $(P,Q)\sim\dno$ if $X=1$;
  \item Drawing $k_1,k'_1,\dots,k_n,k'_n\sim \poisson{k}$, and returning $k_i$ samples from each $P_i$ and $k'_i$ samples from each $Q_i$
\end{itemize}

For $i\in[n]$, we let $N_i$ and $M_i$ denote respectively the number of $1$'s among the $k_i$ samples drawn from $P_i$ and $k'_i$ samples drawn from $Q_i$, and write $N=(N_1,\dots, N_n)\in\mathbb{N}^n$ (and $M\in \mathbb{N}^n$ for $Q$). The next step is then to upperbound $\mutualinfo{X}{(N,M)}$, in order to conclude that it will be $o(1)$ unless $k$ is taken big enough and invoke~\cref{fact:main:fano}. By the foregoing discussion and the relaxation on the $k_i$'s, we have that the conditioned on $X$ the $N_i$'s (and $M_i$'s) are independent (with $N_i \sim \poisson{kp_i}$ and $M_i \sim \poisson{kq_i}$). This implies that
\begin{equation}\label{eq:bound:mutual:info:sum:again}
  \mutualinfo{X}{(N,M)} \leq \sum_{i=1}^n \mutualinfo{X}{(N_i,M_i)}
\end{equation}
so that it suffices to bound each $\mutualinfo{X}{(N_i,M_i)}$ separately.

\begin{lemma}\label{lemma:lb:key:again}
Fix any $i\in [n]$, and let $X, N_i,M_i$ be as above. Then $\mutualinfo{X}{(N_i,M_i)} = O( k^4\eps^4/n^4 )$.
\end{lemma}
\begin{proof}
By symmetry it is enough to consider only the case of $i=1$, so that we let $(A,B)=(N_1,M_1)$.\medskip

Since $A\sim \poisson{kp_1}$ and $B\sim \poisson{kq_1}$ with $(p_1,q_1) = (1/{k}, 1/{k})$ or $(p_1,q_1) = (1/n, 1/n)$ uniformly if $X=0$, and 
\[
(p_1,q_1) = \begin{cases}
(\frac{1}{{k}},\frac{1}{{k}}) & \text{ w.p. } \frac{1}{2}\\
(\frac{1+\eps}{n},\frac{1-\eps}{n}) & \text{ w.p. } \frac{1}{4} \\
(\frac{1-\eps}{n},\frac{1+\eps}{n}) & \text{ w.p. } \frac{1}{4}
\end{cases}
\]
if $X=1$, a computation similar as that of~\cite[Proposition 3.8]{DK:16} yields that, for any $i,j\in\mathbb{N}$
\begin{align*}
  \probaCond{(A,B)=(i,j)}{X=0} 
    &= \frac{1}{2i!j!}\left(  e^{-2k/{k}} \left(\frac{k}{{k}}\right)^{i+j} + e^{-\frac{2k}{n}} \left(\frac{k}{n}\right)^{i+j} \right)\\
    &= \frac{1}{2i!j!}\left(  e^{-2} + e^{-\frac{2k}{n}} \left(\frac{k}{n}\right)^{i+j} \right) \\
  \probaCond{(A,B)=(i,j)}{X=1} 
    &= \frac{1}{2i!j!}\left(  e^{-2k/{k}} \left(\frac{k}{{k}}\right)^{i+j} + e^{-\frac{2k}{n}} \left(\frac{k}{n}\right)^{i+j} \frac{(1+\eps)^i (1-\eps)^j + (1-\eps)^i (1+\eps)^j }{2} \right) \\
    &= \frac{1}{2i!j!}\left(  e^{-2}  + e^{-\frac{2k}{n}} \left(\frac{k}{n}\right)^{i+j} \frac{(1+\eps)^i (1-\eps)^j + (1-\eps)^i (1+\eps)^j }{2} \right).
\end{align*}
Note in particular that for $0\leq i+j \leq 1$, this implies that $\probaCond{(A,B)=(i,j)}{X=0}=\probaCond{(A,B)=(i,j)}{X=1}$. From the above, we obtain
\begin{align*}
  \mutualinfo{X}{(A,B)} 
  &= O(1)\sum_{i,j\geq 0} \frac{\left( \probaCond{(A,B)=(i,j)}{X=0}-\probaCond{(A,B)=(i,j)}{X=1} \right)^2}{\probaCond{(A,B)=(i,j)}{X=0}+\probaCond{(A,B)=(i,j)}{X=1} } \\
  &= O(1)\sum_{i+j\geq 2} \frac{\left( \probaCond{(A,B)=(i,j)}{X=0}-\probaCond{(A,B)=(i,j)}{X=1} \right)^2}{\probaCond{(A,B)=(i,j)}{X=0}+\probaCond{(A,B)=(i,j)}{X=1} } \\
  &= O(1)\sum_{i+j\geq 2} e^{-\frac{4k}{n}} \frac{\mleft(\frac{k}{n}\mright)^{2(i+j)}(1-\frac{( (1+\eps)^i (1-\eps)^j + (1-\eps)^i (1+\eps)^j )}{2})^2}{2i!j!(2e^{-2}+o(1))} \\
  &= \bigO{\mleft({k\eps}/{n}\mright)^4}
\end{align*}
where the second-to-last inequality holds for $k=o(n)$. (Which is the case, as $\sqrt{n}/\eps^2 < n^{3/4}/\eps$ implies that $n^{3/4}/\eps < n$, and we assumed $k=o(n^{3/4}/\eps)$.)
\end{proof}
This lemma, along with~\cref{eq:bound:mutual:info:sum:again}, immediately implies the result:
\begin{equation}
\mutualinfo{X}{(N,M)} \leq \sum_{i=1}^n \bigO{\left(\frac{k\eps}{n}\right)^4} = \bigO{\frac{k^4\eps^4}{n^3}}
\end{equation}
which is $o(1)$ unless $k = \Omega(n^{3/4}/{\eps})$.

\end{proof}

\subsection{Ruling Out Tolerant Testing Without Balancedness} \label{ssec:product-lower:tolerance}

{In this section, we show that any \emph{tolerant} identity testing algorithm for product distributions 
must have sample complexity near-linear in $n$ if the explicitly given distribution
is very biased.}

\begin{theorem}\label{theo:tradeoff:balancedness:tolerance}
There exists an absolute constant $\eps_0 < 1$ such that the following holds. 
Any algorithm that, given a parameter $\eps\in(0,\eps_0]$ and sample access to 
product distributions $P,Q$ over $\{0,1\}^n$, distinguishes between $\normone{P-Q} < \eps/2$ 
and $\normone{P-Q} > \eps$ with probability at least $2/3$ 
requires $\bigOmega{n/\log n}$ samples. 
Moreover, the lower bound still holds in the case where $Q$ is known, and provided as an explicit parameter.
\end{theorem}
\begin{proof}
The basic idea will be to reduce to the case of tolerant testing of two arbitrary distributions $p$ and $q$ over $[n]$. 
In order to do this, we define the following function from distributions of one type to distributions of the other:

If $p$ is a distribution over $[n]$, define $F_{\delta}(p)$ to be the distribution over $\{0,1\}^n$ 
obtained by taking $\Poi(\delta)$ samples from $p$ and returning the vector $x$ 
where $x_i = 1$ if and only if $i$ was one of these samples drawn. 
Note that, because of the Poissonization, $F_\delta(p)$ is a product distribution.
We have the following simple claim:

\begin{claim}
For any $\delta\in(0,1]$ and distributions $p,q$ on $[n]$, $\dtv(F_\delta(p),F_\delta(q)) = (\delta+O(\delta^2))\dtv(p,q)$.
\end{claim}
\begin{proof}
In one direction, we can take correlated samples from $F_\delta(p)$ and $F_\delta(q)$ 
by sampling $a$ from $\Poi(\delta)$ and then taking $a$ samples 
from each of $p$ and $q$, using these to generate our samples from $F_\delta(p),F_\delta(q)$. 
For fixed $a$, the variation distance between $F_\delta(p)$ and $F_\delta(q)$ 
conditioned on that value of $a$ is clearly at most $a\dtv(p,q)$. 
Therefore, $\dtv(F_\delta(p),F_\delta(q))\leq \E[a]\dtv(p,q) = \delta \dtv(p,q).$

In the other direction, note that $F_\delta(p)$ and $F_\delta(q)$ each have probability $\delta+O(\delta^2)$ of returning a vector of weight $1$. 
This is because $\Poi(\delta)=1$ with probability $\delta e^{-\delta} = \delta+O(\delta^2)$ 
and since $\Poi(\delta)>1$ with probability $O(\delta^2)$. 
Let $G(p)$ and $G(q)$ denote the distributions $F_\delta(p)$ and $F_\delta(q)$ 
conditioned on returning a vector of weight $1$. By the above, we have that 
$\dtv(F_\delta(p),F_\delta(q)) \geq (\delta+O(\delta^2))\dtv(G(p),G(q))$. 
Letting $p_i$ (resp. $q_i$) be the probability that $p$ (resp. $q$) assigns to $i\in[n]$, 
we get that for any fixed $i\in[n]$ the probability that $F_\delta(p)$ returns $e_i$ is 
\[
(1-e^{-\delta p_i})\prod_{j\neq i} e^{-\delta p_j} = (e^{\delta p_i}-1) \prod_{j=1}^n e^{-\delta p_j} \;.
\]
Therefore $G(p)$ puts on $e_i$ probability proportional to $(e^{\delta p_i}-1) = (\delta+O(\delta^2))p_i$. 
Similarly, the probability that $G(q)$ puts on $e_i$ is proportional to $(\delta+O(\delta^2))q_i$ 
(where in both cases, the constant of proportionality is $(\delta +O(\delta^2))^{-1}$). Therefore,
\begin{align*}
\dtv(G(p),G(q))
&= \delta^{-1}(1+O(\delta))\sum_{i=1}^n |(\delta+O(\delta^2))p_i - (\delta+O(\delta^2))q_i| \\
& = \delta^{-1}(1+O(\delta))\sum_{i=1}^n ( \delta|p_i-q_i|+O(\delta^2)(p_i+q_i) )\\
&= \delta^{-1}(1+O(\delta))(\delta\dtv(p,q)+O(\delta^2))\\
& = \dtv(p,q)+O(\delta) \;.
\end{align*}
Thus, $\dtv(F_\delta(p),F_\delta(q)) \geq (\delta+O(\delta^2))\dtv(p,q).$ This completes the proof.
\end{proof}
The above claim guarantees the existence of some constant $\delta_0\in(0,1]$ such that 
$\dtv(F_{\delta_0}(p),F_{\delta_0}(q)) \in [0.9{\delta_0} \dtv(p,q),1.1\dtv(p,q)].$ 
However, it is known~\cite{ValiantValiant:11} that for any sufficiently small $\eps>0$ 
there exist distributions $p$ and $q$ over $[n]$ such that one must take at least $c\frac{n}{\log n}$ samples 
(where $c>0$ is an absolute constant) to distinguish between 
$\dtv(p,q) \leq \eps/(2\cdot 0.9{\delta_0})$ and $\dtv(p,q) \geq \eps/(1.1{\delta_0})$. 
Given $q$ samples from $p$ and $q$ we can with high probability simulate $c'q$ samples from 
$P=F_{\delta_0}(p)$ and $Q= F_{{\delta_0}}(q)$ (where $c'=c'(\delta_0)>0$ is another absolute constant). 
Therefore, we cannot distinguish between the cases $\dtv(P,Q)\leq \eps/2$ and $\dtv(P,Q)\geq \eps$ in fewer than $c'\cdot c \frac{n}{\log n}$, 
as doing so would enable us to distinguish between $p$ and $q$ with less than $c\frac{n}{\log n}$ samples~--~yielding a contradiction. 
Moreover, the above still holds when $q$ is explicitly known, specifically even when $q$ is taken to be the uniform distribution on $[n]$.
\end{proof}

%%%%%%%%%%%%%%%%%%%%%%%%%%%%%%%%%%%%%%%%%%%%%%%%%%%%%%%%%%%%%%%%%%%%%%%%%%%%%%%%%%%%%%%%%%%%%%%%%%%%%%%%%%%%%%%%%%%%%%%%%%%%%%%%
\section{From Product Distributions to Bayes nets: an Overview}\label{sec:product:to:bn:overview}

In the following paragraphs, we describe how to generalize our previous
results for product distributions to testing general Bayes nets.
The case of known structure turns out to be manageable, and at a technical
level a generalization of our testers for product distributions. The case
of unknown structure poses various complications and requires a number of
non-trivial new ideas.

\paragraph{Testing Identity and Closeness of Fixed Structure Bayes Nets}
Our testers and matching lower bounds for the fixed structure regime
are given in~\cref{sec:identity-known,sec:closeness-known}.

For concreteness, let us consider the case of testing identity
of a tree-structured ($d=1$) Bayes net $P$ against an explicit tree-structured
Bayes net $Q$ with the same structure.
Recall that we are using as a proxy for {the distance % total variation distance
$\normone{P-Q}$} an appropriate chi-squared-like quantity.
A major difficulty in generalizing our identity tester for products is that
the chi-squared statistic depends not on the probabilities of the various coordinates,
but on the conditional probabilities of these coordinates based on all possible parental
configurations. This fact produces a major wrinkle in our analysis for the following reason:
while in the product distribution case each sample provides information
about each coordinate probability, in the Bayes net case a sample
only provides information about conditional probabilities for parental
configurations that actually occurred in that sample.

This issue can be especially problematic to handle if there are uncommon
parental configurations about which we will have difficulty gathering much information
(with a small sized sample). Fortunately, the probabilities conditioned on such parental
configurations will have a correspondingly smaller effect on the final
distribution and thus, we will not need to know them to quite the same
accuracy. So while this issue can be essentially avoided, we will require
some technical assumptions about {\em balancedness} to let us know that none
of the parental configurations are too rare. Using these ideas, we
develop an identity tester for tree-structured Bayes nets that uses an optimal
$\Theta(\sqrt{n}/\eps^2)$ samples. For known structure Bayes nets of degree $d>1$, the sample complexity
will also depend exponentially on the degree $d$.
Specifically, each coordinate will have as many as $2^d$ parental configurations.
Thus, instead of having only $n$ coordinate probabilities to worry about,
we will need to keep track of $2^d n$ conditional probabilities. This will require that our sample
complexity also scale like $2^{d/2}$. The final complexity of our identity and closeness testers
 will thus be $O(2^{d/2}\sqrt{n}/\eps^2)$.

We now briefly comment on our matching lower bounds.
Our sample complexity lower bound of $\Omega(\sqrt{n}/\eps^2)$ for the product
case can be generalized in a black-box manner to yield a tight lower bound
$\Omega(2^{d/2}\sqrt{n}/\eps^2)$ for testing uniformity of degree-$d$ Bayes nets.
The basic idea is to consider degree-$d$ Bayes nets with the following structure:
The first $d$ nodes are all independent (with marginal probability $1/2$ each),
and will form in some sense a ``pointer'' to one of $2^d$ arbitrary product distributions.
The remaining $n-d$ nodes will each depend on all
of the first $d$. The resulting distribution is now an (evenly weighted) disjoint mixture of $2^d$
product distributions on the $(n-d)$-dimensional hypercube.
In other words, there are $2^d$ product distributions $p_1,\dots,p_{2^d}$,
and our distribution returns a random $i$ (encoded in binary) followed by a random sample form $p_i$.
By using the fact that the $p_i$'s can be arbitrary product distributions, we obtain our desired
sample complexity lower bound.

\paragraph{Testing Identity and Closeness of Unknown Structure Bayes Nets}
As we show in~\cref{sec:identity-uknown,sec:closeness-unknown}, this situation changes substantially when we do not know the
underlying structure of the nets involved. In particular, we show that
even for Bayes nets of degree-$1$ uniformity testing requires
$\Omega(n/\eps^2)$ samples.

The lower bound construction for this case is actually quite simple:
The adversarial distribution $P$ will be developed by taking
a {\em random matching} of the vertices and making each
matched pair of vertices randomly $1 \pm \eps/\sqrt{n}$ correlated. If the
matching were known by the algorithm, the testing procedure could
proceed by approximating these $n/2$ correlations. However, not knowing
the structure, our algorithm would be forced to consider all $\binom{n}{2}$
pairwise correlations, substantially increasing the amount of noise
involved. To actually prove this lower bound, we consider the
distribution $X$ obtained by taking $k$ samples from a randomly chosen $P$
and $Y$ from taking $k$ samples from the uniform distribution.
Roughly speaking, we wish to show that $\chi^2(X,Y)$ is approximately $1$.
This amounts to showing that for a randomly chosen pair of distributions
$P$ and $P'$ from this family, we have that
$\E[P^k(x)P'^k(x)]$ is approximately $1$.
Intuitively, we show that this expectation is only large
if $P$ and $P'$ share many edges in common. In
fact, this expectation can be computed exactly in terms of the lengths
of the cycles formed by the graph obtained taking the union of the
edges from $P$ and $P'$. Noting that $P$ and $P'$ typically share only about
$1$ edge, this allows us to prove our desired lower bound.

However, the hardness of the situation described above is not generic
and can be avoided if the explicit distribution $Q$
satisfies some non-degeneracy assumptions.
Morally, a Bayes nets $Q$ is non-degenerate if it is not close in
variational distance to any other Bayes net of no greater
complexity and non-equivalent underlying structure. For tree
structures, our condition is that for each node the two conditional probabilities
for that node (depending on the value of its parent) are far from each other.

If this is the case, even knowing approximately what the pairwise distributions of coordinates are 
will suffice to determine the structure. One way to see this is the following: the analysis of the Chow-Liu algorithm~\cite{Chow68} 
shows that the tree-structure for $P$ is the maximum spanning tree of the graph whose edge weights 
are given by the shared information of the nodes involved. This tree will have the property that each edge, $e$, 
has higher weight than any other edge connecting the two halves of the tree. 
We show that our non-degeneracy assumption implies that this edge has higher weight by a noticeable margin, 
and thus that it is possible to verify that we have the correct tree with only rough approximations 
to the pairwise shared information of variables.

For Bayes nets of higher degree, the analysis is somewhat more difficult. We need a slightly more complicated notion of
non-degeneracy, essentially boiling down to a sizeable number of
not-approximately-conditionally-independent assumptions. 
For example, a pair of nodes can be positively identified as having an edge between them in the underlying graph 
if they are not conditionally independent upon any set of $d$ other nodes. By requiring that for each edge 
the relevant coordinate variables are not close to being conditionally independent, 
we can verify the identity of the edges of $\structure$ with relatively few samples. 
Unfortunately, this is not quite enough, as with higher degree Bayesian networks, 
simply knowing the underlying undirected graph 
is not sufficient to determine its structure. 
We must also be able to correctly identify the so-called $\lor$-structures. 
To do this, we will need to impose more not-close-to-conditionally-independent 
assumptions that allow us to robustly determine these as well.

Assuming that $Q$ satisfies such a non-degeneracy condition, testing
identity to it is actually quite easy. First one verifies that the
distribution $P$ has all of its pairwise (or $(d+2)$-wise) probabilities close to the
corresponding probabilities for $Q$. By non-degeneracy, this will imply
that $P$ must have the same (or at least an equivalent) structure as $Q$. Once this has been
established, the testing algorithms for the known structure 
can be employed.

\paragraph{Sample Complexity of Testing High-Degree Bayes Nets}
One further direction of research is that of
understanding the dependence on degree of the sample complexity
of testing identity and closeness for degree-$d$ Bayes nets without additional assumptions.
For $d=1$, we showed that these problems can be as hard as learning the distribution.
For the general case, we give an algorithm with sample complexity
$2^{d/2}\poly(n,1/\eps)$ for identity testing (and $2^{2d/3}\poly(n,1/\eps)$ for closeness testing). 
The conceptual message of this result is that, when the
degree increases, testing becomes easier than learning information-theoretically.
It is a plausible conjecture that the correct
answer for identity testing is $\Theta(2^{d/2}n/\eps^2)$ 
and for closeness testing is  $\Theta(2^{2d/3}n/\eps^2)$.
We suspect that our lower bound techniques can be
generalized to match these quantities, but the constructions 
will likely be substantially more intricate.

The basic idea of our $2^{d/2}\poly(n,1/\eps)$ sample upper bound for identity testing
is this: We enumerate over all possible structures for $P$, running a different tester for
each of them by comparing the relevant conditional probabilities.
Unfortunately, in this domain, our simple formula for the
KL-Divergence between the two distributions will no longer hold.
However, we can show that using the old formula will be sufficient by
showing that if there are large discrepancies when computing the KL-divergence,
then there must be large gap between the entropies $H(P)$ and $H(Q)$
in a particular direction. 
As the gap cannot exist both ways, this
suffices for our purposes.

%%%%%%%%%%%%%%%%%%%%%%%%%%%%%%%%%%%%%%%%%%%%%%%%%%%%%%%%%%%%%%%%%%%%%%%%%%%%%%%%%%%%%%%%%%%%%%%%%%%%%%%%%%%%%%%%%%%%%%%%%%%%%%%%

 \section{Testing Identity of Fixed Structure Bayes Nets} \label{sec:identity-known}
In this section, we prove our matching upper and lower bounds for testing the identity of Bayes nets with known graph structure.
In~\cref{ssec:identity-known-upper}, we describe an identity testing algorithm that uses $\bigO{2^{d/2}\sqrt{n}/\eps^2}$
samples, where $d$ is the maximum in-degree and $n$ the number of nodes (dimension).
In~\cref{ssec:identity-known-lower}, we show that this sample upper bound is tight, up to constant factors,
even for uniformity testing.

\subsection{Identity Testing Algorithm}  \label{ssec:identity-known-upper}

In this section, we establish the upper bound part of~\cref{thm:informal-identity-closeness-known} for identity, namely testing identity to a fixed Bayes net given sample access to an unknown Bayes net with the same underlying structure. In order to state our results, we recall the definition of \emph{balancedness} of a Bayes net:
\begin{definition}
A Bayes net $P$ over $\{0,1\}^n$ with structure $\structure$ is said to be \emph{$(c,C)$-balanced} if, 
for all $k$, it is the case that (i) $p_k\in[c,1-c]$ and (ii) $\probaDistrOf{P}{\Pi_k} \geq C$.
\end{definition}
\noindent Roughly speaking, the above conditions ensure that the conditional probabilities of the Bayes net 
are bounded away from $0$ and $1$, and that each parental configuration 
occurs with some minimum probability. 
With this definition in hand, we are ready to state and prove the main theorem of this section:

\begin{restatable}{theorem}{identityknowndegreedub}\label{theo:upper:knowndegreed:identity}
There exists a computationally efficient algorithm with the following guarantees. 
Given as input (i) a DAG $\structure$ with $n$ nodes and maximum in-degree $d$ 
and a known $(c,C)$-balanced Bayes net $Q$ with structure $\structure$, 
where $c=\tildeOmega{1/\sqrt{n}}$ and $C=\tildeOmega{d\eps^2/\sqrt{n}}$; 
(ii) a parameter $\eps > 0$, and (iii) sample access to an unknown Bayes net $P$ with structure $\structure$, 
the algorithm takes $\bigO{2^{d/2}\sqrt{n}/\eps^2}$ samples from $P$, 
and distinguishes with probability at least $2/3$ between the cases $P=Q$ and $\normone{P-Q} > \eps$.
\end{restatable}

We choose $m\geq \alpha\frac{2^{d/2}\sqrt{n}}{\eps^2}$, 
where $\alpha>0$ is an absolute constant to be determined in the course of the analysis. 
Let $\structure$ and $Q$ be as in the statement of the theorem, for $c\geq \beta\frac{\log n}{ \sqrt{n} } \geq \beta\frac{\log n}{m}$ and $C\geq \beta\frac{d+\log n}{m}$, 
for an appropriate absolute constant $\beta>0$.

{Recall that $S$ denotes the set $\{(i,a): i \in [n], a\in \{0,1\}^{{|}\parent{i}{|}}\}$. By assumption, 
we have that  $|\parent{i}| \leq d$ for all $i \in [n]$.}
For each $(i,a)\in S$, corresponding to the parental configuration $\Pi_{i,a}=\{ X_{\parent{i}} = a\}$, 
we define the value $N_{i,a} \eqdef m \probaDistrOf{Q}{\Pi_{i,a}}/\sqrt{2}$. {Intuitively,
$N_{i,a}$ is equal to a small constant factor times} the number of samples satisfying $\Pi_{i,a}$ one 
would expect to see among $m$ independent samples, 
if the unknown distribution $P$ were equal to $Q$.  
We will also use the notation $p_{i,a} \eqdef \probaCond{X_i=1}{X_{\parent{i}=a}}$, 
where $X\sim P$, and $q_{i,a} \eqdef \probaCond{X_i=1}{X_{\parent{i}=a}}$, where $X\sim Q$.

Given $m$ independent samples $X^{(1)},\dots, X^{(m)}$ from a Bayes net $P$ with structure $\structure$, 
we define the estimators $Z_{i,a},Y_{i,a}$ for every $i\in[n]$, $a\in\{0,1\}^{{|\parent{i}|}}$ as follows. 
For every $(i,a)$ such that the number of samples $X^{(j)}$ satisfying the configuration $\Pi_{i,a}$ is between $N_{i,a}$ and $2N_{i,a}$ 
(that is, neither too few nor too many), we look only at the first $N_{i,a}$ such samples $X^{(j_1)},\dots,X^{(j_{N_{i,a}})}$, and let
\begin{align*}
  Z_{i,a} &\eqdef \sum_{\ell=1}^{N_{i,a}} \indicSet{ X^{(j_\ell)}_i = 1} \\
  Y_{i,a} &\eqdef \sum_{\ell=1}^{N_{i,a}} \indicSet{ X^{(j_\ell)}_i = 0} \;.
\end{align*}
We note that $Z_{i,a}+Y_{i,a}=N_{i,a}$ by construction. 
We then define the quantity 
\[
W_{i,a}\eqdef\frac{((1-q_{i,a})Z_{i,a} - q_{i,a}Y_{i,a})^2 + (2q_{i,a}-1)Z_{i,a} - q_{i,a}^2(Z_{i,a}+Y_{i,a})}{N_{i,a}(N_{i,a}-1)}\indic{N_{i,a}>1} + (p_{i,a}-q_{i,a})^2\indic{N_{i,a}\leq 1}\,.
\]

On the other hand, for every $(i,a)$ such that the number of samples $X^{(j)}$ satisfying the configuration $\Pi_{i,a}$ is less than $N_{i,a}$ or more than $2N_{i,a}$, we continue as a thought experiment and keep on getting samples until we see $N_{i,a}$ samples with the right configuration, and act as above (although the actual algorithm will stop and output $\textsf{reject}$ whenever this happens).
From there, we finally consider the statistic $W$:
\begin{equation}
  W \eqdef \sum_{i=1}^n \sum_{a\in\{0,1\}^{{|\parent{i}|}}} \frac{\probaDistrOf{Q}{\Pi_{i,a}}}{q_{i,a}(1-q_{i,a})} W_{i,a}
\end{equation}
Observe that the algorithm will output $\textsf{reject}$ 
as soon as at least one parental configuration $\Pi_{i,a}$ 
was not seen enough times, or seen too many times, among the $m$ samples.

The pseudocode of our algorithm is given in the following figure.

\begin{figure}[h!]\small
  \begin{framed}
    \begin{description}
      \item[Input] Error tolerance $\eps \in (0,1)$, dimension $n$, description $\structure$ of a DAG with maximum in-degree $d$ 
      and of a $(c,C)$-balanced Bayes net $Q$ with structure $\structure$ (where $c\geq \beta\frac{\log n}{m}$ and $C\geq \beta\frac{d+\log n}{m}$), 
      and sampling access to a distribution $P$ over $\{0,1\}^n$ with structure $\structure$.
      \item[-] Preprocess $Q$ so that $q_{i,a} \leq \frac{1}{2}$ for all $(i,a)\in[n]\times\{0,1\}^d$ (and apply the same transformation to all samples taken from $P$)
      \item[-] Set $m \gets \lceil\alpha\frac{\sqrt{n}}{\eps^2}\rceil$, and take $m$ samples $X^{(1)},\dots,X^{(m)}$ from $P$.
      \item[-] Let $N_{i,a} \gets { m \probaDistrOf{Q}{\Pi_{i,a}} }/\sqrt{2}$ for all $(i,a)\in[n]\times\{0,1\}^d$.
      \item[-] Define $Z_{i,a}, Y_{i,a}, W_{i,a}$
          as above, and $W \eqdef \sum_{i=1}^n \sum_{a\in\{0,1\}^{{|\parent{i}|}}} \probaDistrOf{Q}{\Pi_{i,a}} \frac{W_{i,a}}{q_{i,a}(1-q_{i,a})}$.
      \item \textit{(At this point, if any configuration $\Pi_{i,a}$ was satisfied by less than $N_{i,a}$ or more than $2N_{i,a}$ of the $m$ samples, then the algorithm 
      has rejected already.)}
      \item[If] $W \geq \frac{\eps^2}{32}$ return $\textsf{reject}$.
      \item[Otherwise] return $\textsf{accept}$.
    \end{description}
  \end{framed}
  \caption{Testing identity against a known-structure balanced Bayes net.}\label{algo:bn:identity:known}
\end{figure}

\paragraph{Preprocessing} We will henceforth assume that $q_{i,a} \leq \frac{1}{2}$ for all $(i,a)\in[n]\times\{0,1\}^d$. 
This can be done without loss of generality, as $Q$ is explicitly known. 
For any $i$ such that $q_{i,a}>\frac{1}{2}$, we replace $q_{i,a}$ by $1-q_{i,a}$ and work with the corresponding distribution $Q^\prime$ instead. 
By flipping the corresponding bit of all samples we receive from $P$, it only remains to test identity of the resulting distribution $P^\prime$ to $Q^\prime$, as 
all distances are preserved.

\paragraph{First Observation} If $P=Q$, then we want to argue that with probability at least $9/10$ none of the $W_{i,a}$'s will be such that too few samples satisfied $\Pi_{i,a}$ (as this will immediately cause rejection). To see why this is the case, observe that as long as 
$m \probaDistrOf{Q}{\Pi_{i,a}} \geq \beta(d+\log n)$ (for an appropriate choice of absolute constant $\beta > 0$), 
the number $m_{i,a}$ of samples satisfying $\Pi_{i,a}$ among the $m$ we draw will, by a Chernoff bound, such that $m_{i,a} \geq m \probaDistrOf{Q}{\Pi_{i,a}} \geq N_{i,a}$ with probability at least ${1 -} \frac{1}{2^dn}\cdot\frac{1}{10}$. A union bound over the at most $2^d n$ possible parental configurations 
will yield the desired conclusion. But the fact that $P=Q$ is $(c,C)$-balanced indeed implies that $\probaDistrOf{Q}{\Pi_{i,a}} \geq C \geq \beta \frac{d+\log n}{m}$, the last inequality by our choice of $C$.

Therefore, it will be sufficient to continue our analysis, assuming that none of the $W_{i,a}$'s caused rejection 
because of an insufficient number of samples satisfying $\Pi_{i,a}$. As we argued above, 
this came at the cost of only $1/10$ of probability of success in the completeness case, 
and can only increase the probability of rejection, i.e., success, in the soundness case.

\medskip

Moreover, in the analysis of the expectation and variance of $W$, we assume that for every $(i,a)\in S$, 
we have $\probaDistrOf{P}{\Pi_{i,a}} \leq 4 \probaDistrOf{Q}{\Pi_{i,a}}$. 
This is justified by the following two lemmas, which ensure respectively that if it is not the case, 
then we will have rejected with high probability (this time because \emph{too many} samples satisfied $\Pi_{i,a}$); 
and that we still have not rejected (with high probability) if $P=Q$.

\begin{lemma}\label{lemma:estimating:parental:configuration:soundness}
Let $P$ be as in the statement of~\cref{theo:upper:knowndegreed:identity}, 
and suppose there exists a parental configuration $(i^\ast,a^\ast)\in {S}$ 
such that $\probaDistrOf{P}{\Pi_{i^\ast,a^\ast}} > 4 \probaDistrOf{Q}{\Pi_{i^\ast,a^\ast}}$. 
Then, with probability at least $9/10$, the number of samples $m_{i^\ast,a^\ast}$ satisfying $\Pi_{i^\ast,a^\ast}$ 
among the $m$  samples taken will be more than $2N_{i^\ast,a^\ast}$.
\end{lemma}
\begin{proof}
 This follows easily from a Chernoff bound, as 
 \begin{align*}
  &\probaOf{ m_{i,a} < 2m\probaDistrOf{Q}{\Pi_{i^\ast,a^\ast}} } 
  < \probaOf{ m_{i,a} < \frac{1}{2}m\probaDistrOf{P}{\Pi_{i^\ast,a^\ast}} }  
  = \probaOf{ m_{i,a} < \frac{1}{2}\expect{m_{i,a}} } \;,
 \end{align*}
 and $\expect{m_{i,a}} > \beta (d+\log n).$
\end{proof}

\begin{lemma}\label{lemma:estimating:parental:configuration:completeness}
Suppose $P=Q$. Then, with probability at least $9/10$, for every parental configuration 
$(i,a)\in  {S}$ the number of samples $m_{i,a}$ satisfying $\Pi_{i,a}$ among the $m$ samples taken will be at most $2N_{i,a}$.
\end{lemma}
\begin{proof}
 This again follows from a Chernoff bound and a union bound over all $2^dn$ configurations, as we have $\probaOf{ m_{i,a} > 2m\probaDistrOf{Q}{\Pi_{i,a}} } = \probaOf{ m_{i,a} > 2\expect{m_{i,a}} }$, and $\expect{m_{i,a}} > \beta (d+\log n)$.
\end{proof}

\paragraph{Expectation and Variance Analysis}
We start with a simple closed form formula for the expectation of our statistic:
\begin{lemma}\label{identity:lemma:bn:tree:expectation}
We have that 
$\expect{W} = \sum_{i,a} \probaDistrOf{Q}{\Pi_{i,a}} \frac{ (p_{i,a} - q_{i,a})^2 }{q_{i,a}(1-q_{i,a})}$. 
(In particular, if $P=Q$ then $\expect{W} = 0$.)
\end{lemma}
\begin{proof}
Fix any $(i,a)\in {S}$. Since $Z_{i,a}$ follows a $\binomial{N_{i,a} }{ p_{i,a}}$ distribution, we get
\begin{align*}
    \expect{ W_{i,a} } 
    &= \frac{\expect{ (Z_{i,a}-q_{i,a}N_{i,a})^2+(2q_{i,a}-1)Z_{i,a}-q_{i,a}^2 N_{i,a} }}{N_{i,a}(N_{i,a}-1)}\indic{N_{i,a}>1}
     + \expect{ (p_{i,a}-q_{i,a})^2 }\indic{N_{i,a}\leq 1} \\
    &= (p_{i,a}-q_{i,a})^2 \indic{N_{i,a}>1}+(p_{i,a}-q_{i,a})^2\indic{N_{i,a}\leq 1} \\
    &= (p_{i,a}-q_{i,a})^2 \;,
\end{align*}
giving the result by linearity of expectation. 
The last part follows from the fact that $p_{i,a}=q_{i,a}$ for all $(i,a)$ if $P=Q$.
\end{proof}

As a simple corollary, we obtain:
\begin{claim}\label{identity:lemma:bn:tree:expectation:soundness}
If $\normone{P-Q}\geq \eps$, then $\expect{W} \geq \frac{\eps^2}{16}$.
\end{claim}
\begin{proof}
The claim follows from~Pinsker's inequality and~\cref{lemma:kl:bn}, 
along with our assumption that $\probaDistrOf{P}{\Pi_{i,a}} \leq 4\cdot \probaDistrOf{Q}{\Pi_{i,a}}$ for every $(i,a)$:
\begin{align*}
  \normone{P-Q}^2 
  &\leq 2\dkl{P}{Q} \leq 2\sum_{(i,a)} \probaDistrOf{P}{ \Pi_{i,a} } \frac{(p_{i,a}-q_{i,a})^2}{q_{i,a}(1-q_{i,a})} \\
  &\leq 8\sum_{(i,a)} \probaDistrOf{Q}{ \Pi_{i,a} } \frac{(p_{i,a}-q_{i,a})^2}{q_{i,a}(1-q_{i,a})} \;.
\end{align*}
\end{proof}

We now turn to bounding from above the variance of our statistic. 
This will be done by controlling the covariances and variances of the summands individually, 
and specifically showing that the former are zero. 
We have the following:
\begin{claim}\label{identity:lemma:bn:tree:variances}
If $(i,a)\neq (j,b)$, then $\operatorname{Cov}(W_{i,a}, W_{i,b}) = 0$; 
and the variance satisfies
\begin{align*}
\Var\Big[\frac{\probaDistrOf{Q}{\Pi_{i,a}}}{q_{i,a}(1-q_{i,a})} W_{i,a}\Big] 
&\leq \frac{4}{m} \probaDistrOf{Q}{\Pi_{i,a}}\frac{p_{i,a}(1-p_{i,a})}{q_{i,a}^2(1-q_{i,a})^2 }(p_{i,a}-q_{i,a})^2 
+\frac{4}{m^2} \frac{p_{i,a}^2}{q_{i,a}^2(1-q_{i,a})^2 }\indic{N_{i,a}>1} \;.
\end{align*}
(Moreover, if $P=Q$ then $\Var\left[\frac{\probaDistrOf{Q}{\Pi_{i,a}}}{q_{i,a}(1-q_{i,a})} W_{i,a}\right] \leq \frac{4}{m^2}$.)
\end{claim}
\begin{proof}
The key point is to observe that, because of the way we defined the $Z_{i,a}$'s and $Y_{i,a}$'s 
(only considering the $N_{i,a}$ first samples satisfying the desired parental configuration), 
we have that $W_{i,a}$ and $W_{j,b}$ are independent whenever $(i,a)\neq (j,b)$. This directly implies the first part of the claim, i.e., 
  \begin{align*}
   \operatorname{Cov}(W_{i,a}, W_{i,b}) = \expect{\left( W_{i,a} - \expect{ W_{i,a} } \right)\left( W_{j,b} - \expect{ W_{j,b} } \right)} = 0 \;,
  \end{align*}
  when $(i,a)\neq (j,b)$. 
  
  We then consider $\Var\left[\frac{\probaDistrOf{Q}{\Pi_{i,a}}}{q_{i,a}(1-q_{i,a})} W_{i,a}\right]$. Note that 
  \begin{align*}
    \expect{ W_{i,a}^2 } 
    &= \frac{\expect{ ((Z_{i,a}-q_{i,a}N_{i,a})^2+(2q_{i,a}-1)Z_{i,a}-q_{i,a}^2 N_{i,a})^2 }}{N_{i,a}^2(N_{i,a}-1)^2} \indic{N_{i,a}>1} + (p_{i,a}-q_{i,a})^4\indic{N_{i,a}\leq 1} \;,
  \end{align*}
  so that, writing $p,q,N,Z$ for $p_{i,a},q_{i,a},N_{i,a},Z_{i,a}$ respectively (for readability), as well as $M\eqdef N-1$:
    \begin{align*}
    \Var\Big[\frac{\probaDistrOf{Q}{\Pi_{i,a}}}{q(1-q)} W_{i,a}\Big]
    &=\left(\frac{\probaDistrOf{Q}{\Pi_{i,a}}}{q(1-q)}\right)^2\left( \expect{ W_{i,a}^2 } - \expect{W_{i,a}}^2 \right) 
    = \left(\frac{\probaDistrOf{Q}{\Pi_{i,a}}}{q(1-q)}\right)^2\left( \expect{ W_{i,a}^2 }- (p-q)^4 \right) \\
    &= \frac{\expect{ ((Z-qN)^2+(2q-1)Z-q^2 N)^2 - N^2M^2(p-q)^4 }}{N^2M^2 q^2(1-q)^2 }\cdot\probaDistrOf{Q}{\Pi_{i,a}}^2 \indic{N>1}\\
    &= \frac{\expect{ ((Z-qN)^2+(2q-1)Z-q^2 N)^2 - N^2M^2(p-q)^4 }}{m^2M^2 q^2(1-q)^2 }\indic{N>1}\\
    &= \frac{\indic{N>1}}{m^2}\frac{2Np(1-p)}{Mq^2(1-q)^2 }\left((2N-3)p^2 + 2Mq^2-4Mpq + p\right) \;.
    \end{align*}
    If $p=q$, then this becomes $\Var\left[\frac{\probaDistrOf{Q}{\Pi_{i,a}}}{q(1-q)} W_{i,a}\right] = \frac{1}{m^2}\frac{2N}{N-1}\indic{N_{i,a}>1} \leq \frac{4}{m^2}$, providing the second part of the claim. In the general case, we can bound the variance as follows:
%     \begin{align*}
%     \Var\big[\frac{\probaDistrOf{Q}{\Pi_{i,a}}}{q(1-q)} W_{i,a}\big] &\leq \frac{1}{m^2}\frac{2N}{N-1}\frac{p(1-p)}{q^2(1-q)^2 }\indic{N>1}\left((2N-3) + 2(N-1)+4(N-1) + 1\right)\\
%     &\leq \frac{1}{m^2}\frac{2N}{N-1}\frac{p(1-p)}{q^2(1-q)^2 }\indic{N>1}\left(8(N-1)\right)
%     = \frac{8N}{m^2}\frac{p(1-p)}{q^2(1-q)^2 }\indic{N>1}\\\
%     &= \frac{8}{m} \probaDistrOf{Q}{\Pi_{i,a}}\frac{p_{i,a}(1-p_{i,a})}{q_{i,a}^2(1-q_{i,a})^2 }\indic{N_{i,a}>1}.
%   \end{align*}
    \begin{align*}
    \Var\Big[\frac{\probaDistrOf{Q}{\Pi_{i,a}}}{q(1-q)} W_{i,a}\Big] 
    &= \frac{1}{m^2}\frac{2N}{N-1}\frac{p(1-p)}{q^2(1-q)^2 }\indic{N>1}\cdot\left(2(N-1)(p^2+q^2-2pq)-p^2+p\right)\\
    &= \frac{1}{m^2}\frac{2N}{N-1}\frac{p(1-p)}{q^2(1-q)^2 }\indic{N>1}\cdot\left(2(N-1)(p-q)^2+p(1-p)\right)\\
    &= \frac{4N}{m^2}\frac{p(1-p)}{q^2(1-q)^2 }(p-q)^2\indic{N>1} + \frac{1}{m^2}\frac{2N}{N-1}\frac{p^2(1-p)^2}{q^2(1-q)^2 }\indic{N>1}\\
    &\leq \frac{4N}{m^2}\frac{p(1-p)}{q^2(1-q)^2 }(p-q)^2 + \frac{4}{m^2}\frac{p^2(1-p)^2}{q^2(1-q)^2 }\indic{N>1}\\
    &=\frac{4}{m} \probaDistrOf{Q}{\Pi_{i,a}}\frac{p_{i,a}(1-p_{i,a})}{q_{i,a}^2(1-q_{i,a})^2 }(p_{i,a}-q_{i,a})^2  +\frac{4}{m^2}\frac{p_{i,a}^2(1-p_{i,a})^2}{q_{i,a}^2(1-q_{i,a})^2 }\indic{N_{i,a}>1} \\
    &\leq \frac{4}{m} \probaDistrOf{Q}{\Pi_{i,a}}\frac{p_{i,a}(1-p_{i,a})}{q_{i,a}^2(1-q_{i,a})^2 }(p_{i,a}-q_{i,a})^2  +\frac{4}{m^2} \frac{p_{i,a}^2}{q_{i,a}^2(1-q_{i,a})^2 }\indic{N_{i,a}>1}.
  \end{align*}
This completes the proof.
\end{proof}

Using this claim, we now state the upper bound it allows us to obtain:
\begin{lemma}\label{identity:coro:bn:tree:variance}
We have that 
$\Var[W] \leq {24} \frac{2^d n}{m^2} + {26} \frac{\expect{W}}{cm}$. 
(Moreover, if $P=Q$ we have $\Var[W] \leq 4\frac{2^dn}{m^2}$.)
\end{lemma}
\begin{proof}
This will follow from~\cref{identity:lemma:bn:tree:variances}, 
which guarantees that if $P=Q$, $\Var[W] \leq 2^dn\cdot \frac{4}{m^2} = 4\frac{2^d n}{m^2}$. 
Moreover, in the general case,
\begin{align*}
  \Var[W] 
      &\leq \frac{4}{m} \sum_{(i,a)}\probaDistrOf{Q}{\Pi_{i,a}}\frac{p_{i,a}(1-p_{i,a})}{q_{i,a}^2(1-q_{i,a})^2 }(p_{i,a}-q_{i,a})^2 + \frac{4}{m^2} \sum_{(i,a)}\frac{p_{i,a}^2}{q_{i,a}^2(1-q_{i,a})^2 }\indic{N_{i,a}>1} \;.
\end{align*}
We deal with the two terms separately, as follows:
\begin{itemize}
  \item For the second term, we will show that 
  \[
    \frac{4}{m^2} \sum_{(i,a)}\frac{p_{i,a}^2}{q_{i,a}^2(1-q_{i,a})^2 }\indic{N_{i,a}>1} \leq 24\frac{2^d n}{m^2} + \frac{24\expect{W}}{cm} \;.
  \]
  This follows from the following sequence of (in)equalities: first,
\begin{align*}
  \sum_{(i,a)}\frac{p_{i,a}^2}{q_{i,a}^2(1-q_{i,a})^2 }\indic{N_{i,a}>1} 
  &= \sum_{(i,a)} \frac{(p_{i,a}-q_{i,a})^2}{q_{i,a}^2(1-q_{i,a})^2}\indic{N_{i,a}>1} + \sum_{(i,a)} \frac{2p_{i,a}q_{i,a}-q_{i,a}^2}{q_{i,a}^2(1-q_{i,a})^2}\indic{N_{i,a}>1} \\
  &= \sum_{(i,a)} \frac{(p_{i,a}-q_{i,a})^2}{q_{i,a}^2(1-q_{i,a})^2}\indic{N_{i,a}>1} + \sum_{(i,a)} \frac{2q_{i,a}(p_{i,a} - q_{i,a})+q_{i,a}^2}{q_{i,a}^2(1-q_{i,a})^2}\indic{N_{i,a}>1} \\
  &\leq 4\cdot 2^d n+\sum_{(i,a)} \frac{(p_{i,a}-q_{i,a})^2}{q_{i,a}^2(1-q_{i,a})^2}\indic{N_{i,a}>1} + \sum_{(i,a)} \frac{2(p_{i,a} - q_{i,a})}{q_{i,a}(1-q_{i,a})^2}\indic{N_{i,a}>1}\indic{N_{i,a}>1} \\
  &\leq 4\cdot 2^d n+\sum_{(i,a)} \frac{(p_{i,a}-q_{i,a})^2}{q_{i,a}^2(1-q_{i,a})^2}\indic{N_{i,a}>1} + 4\sum_{(i,a)} \frac{p_{i,a} - q_{i,a}}{q_{i,a}(1-q_{i,a})}\indic{N_{i,a}>1}
\end{align*}
Then, applying the AM-GM inequality we can continue with
\begin{align*}
\sum_{(i,a)}\frac{p_{i,a}^2}{q_{i,a}^2(1-q_{i,a})^2 }\indic{N_{i,a}>1} 
  &\leq 4\cdot 2^d n +\sum_{(i,a)} \frac{(p_{i,a}-q_{i,a})^2}{q_{i,a}^2(1-q_{i,a})^2}\indic{N_{i,a}>1} + 2\sum_{(i,a)} \left( 1+\frac{(p_{i,a} - q_{i,a})^2}{q_{i,a}^2(1-q_{i,a})^2} \right)\indic{N_{i,a}>1} \\
  &\leq 6\cdot 2^d n +3\sum_{(i,a)} \frac{(p_{i,a}-q_{i,a})^2}{q_{i,a}^2(1-q_{i,a})^2}\indic{N_{i,a}>1} \\
  &\leq 6\cdot 2^d n+\frac{6}{c}\sum_{(i,a)} \frac{(p_{i,a}-q_{i,a})^2}{q_{i,a}(1-q_{i,a})}\indic{N_{i,a}>1} \\
 &\leq 6\cdot 2^d n+\frac{6m}{c}\sum_{(i,a)} \frac{N_{i,a}}{m} \frac{(p_{i,a}-q_{i,a})^2}{q_{i,a}(1-q_{i,a})}\indic{N_{i,a}>1} \\
  &= 6\cdot 2^d n+\frac{6m}{c}\sum_{(i,a)} \probaDistrOf{Q}{\Pi_{i,a}} \frac{(p_{i,a}-q_{i,a})^2}{q_{i,a}(1-q_{i,a})}\indic{N_{i,a}>1} \\
  &\leq 6\cdot 2^d n+\frac{6m}{c}\expect{W} \;,
\end{align*}
  using our assumption that $q_{i,a} \leq \frac{1}{2}$ for all $(i,a)$.
  \item For the first term, we will establish that 
  \[
  \frac{4}{m} \sum_{(i,a)}\probaDistrOf{Q}{\Pi_{i,a}}\frac{p_{i,a}(1-p_{i,a})}{q_{i,a}^2(1-q_{i,a})^2 }(p_{i,a}-q_{i,a})^2
  \]
  is at most $\frac{2}{cm}\expect{W}$. This is shown as follows:
  \begin{align*}
  \sum_{(i,a)} \probaDistrOf{Q}{\Pi_{i,a}}\frac{p_{i,a}(1-p_{i,a}) (p_{i,a}-q_{i,a})^2}{q_{i,a}^2(1-q_{i,a})^2}  
  &\leq \frac{1}{4}\sum_{(i,a)} \frac{1}{q_{i,a}(1-q_{i,a})}\cdot \probaDistrOf{Q}{\Pi_{i,a}}\frac{(p_{i,a}-q_{i,a})^2}{q_{i,a}(1-q_{i,a})} \\
  &\leq \frac{1}{2c}\sum_{(i,a)}\probaDistrOf{Q}{\Pi_{i,a}}\frac{(p_{i,a}-q_{i,a})^2}{q_{i,a}(1-q_{i,a})} \\
  &= \frac{1}{2c}\expect{W} \;.
\end{align*}
\end{itemize}
Combining the above, we conclude that $\Var[W] \leq {24} \frac{2^d n}{m^2} + {26} \frac{\expect{W}}{cm}$.
\end{proof}

We now have all the tools we require to establish the completeness and soundness of the tester.

\begin{lemma}[Completeness]
If $P=Q$, then the algorithm outputs $\textsf{accept}$ with probability at least $2/3$.
\end{lemma}
\begin{proof}
  We first note that, as per the foregoing discussion and~\cref{lemma:estimating:parental:configuration:completeness}, 
  with probability at least $8/10$ we have between $N_{i,a}$ and $2N_{i,a}$ samples for every parental configuration $(i,a)\in S$, 
  and therefore have not outputted $\textsf{reject}$.
  By Chebyshev's inequality and~\cref{identity:coro:bn:tree:variance},
  \[
      \probaOf{ W \geq \frac{\eps^2}{32} } \leq {4096}\frac{2^d n}{m^2\eps^4} \leq \frac{4}{30}
  \]
  for a suitable choice of $\alpha>0$. Therefore, by a union bound the algorithm will output $\textsf{reject}$ with probability at most $\frac{4}{30}+\frac{2}{10} = \frac{1}{3}$.
\end{proof}

\begin{lemma}[Soundness]
If $\normone{P-Q} \geq \eps$, then the algorithm outputs $\textsf{reject}$ with probability at least $2/3$.
\end{lemma}
\begin{proof}
As noted before, it is sufficient to show that, conditioned on having between $N_{i,a}$ and $2N_{i,a}$ samples for every parental configuration and $\probaDistrOf{P}{\Pi_{i^\ast,a^\ast}} \leq 4 \probaDistrOf{Q}{\Pi_{i^\ast,a^\ast}}$ for all $(i,a)$, the algorithm rejects with probability at least $2/3+1/10 = 23/30$. Indeed, whenever too few or too many samples from a given parental configuration are seen the algorithm rejects automatically, and by~\cref{lemma:estimating:parental:configuration:soundness} this happens with probability at least $9/10$ if some parental configuration is such that $\probaDistrOf{P}{\Pi_{i^\ast,a^\ast}} > 4 \probaDistrOf{Q}{\Pi_{i^\ast,a^\ast}}$.
Conditioning on this case, by Chebyshev's inequality,
  \begin{align*}
      \probaOf{ W \leq \frac{\eps^2}{32} } 
      &\leq \probaOf{ \abs{ W - \expect{W} } \geq \frac{1}{2}\expect{W} } 
      \leq \frac{4\Var[W]}{ \expect{W}^2 } \\
      &\leq {96} \frac{2^d n}{m^2\expect{W}^2} + {104}\frac{1}{cm\expect{W}} \;,
  \end{align*}
  from~\cref{identity:coro:bn:tree:variance}. Since $\expect{W} \geq \frac{\eps^2}{16}$ by~\cref{identity:lemma:bn:tree:expectation:soundness}, we then get $\probaOf{ W \leq \frac{\eps^2}{32} }  = \bigO{\frac{2^d n}{m^2\eps^4} + \frac{1}{cm\eps^2} } \leq \frac{17}{30}$, again for a suitable choice of $\alpha>0$ and $\beta>0$ (recalling that $c\geq \beta\frac{\log n}{ \sqrt{n} }$). %%% \cnote{For reference: this is where we need $cm\eps^2 \gg 1$.}
\end{proof}

%%%%%%%%%%%%%%%%%%%%%%%%%%%%%%%%%%%%%%%%%%%%%%%%%%%%%%%%%%%%%%%%%%%%%%%%%%%%%%%%%%%%%%%%%%%%%%%%%%%%%%%%%%%%%%%%%%%%%%%%%%
%%%%%%%%%%%%%%%%%%%%%%%%%%%%%%%%%%%%%%%%%%%%%%%%%%%%%%%%%%%%%%%%%%%%%%%%%%%%%%%%%%%%%%%%%%%%%%%%%%%%%%%%%%%%%%%%%%%%%%%%%%

\begin{remark} \label{rem:alph}
{\em We note that we can reduce the problem of testing degree-$d$ 
Bayes nets over alphabet $\Sigma$, to testing 
degree-$((d+1) \lceil \log_2(|\Sigma|) \rceil -1)$
Bayes nets over alphabet of size $2$. First consider the case where $|\Sigma| = 2^b$. 
Then it suffices to have $nb$ bits in $n$ clusters of size $b$. 
Each cluster of $b$ will represent a
single variable in the initial model with each of the $2^b$ possibilities
denoting a single letter. Then each bit will need to potentially be
dependent on each other bit in its cluster and on each bit in each
cluster that its cluster is dependent on. Therefore, we need degree
$(d+1)b-1$. Note that this operation preserves balancedness.

Now if $|\Sigma|$ is not a power of $2$, we need to pad the alphabet.
The obvious way to do this is to create a set of unused letters
until the alphabet size is a power of $2$. Unfortunately, this creates
an unbalanced model. To create a balanced one, we proceed as follows: 
we split a number of the letters in $\Sigma$ in two. So, instead of having
alphabet $a, b, c, \ldots$, we have $a_1,a_2,b_1,b_2,c,\ldots$. 
We make it so that when a word would have an $a$ in a certain position, 
we map this to a new word that has either $a_1$ or $a_2$ in that position, 
each with equal probability. We note that this operation preserves $L_1$ distance, 
and maintains the balancedness properties.
}
\end{remark}

\subsection{Sample Complexity Lower Bound} \label{ssec:identity-known-lower}

Here we prove a matching information-theoretic lower bound:

\begin{restatable}{theorem}{uniformityknownlb} \label{theo:lb:known:uniform}
There exists an absolute constant $\eps_0 > 0$ such that, for any $0 < \eps \leq \eps_0$, the following holds:
Any algorithm that has sample access to an unknown Bayes net $P$ over $\{0,1\}^n$ with known structure $\structure$ of
maximum in-degree at most $d < n/2$, and distinguishes between the cases that $P=U$ and $\normone{P-U} > \eps$ 
requires $\Omega(2^{d/2} n^{1/2}/\eps^2)$ samples.
\end{restatable}
\begin{proof}

Our lower bound will be derived from families of Bayes nets with the following structure:
The first $d$ nodes are all independent (and will in fact have marginal probability $1/2$ each), 
and will form in some sense a ``pointer'' to one of $2^d$ arbitrary product distributions. 
The remaining $n-d$ nodes will each depend on all
of the first $d$. The resulting distribution is now an (evenly weighted) disjoint mixture of $2^d$ 
product distributions on the $(n-d)$-dimensional hypercube. 
In other words, there are $2^d$ product distributions $p_1,\dots,p_{2^d}$, 
and our distribution returns a random $i$ (encoded in binary) followed by a random sample form $p_i$. 
Note that the $p_i$ can be arbitrary product distributions.

The unknown distribution $P$ to test is obtained as follows: 
let $X$ be a Bernoulli random variable with parameter $1/2$. 
If $X=0$, $P$ is the uniform distribution on $\{0,1\}^n$, i.e., each of the $2^d$ distributions $p_i$ 
is uniform on $\{0,1\}^{n-d}$. Otherwise, if $X=1$, then every $p_i$ is a product distribution 
on $\{0,1\}^{n-d}$ with, for each coordinate, a parameter chosen uniformly and independently 
to be either $\frac{1}{2}+\frac{\eps}{\sqrt{n}}$ or $\frac{1}{2}-\frac{\eps}{\sqrt{n}}$.

We will show that the shared information between a sample of size $o(2^{d/2}n^{1/2}/\eps^2)$ and $X$ is small. 
In view of this, let $\sigma_i$ (for $1\leq i \leq n-d$) be the set of indices of the samples that were drawn from $p_i$. 
Note that since $X$ is uncorrelated with the $\sigma_i$'s, and as the $\sigma_i$ are a function of the samples, 
$\mutualinfo{X}{S} = \mutualinfo{X}{S \mid \sigma_i}$. This is because $\mutualinfo{X}{S}) = H(X) - H(X \mid S) =
 H(X\mid \sigma_i) - H(X \mid S,\sigma_i) = \mutualinfo{X}{S \mid \sigma_i}$. 
 
Now, for fixed $\sigma_i$, the samples  we draw from $p_i$ are mutually independent of $X$.
Let $S_i$ denote the tuple of these $\abs{\sigma_i}$ samples. Thus, we have that
$\mutualinfo{X}{S \mid \sigma_i} \leq \sum_i \mutualinfo{X}{S_i \mid \sigma_i}$. 
By the same analysis as in the proof of~\cref{theo:lb:product:uniform}, 
this latter term is $O(\binom{\abs{\sigma_i}}{2}\frac{\eps^4}{n})$. Therefore,
$$\mutualinfo{X}{S \mid \sigma_i} \leq \expect{ \sum_i \binom{\abs{\sigma_i}}{2} }O\left(\frac{\eps^4}{n}\right) = O\left(\frac{m^2 \eps^4}{n2^d}\right) \;,$$ 
where we used the fact that $\abs{\sigma_i}$ is $\binomial{m}{1/2^d}$ distributed.
Note that the above RHS is $o(1)$ unless $m = \Omega(2^{d/2}n^{1/2}/\eps^2)$, which completes the proof.
\end{proof}

\section{Testing Identity of Unknown Structure Bayes Nets} \label{sec:identity-uknown}

In this section, we give our algorithms and lower bounds for testing the identity of low-degree Bayes nets with unknown structure.
In~\cref{ssec:identity-unknown-lower}, we start by showing that~--~even for the case of trees~--~uniformity testing of $n$-node Bayes nets requires $\Omega(n/\eps^2)$ samples. In Sections~\ref{ssec:identity-unknown-upper}, 
we design efficient identity testers with sample complexity sublinear in the dimension $n$, under some non-degeneracy
assumptions on the explicit Bayes net.

\subsection{Sample Complexity Lower Bound}  \label{ssec:identity-unknown-lower}

In this section, we establish a tight lower bound on identity testing of Bayes nets in the unknown structure case.
Our lower bound holds even for \emph{balanced} Bayes nets with a \emph{tree} structure. 
In order to state our theorem, we first give a specialized definition of balancedness for the case of trees. 
We say that a Bayes net with tree structure is \emph{$c$-balanced} if it satisfies $p_k\in[c,1-c]$ for all $k$ (note that this immediately implies it is $(c,C)$-balanced).

\begin{restatable}{theorem}{uniformitybntreelb}\label{theo:lb:unknown:structure:bayes:tree:uniform}
There exist absolute constants $c > 0$ and $\eps_0>0$ such that, for any $\eps \in(0,\eps_0)$ and given samples from an unknown $c$-balanced Bayes net $P$ over $\{0,1\}^n$ with unknown tree structure, distinguishing between the cases $P=U$ and $\normone{P-U} > \eps$ (where $U$ is the uniform distribution over $\{0,1\}^n$) 
{with probability $2/3$} requires $\Omega(n/\eps^2)$ samples. (Moreover, one can take $c=1/3$.)
\end{restatable}

{Hence, without any assumptions about the explicit distribution, identity testing is information-theoretically as hard as learning.
This section is devoted to the proof of~\cref{theo:lb:unknown:structure:bayes:tree:uniform}.
}

Fix any integer $m\geq 1$. We will define a family of $\textsf{no}$-instances consisting of distributions $\{P_\lambda\}_\lambda$ over $\{0,1\}^{n}$ such that:
\begin{enumerate}
  \item\label{item:lb:farness} every $P_\lambda$ is $\eps$-far from the uniform distribution $U$ on $\{0,1\}^{n}$: $\| P_\lambda - U\|_1 = \bigOmega{\eps}$;
  \item\label{item:lb:treebn} every $P_\lambda$ is a Bayes net with a tree structure;
  \item\label{item:lb:indisting} unless $m=\bigOmega{\frac{n}{\eps^2}}$, no algorithm taking $m$ samples can distinguish with probability $2/3$ between a uniformly chosen distribution from $\{P_\lambda\}_\lambda$ and $u$; or, equivalently, no algorithm taking \emph{one} sample can distinguish with probability $2/3$ between $P_\lambda^{\otimes m}$ and $U^{\otimes m}$, when $P_\lambda$ is chosen uniformly at random from $\{P_\lambda\}_\lambda$.
\end{enumerate}

The family is defined as follows. We let $\delta \eqdef \frac{\eps}{\sqrt{n}}$, and let a \emph{matching-orientation parameter} $\lambda$ consist of (i) a matching $\lambda^{(1)}$ of $[n]$ (partition of $[n]$ in $\frac{n}{2}$ disjoint pairs $(i,j)$ with $i<j$) and (ii) a vector $\lambda^{(2)}$ of $\frac{n}{2}$ bits. The distribution $P_\lambda$ is then defined as the distribution over $\{0,1\}^{n}$ with uniform marginals, and tree structure with edges corresponding to the pairs $\lambda^{(1)}$; and such that for every $\lambda^{(1)}_k = (i,j)\in \lambda^{(1)}$, $\operatorname{cov}(X_i, X_j) = (-1)^{\lambda^{(2)}_k}\delta$.

\paragraph{Notation}
For $\lambda=(\lambda^{(1)},\lambda^{(2)})$ as above and $x\in\{0,1\}^n$, we define the \emph{agreement count of $x$ for $\lambda$}, $c(\lambda,x)$, as the number of pairs $(i,j)$ in $\lambda^{(1)}$ such that $(x_i,x_j)$ ``agrees'' with the correlation suggested by $\lambda^{(2)}$. Specifically:
\begin{align*}
    c(\lambda,x) &\eqdef \Big\lvert \Big\{ (i,j) \in [n]^2 \;\colon\; \exists \ell \in [n/2],\ \lambda^{(1)}_\ell = (i,j) \text{ and } (-1)^{x_i+x_j} = (-1)^{\lambda^{(2)}_\ell} \Big\} \Big\rvert.
\end{align*}
Moreover, for $\lambda,\mu$ two matching-orientation parameters, we define the sets $A=A_{\lambda,\mu},B=B_{\lambda,\mu},C=C_{\lambda,\mu}$ as
\begin{align*}
  A &\eqdef \setOfSuchThat{ (s,t)\in[n/2]^2 }{ \lambda^{(1)}_s=\mu^{(1)}_t, \quad \lambda^{(2)}_s = \mu^{(2)}_t } \tag{common pairs with same orientations}\\
  B &\eqdef \setOfSuchThat{ (s,t)\in[n/2]^2 }{ \lambda^{(1)}_s=\mu^{(1)}_t, \quad \lambda^{(2)}_s\neq \mu^{(2)}_t } \tag{common pairs with different orientations}\\
  C &\eqdef (\lambda^{(1)}\cup\mu^{(1)})\setminus(A\cup B) \tag{pairs unique to $\lambda$ or $\mu$}
\end{align*}
so that $2(\abs{A}+\abs{B})+\abs{C} = n$.

\paragraph{Proof of~\cref{item:lb:farness}} Fix any matching-orientation parameter $\lambda$. We have
\begin{align*}
  \normone{P_\lambda-U} 
  &= \sum_{x\in\{0,1\}^n} \abs{P_\lambda(x) - U(x)} \\
  &= \sum_{x\in\{0,1\}^n} \abs{U(x) (1+2\delta)^{c(\lambda,x)}(1-2\delta)^{\frac{n}{2}-c(\lambda,x)} - U(x)} \\
  &= \frac{1}{2^n}\sum_{x\in\{0,1\}^n} \abs{(1+2\delta)^{c(\lambda,x)}(1-2\delta)^{\frac{n}{2}-c(\lambda,x)} - 1}  \\
  &= \frac{1}{2^n}\sum_{k=0}^{\frac{n}{2}}\sum_{x\colon c(\lambda,x)=k} \abs{(1+2\delta)^{k}(1-2\delta)^{\frac{n}{2}-k} - 1} \\
  &= \frac{1}{2^n}\sum_{k=0}^{\frac{n}{2}} 2^{\frac{n}{2}}\binom{\frac{n}{2}}{k} \abs{(1+2\delta)^{k}(1-2\delta)^{\frac{n}{2}-k} - 1} \\
  &= \sum_{k=0}^{\frac{n}{2}} \binom{\frac{n}{2}}{k} \abs{\left(\frac{1+2\delta}{2}\right)^{k}\left(\frac{1-2\delta}{2}\right)^{\frac{n}{2}-k} - \frac{1}{2^{\frac{n}{2}}}  } \\
  &= 2\totalvardist{ \binomial{\frac{n}{2}}{\frac{1}{2}} }{ \binomial{\frac{n}{2}}{\frac{1}{2}+\delta} } \\
  &= \bigOmega{\eps} \;,
\end{align*}
where the last equality follows from~\cref{lemma:noinstances:far}.

\paragraph{Proof of~\cref{item:lb:indisting}} Let the distribution $Q$ over $(\{0,1\}^n)^m$ be the uniform mixture
\[
    Q \eqdef \shortexpect_\lambda[ P_\lambda^{\otimes m} ] \;,
\]
where $P_\lambda$ is the distribution on $\{0,1\}^n$ corresponding to the matching-orientation parameter $\lambda$. In particular, for any $x\in\{0,1\}^n$ we have
\[
P_\lambda(x) = U(x) (1+2\delta)^{c(\lambda,x)}(1-2\delta)^{\frac{n}{2}-c(\lambda,x)}
\]
with $U$ being the uniform distribution on $\{0,1\}^n$ and $c(\lambda,x)$, the agreement count of $x$ for $\lambda$, defined as before. Now, this leads to
\[
    \frac{dP_\lambda}{du}(x) = 1+G(\lambda,x) \;,
\]
where $G(\lambda,x) \eqdef (1-2\delta)^{\frac{n}{2}}\left(\frac{1+2\delta}{1-2\delta}\right)^{c(\lambda,x)}-1$. For two matching-orientation parameters $\lambda,\mu$, we can define the covariance $\tau(\lambda,\mu)\eqdef \shortexpect_{x\sim U}[ G(\lambda,x)G(\mu,x) ]$. By the minimax approach (as in~\cite[Chapter 3]{Pollard:2003}), it is sufficient to bound the {$L_1$}-distance between $Q$ and $U^{\otimes m}$ by a small constant. Moreover, we have
\begin{equation}\label{eq:minimax:mixture:bound}
  \normone{Q-U^{\otimes m}} \leq \shortexpect_{\lambda,\mu}\left[ (1+\tau(\lambda,\mu))^m \right] - 1
\end{equation}
and to show the lower bound it is sufficient to prove that the RHS is less than $\frac{1}{10}$ unless $m=\bigOmega{\frac{n}{\eps^2}}$.

\noindent Setting $z\eqdef\frac{1+2\delta}{1-2\delta}$, we can derive, by expanding the definition
\begin{align*}
    \tau(\lambda,\mu) 
    &= 1+(1-{2}\delta)^n\shortexpect_{x\sim U}[ z^{c(\lambda,x)+c(\mu,x)} ] -2(1-2\delta)^{\frac{n}{2}}\shortexpect_{x\sim U}[ z^{c(\lambda,x)} ].
\end{align*}
Since, when $x$ is uniformly drawn in $\{0,1\}^n$, $c(\lambda,x)$ follows a $\binomial{\frac{n}{2}}{\frac{1}{2}}$ distribution, we can compute the last term as
\begin{align*}
    2(1-2\delta)^{\frac{n}{2}}\shortexpect_{x\sim U}[ z^{c(\lambda,x)} ] 
    &= 2(1-2\delta)^{\frac{n}{2}}\left(\frac{1+z}{2}\right)^{\frac{n}{2}} 
    = 2(1-2\delta)^{\frac{n}{2}}\frac{1}{(1-2\delta)^{\frac{n}{2}}} = 2 \;,
\end{align*}
where we used the expression of the probability-generating function of a Binomial. This leads to 
\begin{align*}
    1+\tau(\lambda,\mu) 
    &= (1-2\delta)^n \shortexpect_{x\sim U}[ z^{c(\lambda,x)+c(\mu,x)} ] 
    = (1-2\delta)^n z^{\abs{B}} \shortexpect_{\alpha\sim \binomial{\abs{A}}{\frac{1}{2}}}[ z^{2\alpha} ] 
            \prod_{\substack{\sigma\text{ cycle}\\ \abs{\sigma} \geq 4}} \shortexpect_{\alpha\sim \mathcal{B}_{\lambda,\mu}(\sigma)}[ z^{\alpha} ] \;,
\end{align*}
where ``cycle'' and the probability distribution $\mathcal{B}_{\lambda,\mu}(\sigma)$ are defined as follows. Recall that $\lambda$ and $\mu$ define a weighted multigraph 
over $n$ vertices, where each vertex has degree exactly $2$, the edges are from the pairs $\lambda^{(1)}_i$'s and {$\mu^{(1)}_i$}'s, and the weights are in $\{0,1\}$ according to the $\lambda^{(2)}_i$'s and {$\mu^{(2)}_i$}'s. That multigraph $G_{\lambda,\mu}$ is better seen as the disjoint union of cycles ({and indeed, $A\cup B$ corresponds to the cycles of length $2$, while $C$ corresponds to cycles of length at least $4$)}.

For such a cycle $\sigma$ in $G_{\lambda,\mu}$, we let $\mathcal{B}_{\lambda,\mu}(\sigma)$ be the distribution below. If the number of negative covariances~--~the number of edges with label $\lambda^{(2)}_\ell=1$ or $\mu^{(2)}_\ell=1$~--~along $\sigma$ is even (resp. odd), then 
$\mathcal{B}_{\lambda,\mu}(\sigma)$ is a $\binomial{\abs{\sigma}}{\frac{1}{2}}$ conditioned on being even (resp. odd).\medskip

Instead of the above, we first consider the related quantity with the conditioning removed (indeed, as we will see in~\cref{claim:binomial:removing:evenodd:conditioning}, this new quantity is within an $1+O(\eps^2)$ factor of the actual one): in what follows, all the expectations are taken over the random variable $\alpha$, distributed as indicated in the subscript:
\begin{align*}
    1+\tilde{\tau}(\lambda,\mu)  
    &= (1-2\delta)^n z^{\abs{B}} \shortexpect_{\binomial{\abs{A}}{\frac{1}{2}}}[ z^{2\alpha} ] \prod_{\substack{\sigma\text{ cycle}\\ \abs{\sigma} \geq 4}} \shortexpect_{\binomial{\abs{\sigma}}{\frac{1}{2}}}[ z^{\alpha} ]\\
    &= (1-2\delta)^n z^{\abs{B}} \shortexpect_{\binomial{\abs{A}}{\frac{1}{2}}}[ z^{2\alpha} ] \shortexpect_{\binomial{\sum_{\sigma\colon\abs{\sigma} \geq 4}\abs{\sigma}}{\frac{1}{2}}}[ z^{\alpha} ] \\
    &= (1-2\delta)^n z^{\abs{B}} \shortexpect_{\binomial{\abs{A}}{\frac{1}{2}}}[ z^{2\alpha} ] \shortexpect_{\binomial{\abs{C}}{\frac{1}{2}}}[ z^{\alpha} ] \\
    &= (1-2\delta)^n z^{\abs{B}} \big(\frac{1+z^2}{2}\big)^{\abs{A}}\big(\frac{1+z}{2}\big)^{\abs{C}} \\
    &= \left( (1-2\delta)^2z \right)^{\abs{B}} \big((1-2\delta)^2\frac{1+z^2}{2}\big)^{\abs{A}}{\underbrace{\big((1-2\delta)\frac{1+z}{2}\big)}_{=1}}^{\abs{C}} \tag{since $2\abs{A}+2\abs{B}+\abs{C}=n$}\\
    &= \left( 1-4\delta^2\right)^{\abs{B}} \left(1+4\delta^2\right)^{\abs{A}}.
\end{align*}
Thus, we need to compute
\begin{align*}
    \shortexpect_{\lambda,\mu}\left[ (1+\tilde{\tau}(\lambda,\mu))^m \right] 
    &= \shortexpect_{\lambda,\mu}\left[ \left(1+4\delta^2\right)^{m\abs{A}}\left( 1-4\delta^2\right)^{m\abs{B}} \right] 
    = \shortexpect_{\lambda,\mu}\left[ a^{\abs{A}} b^{\abs{B}} \right]
\end{align*}   
where $a\eqdef\left(1+4\delta^2\right)^{m}$, $b\eqdef\left(1-4\delta^2\right)^{m}$. This leads to
\begin{align*}
    \shortexpect_{\lambda,\mu}\left[ (1+\tilde{\tau}(\lambda,\mu))^m \right] 
    &= \shortexpect_{\lambda,\mu}\left[ \expectCond{ a^{\abs{A}} b^{\abs{B}} }{ \abs{A}+\abs{B} } \right]\\
    &= \shortexpect_{\lambda,\mu}\left[ b^{\abs{A}+\abs{B}}\expectCond{ \left(\frac{a}{b}\right)^{\abs{A}} }{ \abs{A}+\abs{B} } \right] \\
    &= \shortexpect_{\lambda,\mu}\left[ b^{\abs{A}+\abs{B}} \left(\frac{1+\frac{a}{b}}{2}\right)^{ \abs{A}+\abs{B} }  \right] \tag{since $\abs{A} \sim \binomial{\abs{A}+\abs{B}}{\frac{1}{2}}$}\\
    &= \shortexpect_{\lambda,\mu}\left[ \left(\frac{a+b}{2}\right)^{ \abs{A}+\abs{B} } \right]\\
    &= \shortexpect_{\lambda,\mu}\left[ \left(\frac{\left(1+4\delta^2\right)^{m}+\left(1-4\delta^2\right)^{m}}{2}\right)^{ \abs{A}+\abs{B} } \right].
\end{align*}

In particular, consider the following upper bound on $f(k)$, the probability that $\abs{A}+\abs{B}\geq k$: setting $s\eqdef \frac{n}{2}$, for $0\leq k \leq s$,

\begin{align*}
    f(k) &=\probaOf{ \abs{A}+\abs{B}\geq k } \\
    &\leq \frac{s! 2^s}{(2s)!} \cdot \binom{s}{k}\frac{(2s-2k)!}{(s-k)! 2^{s-k}}
    = \frac{2^k k!}{(2k)!}\frac{\binom{s}{k}^2}{\binom{2s}{2k}} \\
    &= \frac{2^k}{k!}\frac{\binom{2(s-k)}{s-k}}{\binom{2s}{s}} 
    = \frac{2^k}{k!} \frac{\prod_{j=0}^{k-1} (s-j)^2}{\prod_{j=0}^{2k-1} (2s-j)} \\
    &= \frac{1}{k!} \frac{\prod_{j=0}^{k-1} (s-j)}{\prod_{j=0}^{k-1} (2s-2j-1)}
    \leq \frac{1}{k!}.
\end{align*}
Therefore, for any $z > 1$, we have%\cnote{Alistair's argument instead? A bit shorter.}
\begin{align*}
  \shortexpect_{\lambda,\mu}\left[ z^{ \abs{A}+\abs{B} } \right]
  &= \int_0^\infty \probaOf{ z^{ \abs{A}+\abs{B} } \geq t }dt \\
  &= \int_0^\infty \probaOf{ \abs{A}+\abs{B} \geq \frac{\ln t}{\ln z} }dt \\
  &= 1+\int_1^\infty \probaOf{ \abs{A}+\abs{B} \geq \frac{\ln t}{\ln z} }dt \\
  &\leq 1+\int_1^\infty \probaOf{ \abs{A}+\abs{B} \geq \flr{\frac{\ln t}{\ln z}} }dt \\
  &\leq 1+\int_1^\infty \frac{dt}{\flr{\frac{\ln t}{\ln z}}!}  \tag{from our upper bound on $f(k)$} 
  \leq 1+\int_1^\infty \frac{dt}{\Gamma\left(\frac{\ln t}{\ln z}\right)} \\
  &= 1+\int_1^\infty \frac{e^u du}{\Gamma\left(\frac{u}{\ln z}\right)}\,. 
\end{align*}
Assuming now that $1< z \leq 1+\gamma$ for some $\gamma \in(0,1)$, from $\ln z < \gamma$ and monotonicity of the Gamma function we obtain
\begin{align*}
  \shortexpect_{\lambda,\mu}\left[ z^{ \abs{A}+\abs{B} } \right]
  &= 1+\int_1^\infty \frac{e^u du}{\Gamma\left(\frac{u}{\gamma}\right)}
  = 1+\gamma \int_{1/\gamma}^\infty \frac{e^{\gamma v} dv}{\Gamma(v)} 
  \leq 1+\gamma \int_{0}^\infty \frac{e^{v} dv}{\Gamma(v)} \leq 1+42\gamma\,.
\end{align*}

Suppose now $m \leq c \frac{n}{\eps^2} = \frac{4c}{\delta^2}$, for some constant $c>0$ to be determined later. 
Then, by monotonicity
\begin{align*}
z &\eqdef \frac{\left(1+4\delta^2\right)^{m}+\left(1-4\delta^2\right)^{m}}{2} 
  \leq \frac{\left(1+4\delta^2\right)^{\frac{4c}{\delta^2}}+\left(1-4\delta^2\right)^{\frac{4c}{\delta^2}}}{2} 
  \leq
\frac{e^{16c}+e^{-16c}}{2} < 1+\frac{1}{42\cdot 20} \eqdef 1+\gamma
\end{align*}
for $c < \frac{3}{1000}$. Therefore, 
\[
  \shortexpect_{\lambda,\mu}\left[ (1+\tilde{\tau}(\lambda,\mu))^m \right] - 1 = \shortexpect_{\lambda,\mu}\left[ z^{ \abs{A}+\abs{B} } \right] -1  < \frac{1}{20} \;,
\]
as desired.

\medskip

To conclude, we bound $\shortexpect_{\lambda,\mu}\left[ (1+\tau(\lambda,\mu))^m \right]$ combining the above with the following claim:
\begin{claim}\label{claim:binomial:removing:evenodd:conditioning}
    Let $z\eqdef\frac{1+2\delta}{1-2\delta}$ as above. Then for any two matching-orientation parameters $\lambda,\mu$, we have
    \[
      \prod_{\sigma\colon\abs{\sigma} \geq 4}\shortexpect_{\alpha\sim \mathcal{B}_{\lambda,\mu}(\sigma)}[ z^{\alpha} ]
      \leq e^{\frac{8\eps^4}{n}} \cdot \shortexpect_{\alpha\sim \binomial{\sum_{\sigma\colon \abs{\sigma}\geq 4}\abs{\sigma}}{\frac{1}{2}}}[ z^{\alpha} ] \;.
    \]
    where $\mathcal{B}_{\lambda,\mu}(\sigma)$ is the probability distribution defined earlier.
\end{claim}
\begin{proof}
  Fix $\lambda,\mu$ as in the statement, and any cycle $\sigma$ in the resulting graph. Suppose first this is an ``even'' cycle:
  \begin{align*}
      \shortexpect_{\alpha\sim \mathcal{B}_{\lambda,\mu}(\sigma)}[ z^{\alpha} ]
      &= \shortexpect_{\alpha\sim \binomial{\abs{\sigma}}{\frac{1}{2}}}[ z^{\alpha} \mid \alpha\text{ even} ]
      = \frac{1}{1/2} \sum_{k=0}^{\abs{\sigma}/2} \binom{\abs{\sigma}}{2k} z^{2k} 
      = \frac{(1+z)^{\abs{\sigma}}+(1-z)^{\abs{\sigma}}}{2^{\abs{\sigma}}}.
  \end{align*}
  Similarly, if $\sigma$ is an ``odd'' cycle,
  $
      \shortexpect_{\alpha\sim \mathcal{B}_{\lambda,\mu}(\sigma)}[ z^{\alpha} ]
      = \frac{(1+z)^{\abs{\sigma}}-(1-z)^{\abs{\sigma}}}{2^{\abs{\sigma}}}.
  $
  We then obtain $\shortexpect_{\alpha\sim \mathcal{B}_{\lambda,\mu}(\sigma)}[ z^{\alpha} ] 
  \leq \shortexpect_{\alpha\sim \binomial{\abs{\sigma}}{\frac{1}{2}}}[ z^{\alpha} ]\left( 1+\abs{\frac{1-z}{1+z}}^{\abs{\sigma}} \right)$, from which
  \begin{align*}
      \prod_{\sigma}\shortexpect_{\alpha\sim \mathcal{B}_{\lambda,\mu}(\sigma)}[ z^{\alpha} ] 
      &\leq \prod_{\sigma}\shortexpect_{\alpha\sim \binomial{\abs{\sigma}}{\frac{1}{2}}}[ z^{\alpha} ]\left( 1+\abs{\frac{1-z}{1+z}}^{\abs{\sigma}} \right) \\
      &= \shortexpect_{\alpha\sim \binomial{\sum_{\sigma}\abs{\sigma}}{\frac{1}{2}}}[ z^{\alpha} ] \prod_{\sigma}\left( 1+\abs{\frac{1-z}{1+z}}^{\abs{\sigma}} \right).
  \end{align*}
  We now bound the last factor: since $\abs{\frac{1-z}{1+z}}=2\delta = \frac{2\eps}{\sqrt{n}}$ we have at most $\frac{n}{2}$ cycles, we get  
  \begin{align*}
      \prod_{\sigma\colon \abs{\sigma} \geq 4}\left( 1+\abs{\frac{1-z}{1+z}}^{\abs{\sigma}} \right)
      &= \prod_{\sigma\colon \abs{\sigma} \geq 4}\left( 1+(2\delta)^{\abs{\sigma}} \right)
      \leq \left( 1+16\delta^{4} \right)^{\frac{n}{2}} \leq e^{8\frac{\eps^4}{n}} \;,
  \end{align*}
  as claimed.
\end{proof}

With this result in hand, we can get the conclusion we want: for any $\lambda,\mu$,
\begin{align*}
  1+\tau(\lambda,\mu) 
  &= (1-2\delta)^n z^{\abs{B}} \shortexpect_{\alpha\sim \binomial{\abs{A}}{\frac{1}{2}}}[ z^{2\alpha} ] \prod_{\substack{\sigma\text{ cycle}\\ \abs{\sigma} \geq 4}} \shortexpect_{\alpha\sim \mathcal{B}_{\lambda,\mu}(\sigma)}[ z^{\alpha} ] \\
  &\leq e^{\frac{8\eps^4}{n}}(1-2\delta)^n z^{\abs{B}} \shortexpect_{\alpha\sim \binomial{\abs{A}}{\frac{1}{2}}}[ z^{2\alpha} ] \cdot\shortexpect_{\alpha\sim \binomial{\sum_{\sigma\colon \abs{\sigma}\geq 4}\abs{\sigma}}{\frac{1}{2}}}[ z^{\alpha} ] \tag{by \cref{claim:binomial:removing:evenodd:conditioning}} \\
  &= e^{\frac{8\eps^4}{n}}(1+\tilde{\tau}(\lambda,\mu)) \;,
\end{align*}
from which 
\begin{align*}
  \shortexpect_{\lambda,\mu}\left[ (1+\tau(\lambda,\mu))^m \right]
  &\leq e^{\frac{8\eps^4m}{n}}\shortexpect_{\lambda,\mu}\left[ (1+\tilde{\tau}(\lambda,\mu))^m \right]\\
  &\leq e^{\frac{8\eps^4m}{n}}\left(1+\frac{1}{20}\right) \tag{for $m\leq \frac{cn}{\eps^2}$}\\
  &\leq e^{8c\eps^2}\frac{21}{20} < 1+\frac{1}{10} \tag{as $c<\frac{3}{1000}$ and $\eps \leq 1$} \;,
\end{align*}
concluding the proof: by \eqref{eq:minimax:mixture:bound},
$
  \normone{Q-U^{\otimes m}} \leq \frac{1}{10} \;,
$
for any $m < c\frac{n}{\eps^2}$.

%%%%%%%%%%%%%%%%%%%%%%%%%%%%%%%%%%%%%%%%%%%%%%%%%%%%%%%%%%%%%%%%%%%%%%%%%%%%%%%%%%%%%%%%%%%%%%%%%%%%%
%%%%%%%%%%%%%%%%%%%%%%%%%%%%%%%%%%%%%%%%%%%%%%%%%%%%%%%%%%%%%%%%%%%%%%%%%%%%%%%%%%%%%%%%%%%%%%%%%%%%%
%%%%%%%%%%%%%%%%%%%%%%%%%%%%%%%%%%%%%%%%%%%%%%%%%%%%%%%%%%%%%%%%%%%%%%%%%%%%%%%%%%%%%%%%%%%%%%%%%%%%%

\subsection{Identity Testing Algorithm against Non-Degenerate Bayes Nets}  \label{ssec:identity-unknown-upper}

We start with the case of trees and then generalize to bounded degree.

\subsubsection{The Case of Trees}\label{ssec:identity-unknown-upper:trees}
In this section, we prove our result on testing identity of a tree structured 
Bayes net with unknown topology. {Recall from~\cref{ssec:identity-unknown-lower} that a Bayes net with tree structure is said to be $c$-balanced if it satisfies $p_k\in[c,1-c]$ for all $k$. We will require the following definition of \emph{non-degeneracy} of a tree, which will be a simpler case of the definition we shall have for general Bayes nets (\cref{def:non:degeneracy}):

\begin{definition}\label{def:nondegenerate:tree}
  For any $\gamma \in(0,1]$, we say a tree Bayes net $P$ over $\{0,1\}^n$ is \emph{$\gamma$-non-degenerate} if for all $i\in[n]$,
  \[
      \abs{ \probaCond{ X_i = 1 }{ X_{\parent{i}} = 1 } - \probaCond{ X_i = 1 }{ X_{\parent{i}} = 0 } } \geq \gamma
  \]
  where $X\sim P$.
\end{definition}
Roughly speaking, this definition states that the choice of the value of its parent has a significant influence on the probability of any node. With these definitions, we are ready to state and prove our result:}

\begin{restatable}{theorem}{identityunknowntreeub}\label{theo:upper:unknowntree:identity}
There exists an efficient algorithm with the following guarantees. 
Given as input {(i)} a tree $\structure$ over $n$ nodes and an {explicit $c$-balanced, $\gamma$-non-degenerate Bayes net} $Q$ 
with structure $\structure$,  where $c, \gamma = \bigOmega{1/n^a}$ for some absolute constant $a>0$; 
{(ii)} parameter $\eps > 0$, and {(iii)} sample access to a Bayes net $P$ 
with unknown tree structure, the algorithm takes $\bigO{\sqrt{n}/\eps^2}$ samples from $P$, 
and distinguishes with probability at least $2/3$ between $P=Q$ and $\normone{P-Q} > \eps$.
\end{restatable}

The algorithm follows a natural idea: {(1)} check that the unknown distribution $P$ indeed has, as it should, the same tree structure as the (known) distribution $Q$; {(2)} if so, invoke the algorithm of the previous section, which works under the assumption that $P$ and $Q$ have the same structure.

Therefore, to establish the theorem it is sufficient to show that {(1)} can be performed efficiently. Specifically, we will prove the following:
\begin{theorem}\label{theo:check:tree:structure} 
There exists an algorithm with the following guarantees. 
For $\gamma\in(0,1)$, $c\in(0,1/2)$, it takes as input
an explicit $c$-balanced, $\gamma$-nondegenerate tree Bayes net $Q$ over $\{0,1\}^n$
with structure $\structure(Q)$, and 
  \[     O\left( \frac{\log^2\frac{1}{c}}{c^6\gamma^4} \log n \right) \]
samples from an arbitrary tree Bayes net $P$ over $\{0,1\}^n$ with unknown structure $\structure(P)$.
  \begin{itemize}
    \item If $P=Q$, the algorithm returns $\textsf{accept}$ with probability at least $4/5$;
    \item If $\structure(P)\neq\structure(Q)$, the algorithm returns $\textsf{reject}$ with probability at least $4/5$.
  \end{itemize}
\end{theorem}

\noindent Note that the above theorem implies the desired result 
as long as $\frac{\log^2\frac{1}{c}}{c^6\gamma^4} = \bigO{\frac{\sqrt{n}}{\eps^2 {\log n}}}$. 
%(In particular, for $c,\gamma=\bigOmega{1}$.)

\begin{proof}[Proof of~\cref{theo:check:tree:structure}]
We start by stating and proving lemmas that will be crucial in stating and analyzing the algorithm:
\begin{fact}\label{fact:unknown:tree:structure:estimation}
Given $\tau > 0$ and sample access to a tree Bayes net $P$ over $\{0,1\}^n$, one can obtain with $O(\frac{\log n}{\tau^2})$ samples estimates $(\hat{\mu}_i)_{i\in[n]}$, $(\hat{\rho}_{i,j})_{i,j\in[n]}$ such that, with probability at least $9/10$, 
\[
      \max\left( \max_{i\in[n]} \lvert \hat{\mu}_i - \expect{X_i} \ \rvert, \max_{i,\in[n]} \lvert \hat{\rho}_{i,j}-\expect{X_iX_j} \rvert  \right) \leq \tau.
\]
\end{fact}
\begin{proof}
The fact follows immediately by an application of Chernoff bounds.
\end{proof}

\begin{lemma}\label{lemma:unknown:tree:structure:lipschitz:mutualinfo}
Let $c\in(0,1/2]$. There exists a constant $\lambda$ and a function $f$ such that
  \[
      \mutualinfo{X_i}{X_j} = f( \expect{X_i},\expect{X_j},\expect{X_iX_j}  ) \;,
  \]
for any $c$-balanced tree Bayes net $P$ over $\{0,1\}^n$ and $X\sim P$, 
where $f$ is $\lambda$-Lipschitz with respect to the $\norminf{\cdot}$ norm on the domain $\Omega_c\subseteq [0,1]\times[0,1]\times[0,1]\to[0,1]$ in which $(\expect{X_i},\expect{X_j},\expect{X_iX_j})_{i,j}$ then take values. Moreover, one can take $\lambda = 16\log\frac{1}{c}$.
\end{lemma} 

\begin{proof}[Proof Sketch:]
{Expanding the definition of mutual influence $\mutualinfo{X}{Y}$ of two random variables, 
it is not hard to write it as a function of $\expect{X},\expect{Y}$, and $\expect{XY}$ only.
This function would not be Lipschitz on its entire domain, however. The core of the proof 
leverages the balancedness assumption to restrict its domain to a convenient subset 
$\Omega_c\subseteq [0,1]\times[0,1]\times[0,1]$, 
on which it becomes possible to bound the partial derivatives of $f$. 
We defer the details of the proof to~\cref{sec:misc:proofs}.}
\end{proof}

{We now show the following crucial lemma establishing the following result: For any balanced Bayes net, 
the shared information between any pair of non-adjacent vertices $i, j$ 
is noticeably smaller than the minimum shared information 
between any pair of neighbors along the path that connects $i, j$.}

\begin{lemma}\label{lemma:unknown:tree:structure:gap:mutualinfo}
Let $c\in(0,1/2]$, and fix any $c$-balanced tree Bayes net $P$ over $\{0,1\}^n$ 
with structure $\structure(Q)$. Then, for any distinct $i,j\in[n]$ such that $i\neq\parent{j}$ and $j\neq\parent{i}$, we have
\[
      \mutualinfo{X_i}{X_j} \leq (1-2c^2)\min_{\{k,\parent{k}\}\in\operatorname{path}(i,j)} \mutualinfo{X_k}{X_{\parent{k}}} \;,
\]
where $X\sim P$. (and $\operatorname{path}(i,j)$ is a path between $i$ to $j$, of the form $i - \dots - k - \dots - j$, where each edge is of the form $(k,\parent{k}$  or $(\parent{k},k)$).
\end{lemma}
\begin{proof}
By induction and the data processing inequality, it is sufficient to prove the statement for a path of length 3, namely
\[
      X_i - X_k - X_j \;.
\]
The result will follow from a version of the strong data processing inequality (see e.g.,~\cite{PW:15}, 
from which we borrow the notation $\eta_{\rm KL},\eta_{\rm TV}$): since $X_i \to X_k \to X_j$ forms a chain with the Markov property, we get
$\mutualinfo{X_i}{X_j} \leq \eta_{\rm KL}(P_{X_j\mid X_k}) \mutualinfo{X_i}{X_k}$ from~\cite[Equation 17]{PW:15}. Now, by~\cite[Theorem 1]{PW:15}, we have
\[
    \eta_{\rm KL}(P_{X_j\mid X_k}) \leq \eta_{\rm TV}(P_{X_j\mid X_k}) = \dtv( P_{X_j\mid X_k=0}, P_{X_j\mid X_k=1} ).
\]
If $k=\parent{j}$ (in our Bayes net), then $\dtv( P_{X_j\mid X_k=0}, P_{X_j\mid X_k=1} ) = \abs{p_{j,0}-p_{j,1}} \leq 1-2c$ from the $c$-balancedness assumption. On the other hand, if $j=\parent{k}$, then by Bayes' rule it is easy to check that (again, from the $c$-balancedness assumption) $\probaCond{X_{\parent{k}}=1}{X_k=a}\in[c^2,1-c^2]$, and
\[
\dtv( P_{X_j\mid X_k=0}, P_{X_j\mid X_k=1} ) = \abs{ \probaCond{X_j=1}{X_k=0} - \probaCond{X_j=1}{X_k=1} } \leq 1-2c^2\,.
\]

Therefore, we get $\mutualinfo{X_i}{X_j} \leq (1-2c^2) \mutualinfo{X_i}{X_k}$ as wanted; by symmetry, $\mutualinfo{X_i}{X_j} \leq (1-2c^2) \mutualinfo{X_j}{X_k}$ holds as well.
\end{proof}

\begin{lemma}\label{lemma:unknown:tree:structure:min:mutualinfo}
Let $c\in(0,1/2], \gamma\in(0,1)$, and fix any $c$-balanced, $\gamma$-nondegenerate tree Bayes net $P$ over $\{0,1\}^n$, with structure $\structure(P)$. Then, there exists an absolute constant $\kappa$ such that for any $i\in[n]$ one has
\[
      \mutualinfo{X_i}{X_{\parent{i}}} \geq \kappa \;,
\]
where $X\sim P$. (Moreover, one can take $\kappa = \frac{c \gamma^2}{2\ln 2}$.)
\end{lemma}
\begin{proof}
Fix any such $i$, and write $X=X_i$, $Y=X_{\parent{i}}$ for convenience; and set $u \eqdef \probaOf{X=1}$, 
$v\eqdef \probaOf{Y=1}$, $a\eqdef \probaCond{X=1}{Y=1}$, and ${b}\eqdef \probaCond{X=1}{Y=0}$. We then have
\begin{align*}
  \mutualinfo{X}{Y} 
  &= \sum_{x,y\in\{0,1\}}\!\!\! \probaOf{X=x,Y=y} \log \frac{\probaOf{X=x,Y=y}}{\probaOf{X=x}\probaOf{Y=y}} \\
  &= \sum_{x,y\in\{0,1\}}\!\!\! \probaCond{X=x}{Y=y}\probaOf{Y=y}\cdot \log \frac{\probaCond{X=x}{Y=y}}{\probaOf{X=x}} \\
  &= (1-u)(1-b)\log\frac{1-b}{1-v} + (1-u)b\log\frac{b}{v} + u(1-a)\log\frac{1-a}{1-v} + ua\log\frac{a}{v} \\
  &= u\varphi(a,v)+(1-u) \varphi(b,v) \;,
\end{align*}
where $\varphi(x,y) \eqdef x\log\frac{x}{y}+(1-x)\log\frac{1-x}{1-y} \geq 0$ for $x,y\in[0,1]$ is the KL-divergence between 
two Bernoulli distributions with parameters $x,y$. 
From our assumptions of $c$-balanced and $\gamma$-nondegeneracy, we know that $u,v,a,b$ satisfy
\begin{align*}
  c&\leq a,b,u,v\leq 1-c \\
  \gamma &\leq \abs{a-b} \;,
\end{align*}
which leads to, noticing that $\abs{a-b}\geq \gamma$ implies that at least one of 
$\abs{a-v}\geq \frac{\gamma}{2}$, $\abs{b-v}\geq \frac{\gamma}{2}$ holds and that $\varphi(\cdot,v)$ is convex with a minimum at $v$:
\begin{align*}
  \mutualinfo{X}{Y}
  \geq c\left( \varphi(a,v)+\varphi(b,v) \right)
  \geq c\min\left( \varphi\left(v-\frac{\gamma}{2},v\right),\varphi\left(v+\frac{\gamma}{2},v\right) \right)
  \geq \frac{1}{2\ln 2} c\gamma^2 \;,
\end{align*}
using the standard lower bound of $\varphi(x,y)\geq\frac{2}{\ln 2}(x-y)^2$ on the KL-divergence.
\end{proof}

\paragraph{The Algorithm} With these in hand, we are ready to describe and analyze the algorithm underlying~\cref{theo:check:tree:structure}:

Let $\gamma\in(0,1)$, $c\in(0,1/2)$ be fixed constants, and $Q$ be a known $c$-balanced, $\gamma$-nondegenerate tree Bayes net over $\{0,1\}^n$, with structure $\structure(Q)$. Furthermore, let $P$ be an unknown tree Bayes net over $\{0,1\}^n$ with unknown structure $\structure(P)$, to which we have sample access.

Let $\kappa=\kappa(c,\gamma) = \frac{c\gamma^2}{2\ln 2}$ as in~\cref{lemma:unknown:tree:structure:min:mutualinfo}, $c'\eqdef\frac{c}{2}$, and $\lambda=\lambda(c')=16\log\frac{2}{c}$ as in~\cref{lemma:unknown:tree:structure:lipschitz:mutualinfo}. In view of applying~\cref{lemma:unknown:tree:structure:gap:mutualinfo} later to $P$, set
\[
    \tau \eqdef \frac{\kappa - (1-2{c'}^2)\kappa}{4\lambda} = \frac{1}{64\ln 2} \frac{c^3\gamma^2}{\log\frac{2}{c}}.
\]

The algorithm then proceeds as follows. (Below, $X$ denotes a random variable distributed according to $P$.)
\begin{enumerate}
  \item Take $m = O\left( \frac{\log n}{\tau^2} \right)$ samples from $P$, and use them to
    \begin{itemize}
      \item Estimate all $n^2$ marginals $\probaCond{X_i=1}{X_j=a}$, and verify that they are all in $[c', 1-c']$ (ensuring that $P$ is $c'$-balanced), with probability at least $9/10$. Else, return \textsf{reject};
      \item Estimate all $\binom{n}{2}+n$ values of $\expect{X_i}$ and $\expect{X_iX_j}$ to an additive $\tau$, with probability at least $9/10$, as in~\cref{fact:unknown:tree:structure:estimation}. (Call these estimates $\hat{\mu}_i$, $\hat{\rho}_{i,j}$.)
    \end{itemize}
    At the end of this step, we are guaranteed that $P$ is $c'$-balanced (or else we have rejected with probability at least $9/10$).
  \item Check that all $\hat{\mu}_i$, $\hat{\rho}_{i,j}$'s are all within an additive $\tau$ of what they should be under $Q$. If so, return \textsf{accept}, else return \textsf{reject}.
\end{enumerate}

Clearly, the algorithm only uses $O\left( \frac{\log^2\frac{1}{c}}{c^6\gamma^4} \log n \right)$ samples from $P$. We now establish its correctness: first, with probability at least $4/5$ by a union bound, all estimates performed in the first step are correct; we condition on that.
\begin{description}
  \item[Completeness.] If $P=Q$, then $P$ is $c$-balanced, and thus \textit{a fortiori} $c'$-balanced: the algorithm does not reject in the first step. Moreover, clearly all $(\hat{\mu}_i)_i$, $(\hat{\rho}_{i,j})_{i,j}$ are then within an additive $\tau$ of the corresponding values of $P=Q$, so the algorithm returns \textsf{accept}.
  \item[Soundness.] By contrapositive. If the algorithm returns \textsf{accept}, then $P$ is $c'$-balanced by the first step. Given our setting of $\tau$, by~\cref{lemma:unknown:tree:structure:lipschitz:mutualinfo} our estimates $(\hat{\mu}_i)_i$, $(\hat{\rho}_{i,j})_{i,j}$ are such that all corresponding quantities
  \[
      \hat{I}_{i,j}\eqdef f(\hat{\mu}_i,\hat{\mu}_j,\hat{\rho}_{i,j})
  \]
  are within $\tau\lambda = \frac{\kappa - (1-2{c'}^2)\kappa}{4}$ of the mutual informations $\mutualinfo{X_i}{X_j}$ for $P$. But then, by~\cref{lemma:unknown:tree:structure:gap:mutualinfo} this implies that the relative \emph{order} of all $\hat{I}_{i,j},\hat{I}_{i',j'}$ is the same as that of $\mutualinfo{X_i}{X_j},\mutualinfo{X_{i'}}{X_{j'}}$. This itself implies that running the Chow--Liu algorithm on input these $\hat{I}_{i,j}$'s would yield the same, 
uniquely determined tree structure $\structure(P)$ as if running it on the actual $\mutualinfo{X_i}{X_j}$'s. 
{To see this, we note that the Chow-Liu algorithm works 
by computing a maximum-weight spanning tree (MST)  with respect to the weights given by the pairwise mutual information. 
The claim follows from the fact that the MST only depends on the relative ordering of the edge-weights.}
  
  But since the $(\hat{\mu}_i)_i$, $(\hat{\rho}_{i,j})_{i,j}$ are also within an additive $\tau$ of the corresponding quantities for $Q$ (per our check in the second step), the same argument shows that running the Chow--Liu algorithm would result in the same, uniquely determined tree structure $\structure(Q)$ as if running it on the actual mutual informations from $Q$. Therefore, $\structure(P)=\structure(Q)$, concluding the proof.
\end{description}
\end{proof}

%%%%%%%%%%%%%%%%%%%%%%%%%%%%%%%%%%%%%%%%%%%%%%%%%%%%%%%%%%%%%%%%%%%%%%%%%%%%%%%%%%%%%%%%%%%%%%%%%%%%%%%%%%%%%%%%%%%%%%%%%%%%%%%%%%%%%%%%%%
%%%%%%%%%%%%%%%%%%%%%%%%%%%%%%%%%%%%%%%%%%%%%%%%%%%%%%%%%%%%%%%%%%%%%%%%%%%%%%%%%%%%%%%%%%%%%%%%%%%%%%%%%%%%%%%%%%%%%%%%%%%%%%%%%%%%%%%%%%
%%%%%%%%%%%%%%%%%%%%%%%%%%%%%%%%%%%%%%%%%%%%%%%%%%%%%%%%%%%%%%%%%%%%%%%%%%%%%%%%%%%%%%%%%%%%%%%%%%%%%%%%%%%%%%%%%%%%%%%%%%%%%%%%%%%%%%%%%%

\subsubsection{The Case of Bounded Degree}\label{ssec:identity-unknown-upper:degreed}

In this section, we show how to test identity of unknown structure Bayesian networks with maximum in-degree $d$ 
under some non-degeneracy conditions. Intuitively, we want these conditions to ensure \emph{identifiability} of the structure: 
that is, that any (unknown) Bayes net close to a non-degenerate Bayes net $Q$ must also share the same structure. 
To capture this notion, observe that, by definition, non-equivalent Bayes net structures 
satisfy different conditional independence constraints: 
our non-degeneracy condition is then to rule out some of these possible 
new conditional independence constraints, 
as \emph{far} from being satisfied by the non-degenerate Bayes net.
Formally, we have the following definition:

\begin{definition}[Non-degeneracy]\label{def:non:degeneracy}
For nodes $X_i$, $X_j$, set of nodes $S$, and a distribution $P$ over $\{0,1\}^n$, we say that $X_i$ and $X_j$ are \emph{$\gamma$-far from independent conditioned on $X_S$} if for all distributions $Q$ over $\{0,1\}^n$ such that $\dtv(P,Q) < \gamma$, it holds that $X_i$ and  $X_j$ are not independent conditioned on $X_S$.

A Bayes net $P$ is then called \emph{$\gamma$-non-degenerate with respect to structure $\structure$ and degree $d$}
if for any nodes $X_i$, $X_j$ and set of nodes $S$ of size $|S| \leq d$ not containing $i$ or $j$ satisfying one of the following:
\begin{itemize}
\item[(i)] $X_i$ is a parent of $X_j$,
\item[(ii)] $S$ contains a node $X_k$ that is a child of both $X_i$ and $X_j$,
\item[(iii)] $X_i$ is a grandparent of $X_j$ and there is a child of $X_i$ and parent of $X_j$, $X_k$, that is not in $S$,
\item[(iv)] $X_i$ and $X_j$ have a common parent $X_k$ that is not in $S$
\end{itemize}
we have that $X_i$ and $X_j$ are $\gamma$-far from independent conditioned on $X_S$ (where all relations are under structure $\structure$).
\end{definition}

\tikzset{
  treenode/.style = {align=center, inner sep=0pt, text centered,
    font=\sffamily},
  blacknode/.style = {treenode, circle, white, font=\sffamily\bfseries, draw=black,
    fill=black, text width=1.5em},% arbre rouge noir, noeud noir
  whitenode/.style = {treenode, circle, black, draw=black, 
    text width=1.5em, very thick},
  rednode/.style = {treenode, circle, white, font=\sffamily\bfseries, draw=red,
    fill=red, text width=1.5em}
}

\begin{figure}[ht]\centering
  \begin{tikzpicture}[->,>=stealth',level/.style={sibling distance = 1cm/#1,  level distance = 1.5cm}] 
  \node[whitenode]{$X_j$}
      child[<-]{ node[whitenode] {$X_i$} 	}
  ; 
  \end{tikzpicture}
  \hspace{.1\textwidth}
  \begin{tikzpicture}[->,>=stealth',level/.style={sibling distance = 1cm/#1,  level distance = 1.5cm}] 
  \node[blacknode] {$X_k$}
      child[<-]{ node[whitenode] {$X_i$} 	}
      child[<-]{ node[whitenode] {$X_j$} 	}
  ; 
  \end{tikzpicture}
  \hspace{.1\textwidth}
  \begin{tikzpicture}[->,>=stealth',level/.style={sibling distance = 1cm/#1,  level distance = 1.5cm}] 
  \node[whitenode] {$X_j$}
      child[<-]{ node[rednode] {$X_k$}
          child[<-]{ node[whitenode] {$X_i$} }
      }
  ; 
  \end{tikzpicture}
  \hspace{.1\textwidth}
  \begin{tikzpicture}[->,>=stealth',level/.style={sibling distance = 1cm/#1,  level distance = 1.5cm}] 
  \node[rednode] {$X_k$}
      child{ node[whitenode] {$X_i$} 	}
      child{ node[whitenode] {$X_j$} 	}
  ; 
  \end{tikzpicture}
  \caption{The four possible conditions of~\cref{def:non:degeneracy}, from left (i) to right (iv). 
  The black nodes are the ones in $S$, the red ones (besides $X_i,X_j$) are not in $S$.}
\end{figure}

We shall also require some terminology: namely, the definition of the \emph{skeleton} of a Bayesian network as the underlying undirected graph of its structure. We can now state the main result of this section:
\begin{theorem} \label{thm:unknown-structure-identity}
There exists an algorithm with the following guarantees. Given the full description of a Bayes net $Q$ of degree at most $d$ which is $(c,C)$ balanced and $\gamma$-non-degenerate for $c=\tildeOmega{1/\sqrt{n}}$ and $C=\tildeOmega{d\eps^2/\sqrt{n}}$, parameter $\eps\in(0,1]$, and sample access to a distribution  $P$, promised to be a Bayes net of degree at most $d$ whose skeleton has no more edges than $Q$'s, the algorithm takes $\bigO{2^{d/2}\sqrt{n}/\eps^2+(2^d+ d \log n)/\gamma^2}$ samples from $P$, runs in time $\bigO{n}^{d+3}(1/\gamma^2+1/\eps^2)$, and distinguishes with probability at least $2/3$ between (i) $P=Q$ and (ii) $\normone{P-Q}  > \eps$.
\end{theorem}

In~\cref{lem:struct-equiv}, we show that these non-degeneracy conditions are enough to ensure identifiability of the structure, up to equivalence. In~\cref{prop:conditional:independence:tester}, we give a test for conditional independence specialized to Bernoulli random variables. 
In the last part, we provide a test for showing whether a non-degenerate Bayes net has a given structure using this conditional independence test, establishing~\cref{prop:struct-test}. We then can combine this structure test with our test for Bayes nets with known structure to obtain~\cref{thm:unknown-structure-identity}. This structure tester, which may be of independent interest, has the following guarantees,

\begin{theorem}\label{prop:struct-test}
Let $\structure$ be a structure of degree at most $d$ and $P$ be a Bayesian network with structure $\structure'$ that also has degree at most $d$ and whose skeleton has no more edges than $\structure$. Suppose that $P$ either (i) can be expressed as a Bayesian network with structure $\structure$ that is $\gamma$-non-degenerate with degree $d$; or (ii) cannot be expressed as a Bayesian network with structure $\structure$. Then there is an algorithm which can decide which case holds with probability $99/100$, given $\structure$, $\gamma$, and sample access to $P$. The algorithm takes $\bigO{(2^d+ d \log n)/\gamma^2}$ samples and runs in time $\bigO{n^{d+3}/\gamma^2}$.
\end{theorem}

\noindent Using the above theorem, we can prove the main result of this section:
\begin{proof}[Proof of~\cref{thm:unknown-structure-identity}]
We first invoke the structure test given in~\cref{algo:bn:structure-test}. If it accepts, we run the known structure test given in~\cref{theo:upper:knowndegreed:identity}. We accept only if both accept.

The correctness and sample complexity now both follow from~\cref{prop:struct-test} and~\cref{theo:upper:knowndegreed:identity}. 
Specifically, if the structure test accepts, then with high probability, 
we have that $Q$ can be expressed as a Bayes net with the same structure as $P$, 
and thus we have the pre-conditions for the known structure test. If either test rejects, then $P \neq Q$.
\end{proof}

%%%%%%%%%%%%%%%%%%%%%%%%%%%%%%%%%%%%%%%%%%%%%%%%%%%%%%%%%%%%%%%%%%%%%%%%%%%%%%%%

\paragraph{Non-degeneracy and Equivalent Structures} \label{ssec:non-degeneracy}

The motivation behind the $\gamma$-non-degeneracy condition is the following: 
if $Q$ is $\gamma$-non-degenerate, then for any Bayesian network $P$ with degree at most $d$ 
that has $\dtv(P,Q) < \gamma$ we will argue that $P$ can be described using the same structure $\structure$ 
as we are given for $Q$. Indeed, the structure $\structure'$ of $P$ will have the property that $\structure$ and $\structure'$ 
both can be used to describe the same Bayesian networks, a property known as \emph{$I$-equivalence}. 
It will then remain to make this algorithmic, that is to describe how to decide whether $P$ can be described 
with the same structure as $Q$ or whether  $\dtv(P,Q) \geq \gamma$. 
Assuming we have this decision procedure, then if the former case happens to hold we can invoke our existing known-structure tester (or reject if the latter case holds).

We will require for our proofs the following definition:
\begin{definition}[$\lor$-structure] 
For a structure $\structure$, a triple $(i,j,k)$ is a \emph{$\lor$-structure} (also known as an \emph{immorality})
if $i$ and $j$ are parents of $k$ but neither $i$ nor $j$ is a parent of the other.
\end{definition} 
\noindent The following result, due to Verma and Pearl~\cite{VermaPearl:90}, will play a key role:
\begin{lemma}\label{lem:Verma-Pearl} 
Two structures $\structure$ and $\structure'$ are $I$-equivalent if and only if 
they have the same skeleton and the same $\lor$-structures.
\end{lemma}

Note that, for general structures $\structure$, $\structure'$, it may be possible to represent all Bayesian networks with structure $\structure$ as ones with structure $\structure'$, but not vice versa.  Indeed, this can easily be achieved by adding edges to $\structure$ to any node (if any) with less than $d$ parents.  
This is the rationale for the assumption in~\cref{prop:struct-test} that $\structure'$ has no more edges than $\structure$: as this assumption is then required for $\structure$ and $\structure'$ to be $I$-equivalent unless $\dtv(P,Q) \geq \gamma$.

We now prove that any Bayesian network $Q$ satisfiying the conditions of~\cref{prop:struct-test} and being non-degenerate with respect to a structure can in fact be expressed as having that structure.

\begin{lemma} \label{lem:struct-equiv}
Fix $\gamma > 0$. If $Q$ is a Bayesian network with structure $\structure'$ of degree at most $d$ that is $\gamma$-non-degenerate with respect to a structure $\structure$ with degree at most $d$ and $\structure'$ has no more edges than $\structure$, then $\structure$ and $\structure'$ are $I$-equivalent.
\end{lemma}

Note that $Q$ being $\gamma$-non-degenerate for some $\gamma > 0$ is equivalent to a set of conditional independence conditions all being false, since if $X_i$ and $X_j$ are not conditionally independent with respect to $X_S$, then there is a configuration $a$ such that $\Pr_Q[X_S =a] >0$ and $I(X_i;X_j \mid X_S=a) \geq 0$.
\begin{proof}
We first show that $\structure$ and $\structure'$ have the same skeleton and then that they have the same $\lor$-structures. We need the following:
\begin{claim}  \label{clm:super-parents}
Let $S$ be the set of parents of $X_i$ in a Bayesian network  $Q$ with structure $\structure$. Let $X_j$ be a node that is neither in $S$ nor a descendant of $X_i$. Then $X_i$ and $X_j$ are independent conditioned on $X_S$.
\end{claim}
\begin{proof}%\cnote{Proof may be omitted if necessary?}
Firstly, we note that there is a numbering of the nodes which is consistent with the DAG of $\structure$ such that any $j \in S$ has $j < i$. Explicitly, we can move $X_i$ and all its descendants to the end of the list of nodes to obtain this numbering.

Letting $D \eqdef \{1,\dots, i-1\}$, we have that, from the definition of Bayesian networks, $\Pr_Q[ X_i=1 \mid X_{D}=b ] = \Pr_Q[ X_i=1 \mid X_S=b_S ]$ for all configurations $b$ of $D$. Then for any configuration $a$ of $S'\eqdef S\cup\{j\}$, we have
\begin{align*}
\Pr_P[X_j=1 \mid X_{S'}=a]  
&=  \sum_{b:b_S=a} \Pr_P[X_j=1 \mid X_D = b] \Pr_P[X_D=b \mid X_{S'}=a] \\
&= \Pr_P[X_j=1 \mid X_S=a_S]  \sum_{b:b_S=a} \Pr_P[X_D=b \mid X_{S'}=a]\\
&= \Pr_P[X_j=1 \mid X_S=a_S] 
\end{align*}
concluding the proof.
\end{proof} 

Suppose for a contradiction that $(i,j)$ is an edge in the skeleton of $\structure$ but not in $\structure'$. 
Without loss of generality, we may assume that $X_j$ is not a descendant of $X_i$ in $\structure'$ (since otherwise we can swap the roles of $i$ and $j$ in the argument). Then as $X_i$ is not in $S$, the set of parents of $X_j$ in $\structure'$, either, by~\cref{clm:super-parents} $X_i$ and $X_j$ are independent conditioned on $X_S$. However since one of $X_i$ and $X_j$ is a parent of the other in $\structure$, condition (i) of $\gamma$-non-degeneracy gives that $X_i$ and $X_j$ are $\gamma$-far from independent conditioned on $X_S$. This is a contradiction, so all edges in the skeleton of $\structure$ must be edges of $\structure'$. But by assumption $\structure'$ has no more edges than $\structure$, and so they have the same skeleton.

Next we show that $i$ and $j$ have the same $\lor$-structures. Assume by contradiction that $(i,j,k)$ is a $\lor$-structure in $\structure$ but not $\structure'$. Since $\structure$ and $\structure'$ have the same skeleton, this cannot be because $X_i$ is the parent of $X_j$ or vice versa. Therefore, must be that at least one of $X_i$ or $X_j$ is the child of $X_k$ rather than its parent in $\structure'$.
As before, without loss of generality we may assume that $X_j$ is not a descendant of $X_i$ in $\structure'$. This implies that $X_k$ cannot be a child of $X_i$, since then $X_j$ must be a child of $X_k$  and so a descendant of $X_i$. Thus $S$, the set of parents of $X_i$ in $\structure'$, contains $X_k$ but not $X_j$; and~\cref{clm:super-parents} then implies that $X_i$ and $X_j$ are independent conditioned on $X_S$. However, in $\structure$ $X_k$ is the child of both $X_i$ and $X_j$ and so by condition (ii) of $\gamma$-non-degeneracy, we have that $X_i$ and $X_j$ are $\gamma$-far from independent conditioned on $X_S$. This contradiction shows that all $\lor$-structures in $\structure$ are $\lor$-structures in $\structure'$ as well.

Finally, we assume for the sake of contradiction that $(i,j,k)$ is a $\lor$-structure in $\structure'$ but not $\structure$.
Again without loss of generality, we assume that $X_j$ is not a descendant of $X_i$ in $\structure'$; and let $S$ be the parents of $X_i$ in $\structure'$. Note that neither $X_k$ nor $X_j$ is in $S$ since this is a $\lor$-structure.
Now by~\cref{clm:super-parents}, $X_i$ and $X_j$ are independent conditioned on $X_S$. In $\structure$, however, $(i,j,k)$ is not a $\lor$-structure yet $(i,k)$, $(j,k)$  (but not $(i,j)$) are in the skeleton of $\structure$.
Thus at least one of $X_i$, $X_j$ is a child of $X_k$. If only one is a child, then the other must be $X_k$'s parent. In the case of two children, we apply condition (iv) and in the case of a parent and a child, we apply condition (iii) of $\gamma$-non-degeneracy.
Either way, we obtain that, since $X_k$ is not in $S$, $X_i$ and $X_j$ are $\gamma$-far from independent conditioned on $X_S$. This contradiction shows that all $\lor$-structures in $\structure'$ are also $\lor$-structures in $\structure.$

We thus have all the conditions for~\cref{lem:Verma-Pearl} to apply and conclude that $\structure$ and $\structure'$ are $I$-equivalent.
\end{proof}

%%%%%%%%%%%%%%%%%%%%%%%%%%%%%%%%%%%%%%%%%%%%%%%%%%%%%%%%%%%%%%%%%%%%%%%%%%%%%%%%

\paragraph{Conditional Independence Tester} \label{ssec:cond-ind-test}
We now turn to establishing the following proposition:
\begin{proposition}\label{prop:conditional:independence:tester}
  There exists an algorithm that, given parameters $\gamma, \tau > 0$, set of coordinates $S\subseteq [n]$ and coordinates $i,j\in[n]\setminus S$, as well as sample access to a distribution $P$ over $\{0,1\}^n$, satisfies the following. With probability at least $1-\tau$, the algorithm accepts when $X_i$ and $X_j$ are independent conditioned on $X_S$ and rejects when no distribution $Q$ with $\dtv(P,Q) < \gamma$ has this property (and may do either if neither cases holds). Further, the algorithm takes $O((2^d + \log(1/\tau))/\gamma^2)$ samples from $P$ and runs in time $O((2^d + \log(1/\tau))/\gamma^2)$.
\end{proposition}
\begin{figure}[h!]\small
  \begin{framed}
    \begin{description}
      \item[Input] $\gamma, \tau > 0$, $i,j \in \{0,1\}^n$, $S \subseteq \{0,1\}^n$ with $i,j \notin S$, and sample access to a distribution $P$ on $\{0,1\}^n$. 
      \item[-] Take $O((2^d + \log(1/\tau))/\gamma^2)$ samples from $P$. Let $\tilde P$ be the resulting empirical distribution.
      \item[For each] configuration $a \in \{0,1\}^{|S|}$ of $S$,
      \begin{description}
        \item[-] Compute the empirical conditional means $\mu_{i,a} = \E_{X \sim \tilde P}[X_i \mid X_S=a]$ and $\mu_{j,a} = \E_{X \sim \tilde P}[X_j\mid X_S=a]$.
        \item[-] Compute the conditional covariance $\Cov_{\tilde P}[X_i,X_j \mid X_S=a]=\E_{X \sim \tilde P}[(X_i-\mu_{i,a})(X_j - \mu_{j,a}) \mid X_S=a]$.
      \end{description}
    \item Compute the expected absolute value of the conditional covariance $\beta = \E_{Y \sim \tilde P}[|\Cov_{\tilde P}[X_i,X_j \mid X_S=Y_S]|]$.
    \item[If] $\beta \leq \gamma/3$,  return $\textsf{accept}$
    \item[Else]  return $\textsf{reject}.$
    \end{description}
  \end{framed}
  \caption{Testing whether $X_i$ and $X_j$ are independent conditioned on $S$ or are $\gamma$-far from being so.} \label{algo:cond:independ}
\end{figure}

\begin{proof} The algorithm is given in~\cref{algo:cond:independ}. Its sample complexity is immediate, and that the algorithm takes linear time in the number of samples is not hard to see. It remains to prove correctness.

To do so, define $D \eqdef S \cup \{i,j\}$. Let $P_D,\tilde P_D$ be the distributions of $X_S$ for $X$ distributed as $P,\tilde P$ respectively. Since $P_D$ is a discrete distribution with support size $2^{d+2}$, by standard results the empirical $\tilde P_D$ obtained from our $ O((2^{d+2} + \log 1/\tau)/\gamma^2)$ samples is such that $\dtv(P_D, \tilde P) \leq \gamma/10$ with probability at least $1-\tau$. We hereafter assume that this holds.

Note that the distribution $P_D$ determines whether $P$ is such that $X_i$ and $X_j$ are independent conditioned on $S$ or is $\delta$-far from being so for any $\delta$. Thus if these two nodes are $\gamma$-far from being conditionally independent in $P$,  then they are $\gamma$-far in $P_D$ and therefore are $9\gamma/10$-far in $\tilde P_D$. We now need to show that the expected absolute value of the conditional covariance is a good approximation of the distance from conditional independence, which is our next claim:
\begin{claim} \label{clm:good-conditional-dependence-metric}
For a distribution $Q$ on $\{0,1\}^n$, let $\gamma$ be the minimum $\gamma > 0$ such that $X_i$ and $X_j$ are $\gamma$-far from independent conditioned on $X_S$ in $Q$. Let  $\beta = \E_{Y \sim Q}[|\Cov_Q[X_i,X_j \mid X_S=Y_S]|]$. Then we have $\beta/3 \leq \gamma \leq 2\beta$.
\end{claim}
\begin{proof}
For simplicity, we assume  that $|D|=n$ and that we have only coordinates $i$, $j$ and $S$.

Firstly, we show that $\beta \leq \gamma$. By assumption, there is a distribution $R$ with $\dtv(Q,R)=\gamma$ which has that $X_i$ and $X_j$ are independent conditioned on $X_S$. Thus $R$ has $|\Cov_R[X_i,X_j \mid X_S=a]|=0$ for all configurations $a$.
Since  $0 \leq |\Cov_Q[X_i,X_j \mid X_S=a]| \leq 1$, it follows that $|\beta -  \E_{Y \sim R}[|\Cov_Q[X_i,X_j\mid X_S=Y_S]|]| \leq 3\dtv(Q,R)$ {as $\Cov_Q[X_i,X_j \mid X_S=Y_S] = \E[X_iX_j \mid X_S=Y_S] - \E[X_i \mid X_S=Y_S]\E[X_j \mid X_S=Y_S]$} and so $\beta \leq 3\gamma$.

Next we need to show that $\beta \leq 2\gamma$. To show this, we construct a distribution $S$ on $\{0,1\}^n$ with $\dtv(Q,S) = 2 \beta$ in which $X_i$ and $X_j$ are independent conditioned on $X_S$. 
%This can be extended to a distribution $S$ with $\dtv(S,Q)=2\beta$ with the same conditional independence in an easy way that we omit here. 
Explicitly, for a configuration $a$ of $S$ and $b,c \in \{0,1\}$, we set
\begin{align*}
    \Pr_S[X_S=a,X_i=b,X_j=c] 
    &\eqdef \Pr_Q[X_S=a,X_i=b,X_j=c] + (-1)^{b+c} \Cov_Q[X_i,X_j \mid X_S=a] \Pr_Q[X_S=a] \; .
\end{align*}
For each configuration $a$, this increases two probabilities by $|\Cov_Q[X_i,X_j \mid X_S=a]|\Pr_Q[X_S=a]$ and decrease two probabilities by the same amount. Thus, provided that all probabilities are still non-negative (which we show below),  $S$ is a distribution with $\dtv(Q,S) = \sum_a 2|\Cov_Q[X_i,X_j \mid X_S=a]|\Pr_Q[X_S=a]=2\beta$. 

Now consider the conditional joint distribution of $X_i$, $X_j$for a given $a$. Let $p_{b,c} \eqdef \Pr_Q[X_i=b,X_j=c \mid X_S=a]$. Then the conditional covariance $\Cov_Q[X_i,X_j \mid X_S=a]$, which we denote by $\alpha$ for simplicity here, is
\begin{align*}
\alpha & = \E[X_iX_j \mid X_S=a] -\E[X_i\mid X_S=a] \E[X_j\mid X_S=a] \\
        & = p_{1,1} -(p_{1,0} + p_{1,1})(p_{0,1} + p_{1,1}) \\
        & = p_{1,1} (1- p_{1,0} -p_{0,1}  -  p_{1,1}) - p_{1,0} p_{0,1} \\
        & = p_{1,1} p_{0,0} - p_{1,0} p_{0,1} \;.
\end{align*}
In $S$, these probabilities change by $\alpha$. $p_{1,1}$ and $p_{0,0}$ are increased by $\alpha$ and $p_{0,1}$ and $p_{1,0}$ are decreased by it. Note that if $\alpha > 0$, $p_{1,1}$ and $p_{0,0}$ are at least $p_{1,1} p_{0,0} \geq \alpha$ and when $\alpha<0$, $p_{0,1}$ and $p_{1,0}$ are at least $p_{1,0} p_{0,1} \geq - \alpha$. Thus all probabilities in $S$ are in $[0,1]$, as claimed.

 A similar expression for the conditional covariance in $S$ to that for $\alpha$ above yields
\begin{align*} 
\Cov_{S}&[X_i,X_j \mid X_S=a]  \\
&= (p_{1,1} - \alpha)(p_{0,0} - \alpha) - (p_{1,0} +\alpha) (p_{0,1} + \alpha) \\
&= - (p_{0,0}+p_{1,1}+p_{0,1}+p_{1,0})\alpha + p_{1,1} p_{0,0} - p_{1,0} p_{0,1} \\
&= p_{1,1} p_{0,0} - p_{1,0} p_{0,1} - \alpha = 0 \;.
\end{align*}
Since $X_i$ and $X_j$ are Bernoulli random variables, the conditional covariance being zero implies that they are conditionally independent.
\end{proof}

\begin{description}
  \item[Completeness.] Suppose by contrapositive that the algorithm rejects.~\cref{clm:good-conditional-dependence-metric} implies that in $\tilde{P}$,  $X_i$ and $X_j$ are $\gamma/9$-far from independent conditioned on $X_S$. Thus they are $\gamma/9$ far in $\tilde P_D$ and, since $\dtv(P_D, \tilde P_D) \leq \gamma/10$, this implies that they are not conditionally independent in $P_D$. Thus, in $P$,  $X_i$ and $X_j$ are not independent conditioned on $X_S$. 

  \item[Soundness.] Now suppose that $X_i$ and $X_j$ are $\gamma$-far from independent conditioned on $X_S$ in $P$. Per the foregoing discussion, this implies that they are $(9\gamma/10)$-far from being so in $\tilde P_D$. Now~\cref{clm:good-conditional-dependence-metric} guarantees that $\E_{Y \sim \tilde P}[|\Cov_{\tilde P}[X_i,X_j|X_S=Y_S]|] \leq 9\gamma/20 > \gamma/3$, and therefore the algorithm rejects in this case. This completes the proof of correctness.
\end{description}
\end{proof}

%%%%%%%%%%%%%%%%%%%%%%%%%%%%%%%%%%%%%%%%%%%%%%%%%%%%%%%%%

\paragraph{Structure Tester} \label{ssec:structure-test}

Finally, we turn to the proof of~\cref{prop:struct-test}, analyzing the structure testing algorithm described in~\cref{algo:bn:structure-test}.

\begin{figure}[h!]\small
  \begin{framed}
    \begin{description}
      \item[Input] $\gamma > 0$, a structure $\structure$ and a Bayesian network $P$ 
      \item[-] Draw $O((2^d + d\log n)/\gamma^2)$ samples from $P$. Call this set of samples $S$.
      \item[For each] nodes $X_i$, $X_j$  and set $S$ of nodes with $|S| \leq d$ and $i,j \neq S$
      \begin{description}
        \item[If] one of the following conditions holds in structure $\structure$
        \begin{itemize}
          \item[(i)] $X_i$ is the parent of $X_j$,
          \item[(ii)] $S$ contains a node $X_k$ that is a child of both $X_j$ and $X_j$,
          \item[(iii)] $X_i$ is a grandparent of $X_j$ and there is a child of $X_i$ and parent of $X_j$, $X_k$ that is not in~$S$,
          \item[(iv)] $X_i$ and $X_j$ have a common parent $X_k$ that is not in $S$
        \end{itemize}
        \item[Then] run the conditional independence tester of~\cref{prop:conditional:independence:tester} (\cref{algo:cond:independ}) using the set of samples $S$ to test whether $X_i$ and $X_j$ are independent conditioned on $X_S$.
        \item[If] the conditional independence tester accepts, return $\textsf{reject}$.
      \end{description}
      \item[Otherwise] return $\textsf{accept}$.
    \end{description}
  \end{framed}
  \caption{Testing whether $P$ has structure as $\structure$ }\label{algo:bn:structure-test}
\end{figure}

\begin{proof}[Proof of~\cref{prop:struct-test}]
We first show correctness. There are at most $n^{d+2}$ possible choices of $X_i$, $X_j$ and $|S|$ 
and thus we run the conditional independence tester at most $n^{d+2}$ times. 
With $O((2^d + d\log n)/\gamma^2)$ samples, each test gives an incorrect answer with probability 
no more than $\tau=n^{-\Omega(d)}$. With appropriate choice of constants 
we therefore have that all conditional independence tests are correct with probability $99/100$. 
We henceforth condition on this, i.e., that all such tests are correct.

\begin{description}
  \item[Completeness.] If $P$ is $\gamma$-non-degenerate with respect to structure $\structure$ and degree $d$, then by the definition of non-degeneracy, for any $X_i$, $X_j$ and $S$ that satisfy one of conditions (i)--(iv) we have that $X_i$ and $X_j$ are $\gamma$-far from  independent conditioned on $X_S$. Thus every conditional independence test rejects and the algorithm accepts.

  \item[Soundness.] Now suppose by contrapositive that the algorithm accepts. For any $X_i$, $X_j$, and $S$ that satisfy one of conditions (i)--(iv), the conditional independence test must have rejected, that is any such $X_i$ and $X_j$ are not independent conditioned on such an $X_S$. Let $\gamma'$ be the mimuimum over all $X_i$, $X_j$, and $S$ that satisfy one of conditions (i)--(iv) and distributions $Q$ over $\{0,1\}$ such that $X_i$ and $X_j$ are independent conditioned on $X_S$ in $Q$, of the total variation distance between $P$ and $Q$. Since there are only finitely many such combinations of $X_i$, $X_j$, and $S$, this $\gamma'$ is positive. Thus $P$ is $\gamma'$-non-degenerate with respect to $\structure$ and $d$. Since we assumed that $P$ has a structure $\structure'$ with degree at most $d$ and whose skeleton has no more edges than that of $\structure$, we can apply~\cref{lem:struct-equiv}, which yields that $\structure$ and $\structure'$ are $I$-equivalent. Thus $P$ can indeed be expressed as a Bayesian network with structure $\structure$. This completes the proof of correctness.
\end{description}

To conclude, observe that we run the loop at most $n^{d+2}$ times, each using time at most $O((2^d + d\log n)/\gamma^2)$. The total running time is thus $O(n^{d+3}/\gamma^2)$.
\end{proof}

\section{Testing Closeness of Bayes Nets} \label{sec:closeness}

%%%%%%%%%%%%%%%%%%%%%%%%%%%%%%%%%%%%%%%%%%%%%

%%%%%%%%%%%%%%%%%%%%%%%%%%%%%%%%%%%%%%%%%%%%%

%\subsection{Testing closeness between two balanced Bayes net}\label{sec:closeness:general}

\subsection{Fixed Structure Bayes Nets} \label{sec:closeness-known}

We now establish the upper bound part of~\cref{thm:informal-identity-closeness-known} for closeness, namely testing closeness between two unknown Bayes nets 
with the same (known) underlying structure.

\begin{restatable}{theorem}{closenessknowndegreedub}\label{theo:upper:knowndegreed:closeness}
There exists a computationally efficient algorithm with the following guarantees. 
Given as input (i) a DAG $\structure$ with $n$ nodes and maximum in-degree $d$, 
(ii) a parameter $\eps > 0$, and (iii) sample access to two unknown $(c,C)$-balanced Bayes nets $P,Q$ with structure $\structure$, 
where $c=\tildeOmega{1/\sqrt{n}}$ and $C=\tildeOmega{d\eps^2/\sqrt{n}}$;
the algorithm takes $\bigO{2^{d/2}\sqrt{n}/\eps^2}$ samples from $P$ and $Q$, 
and distinguishes with probability at least $2/3$ between the cases $P=Q$ and $\normone{P-Q} > \eps$.
\end{restatable}
\begin{proof}
We choose $m\geq \alpha\frac{2^{d/2}\sqrt{n}}{\eps^2}$, where $\alpha>0$ is an absolute constant to be determined in the course of the analysis. Let $\structure$ and $P,Q$ be as in the statement of the theorem, for {$c\geq \beta\frac{\log n}{ \sqrt{n} } \geq \beta\frac{\log n}{m}$ and} $C\geq \beta\frac{d+\log n}{m}$, for some other absolute constant $\beta>0$.

The algorithm proceeds as follows: first, taking $m$ samples from both $P$ and $Q$, it computes for each parental configuration $(i,a)\in[n]\times\{0,1\}^d$ the number of times $\hat{N}_{i,a}$ and $\hat{M}_{i,a}$ this configuration was observed among the samples, for respectively $P$ and $Q$. If for any $(i,a)$ it is the case that $\hat{N}_{i,a}$ and $\hat{M}_{i,a}$ are not within a factor $4$ of each other, the algorithm returns $\textsf{reject}$. (Using the same number of samples, it also estimates $p_{i,a}$ and $q_{i,a}$ within an additive $1/3$, and applies the same standard transformation as before so that we can hereafter assume $p_{i,a},q_{i,a}\leq 2/3$ for all $(i,a)$.)

Note that $\expect{\hat{N}_{i,a}} = m \probaDistrOf{P}{\Pi_{i,a}}$  and $\expect{\hat{M}_{i,a}} = m \probaDistrOf{Q}{\Pi_{i,a}}$; given the $C$-balancedness assumption and by Chernoff and union bounds, with probability at least $9/10$ we have that $\hat{N}_{i,a}$ and $\hat{M}_{i,a}$ are within a factor $2$ of their expectation simultaneously for all $n2^d$ parental configurations. We hereafter condition on this (and observe that this implies that if $P=Q$, then the algorithm rejects in the step above with probability at most $1/10$).

The algorithm now draws independently $n2^d$ values $(M_{i,a})_{(i,a)}$, where $M_{i,a}\sim\poisson{\hat{N}_{i,a}}$; and takes fresh samples from $P,Q$ until it obtains $M_{i,a}$ samples for each parental configuration $\Pi_{i,a}$ (for each of the two distributions). If at any point the algorithm takes more than $10m$ samples, it stops and returns $\textsf{reject}$.

\noindent (Again, note that by concentration (this time of Poisson random variables)\footnote{Specifically, if $X\sim\poisson{\lambda}$ then we have $\probaOf{\abs{X-\lambda} > \lambda/2} = e^{-\Omega(\lambda)}$.}, our assumption that $\hat{N}_{i,a} \geq m\probaDistrOf{P}{\Pi_{i,a}}/2 \geq mC/2 = \beta(d+\log n)$ and a union bound, the algorithm will reject at this stage with probability at most $1/10$.)

Conditioning on not having rejected, we define for each parental configuration $\Pi_{i,a}$ the quantity $U_{i,a}$ (resp. $V_{i,a}$) as the number of samples from $P$ (resp. $Q$) among the first $M_{i,a}$ satisfying $\Pi_{i,a}$ for which $X_i=1$. In particular, this implies that $U_{i,a} \sim\poisson{p_{i,a}\hat{N}_{i,a}}$, $V_{i,a} \sim\poisson{q_{i,a}\hat{N}_{i,a}}$ (and are independent), and that the random variables $W_{i,a}$ defined below:

\[
W_{i,a} \eqdef \frac{(U_{i,a}-V_{i,a})^2 - (U_{i,a}+V_{i,a})}{U_{i,a}+V_{i,a}}
\]
are independent. We then consider the statistic $W$:
\[
  W \eqdef \sum_{i=1}^n \sum_{a\in\{0,1\}^d} W_{i,a}.
\]

\begin{claim}\label{claim:general:closeness:expectation}
  If $P=Q$, then $\expect{W} = 0$. Moreover, if $\normone{P-Q} > \eps$ then $\expect{W} > \frac{m\eps^2}{144}$.
%   Moreover, we have 
%   \[
%   \expect{W} \geq \frac{1}{9}\sum_{(i,a)}\sqrt{\probaDistrOf{P}{\Pi_{i,a}}\probaDistrOf{Q}{\Pi_{i,a}}} \frac{(p_{i,a}-q_{i,a})^2}{(p_{i,a}+q_{i,a})(2-p_{i,a}-q_{i,a})}
%   \]
\end{claim}
\begin{proof}
We start by analyzing the expectation of $W_{i,a}$, for any fixed $(i,a)\in[n]\times\{0,1\}^d$. The same argument as~\cref{claim:general:closeness:expectation} leads to conclude that
$\expect{W_{i,a}}=0$ if $P=Q$ (proving the first part of the claim), and that otherwise we have 
\begin{align}\label{eq:general:closeness:expectation}
  \expect{W_{i,a}}
  &\geq \frac{\min(1,mc)}{3}\hat{N}_{i,a} \frac{(p_{i,a}-q_{i,a})^2}{p_{i,a}+q_{i,a}} \notag\\
  &= \frac{1}{3}\hat{N}_{i,a} \frac{(p_{i,a}-q_{i,a})^2}{p_{i,a}+q_{i,a}} \notag\\
  &\geq \frac{2}{9}\hat{N}_{i,a} \frac{(p_{i,a}-q_{i,a})^2}{(p_{i,a}+q_{i,a})(2-p_{i,a}-q_{i,a})} 
\end{align}
(since $mc \geq \beta\log n \gg 1$ and $0<p_{i,a},q_{i,a} \leq 2/3$). Summing over all $(i,a)$'s and recalling that $\hat{N}_{i,a} \geq m\probaDistrOf{P}{\Pi_{i,a}}/2$, $\hat{N}_{i,a} \geq m\probaDistrOf{Q}{\Pi_{i,a}}/2$ yields the bound:
\begin{align*}
    \expect{W} &\geq \frac{m}{9}\sum_{(i,a)}\frac{\sqrt{\probaDistrOf{P}{\Pi_{i,a}}\probaDistrOf{Q}{\Pi_{i,a}}}  (p_{i,a}-q_{i,a})^2}{(p_{i,a}+q_{i,a})(2-p_{i,a}-q_{i,a})} \\
    &\geq \frac{m}{18} \hellinger{P}{Q}^2 \geq \frac{m}{18}\left( 1-\sqrt{1-\frac{1}{4}\normone{P-Q}^2} \right)
\end{align*}
(where we relied on~\cref{lemma:hellinger:bn} for the second-to-last inequality). This gives the last part of the claim, as the RHS is at least $\frac{m\eps^2}{144}$ whenever $\normone{P-Q}^2 > \eps^2$.
\end{proof}

\noindent We now bound the variance of our estimator:
\begin{claim}\label{claim:general:closeness:variance}
  $\Var[W] \leq n 2^{d+1} + 5\sum_{(i,a)} \hat{N}_{i,a} \frac{(p_{i,a}-q_{i,a})^2}{p_{i,a}+q_{i,a}} = O(n 2^d + \expect{W})$. In particular, if $P=Q$ then $\Var[W] \leq n 2^{d+1}$.
\end{claim}
\begin{proof}
We follow the proof of~\cref{claim:product:closeness:variance} to analyze the variance of $W_{i,a}$, obtaining a bound of $\Var[W_{i,a}] \leq 2+5\hat{N}_{i,a} \frac{(p_{i,a}-q_{i,a})^2}{p_{i,a}+q_{i,a}}$. Invoking \cref{eq:general:closeness:expectation} and summing over all $(i,a)\in[n]\times\{0,1\}^d$ then lead to the desired conclusion.
\end{proof}

\noindent The correctness of our algorithm will then follow for the two claims above:
\begin{lemma}
Set $\tau \eqdef \frac{\eps^2}{288}$. Then we have the following:
  \begin{itemize}
    \item If $\normone{P-Q} = 0$, then $\probaOf{ W \geq \tau am } \leq \frac{1}{10}$.
    \item If $\normone{P-Q} > \eps$, then $\probaOf{ W < \tau m } \leq \frac{1}{10}$.
  \end{itemize}
\end{lemma}
\begin{proof}
We start with the soundness case,i.e. assuming $\normone{P-Q} > \eps$, which by~\cref{claim:general:closeness:expectation} implies $\expect{W} > 2\tau$. Then, by Chebyshev's inequality,
\begin{align}
  \probaOf{ W < \tau m } &\leq \probaOf{ \expect{W} - W > \frac{1}{2}\expect{W} }
  \leq \frac{4\Var[W]}{\expect{W}^2} \notag\\
  &\leq \frac{8n 2^d}{\expect{W}^2}  + \frac{12}{5\expect{W}} \tag{\cref{claim:general:closeness:variance}} \\
  &=O\!\left( \frac{n 2^d}{\eps^4 m^2} + \frac{1}{m\eps^2}\right) \notag \;.
\end{align}
We want to bound this quantity by $1/10$, for which it is enough to have
$\frac{n 2^d}{\eps^4 m^2} \ll 1$ and $\frac{1}{m\eps^2} \ll 1$, which both hold for an appropriate choice of the absolute constant $\alpha>0$ 
in our setting of $m$.

Turning to the completeness, we suppose $\normone{P-Q} = 0$. Then, by Chebyshev's inequality, and invoking~\cref{claim:general:closeness:variance},
\begin{align*}
  \probaOf{ W \geq \tau m } &= \probaOf{ W \geq \expect{W} + \tau m }
  \leq \frac{\Var[W]}{\tau^2 m^2} \\
  &= O\!\left(\frac{n 2^d}{\eps^4 m^2}\right)
\end{align*}
which is no more than $1/10$ for the same choice of $m$.
 \end{proof}
Combining all the elements above concludes the proof, as by a union bound the algorithm is correct with probability at least $1-(\frac{1}{10}+\frac{1}{10}+\frac{1}{10}) > \frac{2}{3}$.
\end{proof}

\subsection{Unknown Structure Bayes Nets} \label{sec:closeness-unknown}

As for the case identity testing, we give a closeness tester for balanced non-degenerate Bayes Nets.
An additional assumption that we require is that the ordering of the nodes in the corresponding DAGs 
is known to the algorithm. Formally, we show:

\begin{theorem}\label{thm:informal-upper-closeness-unknown-nondegen}
There exists a computationally efficient algorithm with the following guarantees. 
Given as input (i) a parameter $\eps > 0$, (ii) an ordering of nodes $\pi$, and 
(ii) sample access to unknown $\gamma$-non-degenerate, $(c,C)$-balanced Bayes nets $P,Q$ 
such the structures of $P$ and $Q$ give the same ordering $\pi$ to nodes, 
where $c=\tildeOmega{1/\sqrt{n}}$ and $C=\tildeOmega{d\eps^2/\sqrt{n}}$;
the algorithm takes  $N=O(2^{d/2}\sqrt{n}/\eps^2 + 2^d /\gamma^2+d\log(n)/\gamma^2)$ samples from $P$ and $Q$, 
runs in time $n^d\poly(N)$, and distinguishes with probability at least $2/3$ 
between the cases $P=Q$ and $\normone{P-Q} > \eps$.
\end{theorem}
\begin{proof}
The argument's idea is the following: we first test that $P$ and $Q$ have the same skeleton. Since they have the same ordering, that suffices to show that they have the same structure. If this is the case, then we use our known-structure tester.

In more detail, given the $\gamma$-non-degeneracy assumption, for each pair of coordinates $i,j$ and set of coordinates $S$ with $\abs{S}\leq d$, we can, using the conditional independence tester from~\cref{prop:conditional:independence:tester}, test whether each of $P$ and $Q$ has $X_i$ and $X_j$ conditionally independent on $X_S$ or $\gamma$-far from it with $n^{-d-2}/100$ probability of error in $O((2^d+d\log(n))/\gamma^2)$ samples. 
Running tests on the same samples for all $n^{d+2}$ combinations of $i,j,S$, we can with probability at least $99/100$ correctly classify which of the two cases holds, for all $i,j,S$ that are either conditionally independent or $\gamma$-far. 
We note that by non-degeneracy, there is an edge between $i$ and $j$ in the structure defining $P$ only if $X_i$ and $X_j$ are $\gamma$-far from independent conditioned on $X_S$ for all $S$ (i.e., if there is no edge then there must exist a $S$ such that $X_i$ and $X_j$ are conditionally independent on $X_S$).
Therefore, assuming our conditional independence testers all answered as they should, we can use this to successfully identify the set of edges in the structure of $P$ (and thus, since we know the ordering, the entire structure).

Having determined the underlying structures of $P$ and $Q$, our tester rejects if these structures differ (as using~\cref{lem:struct-equiv},  $\gamma$-non-degeneracy implies that neither can equal a Bayes net with non-equivalent structure and fewer edges). Otherwise, we run the tester from~\cref{theo:upper:knowndegreed:closeness} (since we satisfy its assumptions) and return the result.
\end{proof}

\section{Identity and Closeness Testing for High-Degree Bayes Nets} \label{sec:it:ub}

%%%%%%%%%%%%%%%%%%%%%%%%%%%%%%%%%%%%%%%%%%%%%%%%%%%%%%%%%%%%%%%%%%%%%%%%%%%%%%%%%%%%%%%%%%%%%%%%%%%%%%%%%%%%%%%%%%%%%%%%%%%%%%%%%%%%%%%%%%
%%%%%%%%%%%%%%%%%%%%%%%%%%%%%%%%%%%%%%%%%%%%%%%%%%%%%%%%%%%%%%%%%%%%%%%%%%%%%%%%%%%%%%%%%%%%%%%%%%%%%%%%%%%%%%%%%%%%%%%%%%%%%%%%%%%%%%%%%%
%%%%%%%%%%%%%%%%%%%%%%%%%%%%%%%%%%%%%%%%%%%%%%%%%%%%%%%%%%%%%%%%%%%%%%%%%%%%%%%%%%%%%%%%%%%%%%%%%%%%%%%%%%%%%%%%%%%%%%%%%%%%%%%%%%%%%%%%%%

Finally, in this section we give testing algorithms for identity and closeness of degree-$d$ Bayes nets with unknown structure, \emph{without} balancedness assumptions. 
Compared to the testing algorithm of~\cref{thm:informal-upper-identity-unknown-nondegen} and~\cref{thm:informal-upper-closeness-unknown-nondegen} (which work under such assumptions) the dependence on the number of nodes $n$ the testers in this section are suboptimal, 
they achieve the ``right'' dependence on the degree $d$ (specifically, $2^{d/2}$ for identity and $2^{2d/3}$ for closeness).
Hence, these testers achieve sub-learning sample complexity for the case that $d = \Omega(\log n)$.

\begin{theorem} \label{thm:unknown-structure-identity:informationtheoretic}
There exists two algorithms with the following guarantees:
\begin{itemize}

\item (Identity) Given the full description of a Bayes net $Q$ of degree at most $d$, 
parameter $\eps\in(0,1]$, and sample access to a distribution  $P$ promised to be a Bayes net (i) of degree at most $d$
and (ii) such that the structures of $P$ and $Q$ give the same ordering to nodes, 
the first takes $N=2^{d/2}\poly(n/\eps)$ samples from $P$, runs in time $n^d\poly(N)$, 
and distinguishes with probability at least $2/3$ between (i) $P=Q$ and (ii) $\normone{P-Q}  > \eps$.
 
\item (Closeness) Given parameter $\eps\in(0,1]$, and sample access to two distributions  $P,Q$ 
promised to be Bayes nets (i) of degree at most $d$
and (ii) such that the structures of $P$ and $Q$ give the same ordering to nodes, 
the second takes $N=2^{2d/3}\poly(n/\eps)$ samples from $P$ and $Q$, 
runs in time $n^d\poly(N)$, and distinguishes with probability at least $2/3$ 
between (i) $P=Q$ and (ii) $\normone{P-Q}  > \eps$.
 \end{itemize}
\end{theorem}
\begin{proof}
We first establish the first part of the theorem, namely the existence of an identity testing algorithm 
with optimal dependence on the degree $d$. The algorithm is quite simple: it goes over each set $S\subseteq [n]$ of at most $d+1$ coordinates, 
and checks that for each of them it holds that the conditional distributions $P,S,Q_S$ 
are equal (versus $\normone{P_S-Q_S} > \poly(\frac{\eps}{n})$).

Since $P_S$ and $Q_S$ are supported on sets of size $O(2^d)$, 
and as there are only $O(n^{d+1})$ such sets to consider, 
the claimed sample complexity suffices to run all tests correctly 
with probability $9/10$ overall (by a union bound).

The more difficult part is to argue correctness, 
that is to show that if the test accepts then one must have $\normone{P-Q} < \eps$. 
To do so, assume (without loss of generality) that
$H(P) \leq H(Q)$: we will show that $\dkl{P}{Q}$ is small, which implies that the $L_1$ distance is as well.

Let the ordering of $P$ be coordinates $1,2,3,\dots$. We note that $\dkl{P}{Q} = \sum_i \dkl{P_i}{Q_i  \mid  P_1,\dots,P_{i-1}}$ (i.e. the expectation over $P_1,\dots,P_{i-1}$ of the KL-divergence of the conditional distributions of $P_i$ and $Q_i$, conditioned on these $(i-1)$ coordinates). It thus suffices to show that each of these terms is small.

Let $S_i$ be the set of parents of node $i$ under $P$. We have that:
\begin{align*}
\dkl{P_i}{Q_i  \mid  P_1,\dots,P_{i-1}} &= \dkl{P_i}{Q_i  \mid  P_{S_i}} + \shortexpect_{P_1,\dots,P_{i-1}}[\dkl{Q_i\mid  P_{S_i}}{Q_i  \mid  P_1,\dots,P_{i-1}}] \;.
\end{align*}
Further, note that the fact that the tester accepted implies that $\dkl{P_i}{Q_i  \mid  P_{S_i}}$ is small. Now, we have that 
\begin{align*}
H(P) &= \sum_i H(P_i \mid P_1,\dots,P_{i-1}) = \sum_i H(P_i\mid P_{S_i})  \;, \\
H(Q) &= \sum_i H(Q_i\mid Q_1,\dots,Q_{i-1}) \\
&= \sum_i H(Q_i\mid Q_{S_i}) - \mutualinfo{Q_i}{Q_1,\dots,Q_{i-1} \mid Q_{S_i}} \;.
\end{align*}
But since the $(d+1)$-wise probabilities are close, we have that $H(P_i \mid P_{S_i})$ is close to $H(Q_i \mid Q_{S_i})$ (up to an additive $\poly(\eps/n)$). Therefore, for each $i$, we have that $\mutualinfo{Q_i}{Q_1,\dots,Q_{i-1} \mid Q_{S_i}} = \poly(\eps/n)$. In order to conclude, let us compare $\mutualinfo{Q_i}{Q_1,\dots,Q_{i-1} \mid Q_{S_i}}$ and $\shortexpect_{P_1,\dots,P_{i-1}}[\dkl{Q_i\mid  P_{S_i}}{Q_i  \mid  P_1,\dots,P_{i-1}}]$. The former is the sum, over assignments $y\in\{0,1\}^{i-1}$ consistent with an assignment $x\in\{0,1\}^{S_i}$, of
\begin{align*}
  \probaOf{ Q_{S_i} = x } H(Q_i \mid Q_{S_i} = x) + \probaOf{ Q_{1,\dots,i-1} = y } H(Q_i \mid Q_{1,\dots,i-1}=y ).
\end{align*}
The latter is the sum over the same $y$'s of
\begin{align*}
  \probaOf{ P_{S_i} = x } H(Q_i \mid Q_{S_i} = x) + \probaOf{ P_{1,\dots,i-1} = y} H(Q_i \mid Q_{1,\dots,i-1}=y ) \;.
\end{align*}
But because of the $d$-way probability similarities, 
the terms $\probaOf{ P_{S_i} = x }$ and $\probaOf{ Q_{S_i} = x }$ terms are very close, 
within an additive $\poly(\eps/n)$.

\emph{(Here we use the extra assumption that $P$ and $Q$ use the same ordering.)}
Denote by $T_i$ the parents of $i$ under the topology of $Q$. 
Then $H(Q_i \mid Q_{1,\dots,i-1}=y)$ depends only on the values of the coordinates in $T_i$. Thus the last part
of the sum is a sum over $z$ of $\probaOf{ Q_{T_i} = z } H(Q_i \mid Q_{T_i} =z)$ and $\probaOf{ P_{T_i} = z } H(Q_i \mid Q_{T_i} =z)$, 
which are  also close by a similar argument.
Thus, 
\begin{align*}
\shortexpect_{P_1,\dots,P_{i-1}}&[\dkl{Q_i\mid  P_{S_i}}{Q_i  \mid  P_1,\dots,P_{i-1}}]\\
 &= \mutualinfo{ Q_i }{ Q_1,\dots,Q_{i-1} \mid Q_{S_i} } + \poly(\eps/n) \\
 &= \poly(\eps/n) \;.
\end{align*}
This implies that $P,Q$ are close in KL-divergence, and therefore in $L_1$.
\medskip

The second part of the theorem, asserting the existence of a closeness testing algorithm 
with optimal dependence on $d$, will be very similar. Indeed, by the proof above it suffices 
to check that the restrictions of $P$ and $Q$ to any set of $(d+3)$-coordinates are $\poly(\eps/n)$-close. 
Using known results~\cite{CDVV14}, this can be done for any specific collection of $d+3$ coordinates 
with $N$ samples in $\poly(N)$ time, and high probability of success, 
implying the second part of the theorem.
\end{proof}
%%%%%%%%%%%%%%%%%%%%%%%%%%%%%%%%%%%%%%
%%%%%%%%%%%%%%%%%%%%%%%%%%%%%%%%%%%%%%%%%%%%%%%%%%%%%%%%%%%%%%%%%%%%%%%%%%%%%%%%
%%%%%%%%%%%%%%%%%%%%%%%%%%%%%%%%%%%%%%%%%%%%%%%%%%%%%%%%%%%%%%%%%%%%%%%%%%%%%%%%

\bibliographystyle{alpha}
\bibliography{references}

\appendix

%%%%%%%%%%%%%%%%%%%%%%%%%%%%%%%%%%%%%%%%%%%%%%%%%%%%%%%%%%%%%%%%%%%%%%%%%%%%%%%%%%%
%%%%%%%%%%%%%%%%%%%%%%%%%%%%%%%%%%%%%%%%%%%%%%%%%%%%%%%%%%%%%%%%%%%%%%%%%%%%%%%%%%%

\section{Sample Complexity of Learning Bayesian Networks} \label{sec:learn}

In this section, we derive tight bounds on the sample complexity of learning Bayes nets.
Recall that $n$ will denote the number of nodes and $d$ the maximum in-degree; 
before stating the results and providing their proofs, we outline the high-level idea of the argument.

If the structure is known, there is a very
simple learning algorithm involving finding the empirical values for
the relevant conditional probabilities and constructing the Bayes net
using those terms. By computing the expected KL-Divergence between
this hypothesis and the truth, we can show that it is possible to
learn an $\eps$-aproximation in $\tildeO{2^d n /\eps^2}$ samples. Learning a
Bayes net with unknown structure seems substantially more challenging,
but at least the sample complexity is no greater. In particular, we
can simply go over all possible topologies and come up with one
hypothesis for each, and use a standard tournament to pick out the
correct one. \cref{sec:learn-upper} contains the details of both.

We also prove in~\cref{sec:learn-lower} a matching lower bound (up to logarithmic factors).
For this we can even consider Bayesian networks of fixed topology. In
particular, we consider the topology where each of the last
$(n-d)$-coordinates depend on all of the first $d$ (which we can assume to
be uniform and independent). The distribution we end up with is what
we call a disjoint mixture of $2^d$ product distributions. In particular
for each of the $2^d$ possible combinations of the first $d$ coordinates,
we have a (potentially different) product distribution over the
remaining coordinates. In order to learn our final distribution we
must learn at least half of these product distributions to $O(\eps)$
error. This requires that we obtain at least $\Omega((n-d)/\eps^2)$
samples from $\Omega(2^d)$ of the parts of our mixture. Thus, learning
will require $\Omega(2^d (n-d)/\eps^2))$ total samples.
\medskip

\subsection{Sample Complexity Upper Bound} \label{sec:learn-upper}

\paragraph{Known Structure} The algorithm will return a Bayes net $Q$ with the same (known) structure as the unknown $P$. Define $p_{i,a}$ as the probability (under $P$) that the $i$-th coordinate is $1$, conditioned on the parental configuration for $i$ being equal to $a$ (denoted $\Pi_{i,a}$); and $q_{i,a}$ the corresponding parameter for our hypothesis $Q$.
 
Given this, the algorithm is simple. %First, by a standard argument we can assume that all $p_{i,a}$ are in $[0,2/3]$ (as one can detect whether it is the case by taking $O(\log n)$ samples, and swap the corresponding coordinate of each sample if it is not the case).
We set $m=O(\frac{2^d n}{\eps^2}\log(2^d n))$, and consider separately two sets of configurations $(i,a)$:
\begin{itemize}
  \item the \emph{light} ones, for which $p_{i,a}{(1-p_{i,a})}\probaDistrOf{P}{ \Pi_{i,a} } \leq \frac{\eps}{2n 2^d}$;
  \item the \emph{heavy} ones, for which $p_{i,a}{(1-p_{i,a})}\probaDistrOf{P}{ \Pi_{i,a} } > \frac{\eps}{2n 2^d}$.
\end{itemize}

We take $m$ samples from $P$, and let $q_{i,a}$ be the empirical conditional probabilities. For any $(i,a)$ for which we see less than $\tau\eqdef O((\log n)/\eps)$ samples, we set $q_{i,a}=0$. Note that with high probability, for the right choice of constant in the definition of $\tau$, for all heavy configurations simultaneously  the estimate $q_{i,a}$ will (i) not be set to zero; and (ii) be {such that $q_{i,a}(1-q_{i,a})$ is within a factor $2$ of the real value $p_{i,a}(1-p_{i,a})$}. Conversely, each light configuration will either have $q_{i,a}$ be zero, or {$q_{i,a}(1-q_{i,a})$ within a factor $2$ of $p_{i,a}(1-p_{i,a})$}.

Conditioned on this happening, we can analyze the error. Note first that by our definition of light configuration and the triangle inequality, the sum of $L_1$ error over all light configurations will be at most $\eps/2$. We now bound the contribution to the error of the heavy configurations. (In what follows, we implicitly restrict ourselves to these.)

By Pinsker's inequality, $\normone{P-Q}^2 \leq 2 \dkl{P}{Q} \leq \sum_{(i,a)} \probaDistrOf{P}{ \Pi_{i,a} }  \frac{(p_{i,a}-q_{i,a})^2}{q_{i,a}(1-q_{i,a})}$. It only remains to bound the expected KL-divergence \textit{via} the above bound.
For each pair $(i,a)$, on expectation we see $m \probaDistrOf{P}{ \Pi_{i,a} }$ samples satisfying the parental configuration. The expected squared error $(p_{i,a}-q_{i,a})^2$ between $p_{i,a}$ and our estimate $q_{i,a}$ is then 
$O(\frac{p_{i,a}{(1-p_{i,a})}}{m\probaDistrOf{P}{\Pi_{i,a} }} )$ by a standard variance argument. This implies that the expected KL-divergence is bounded as
\[
    \expect{ \dkl{P}{Q} } = \frac{1}{m}\sum_{(i,a)} O\!\left(\expect{ \frac{ p_{i,a}(1-p_{i,a})}{q_{i,a}(1-q_{i,a}) } } \right) %= O\!\left(\frac{2^d n}{m}\right)
\]
and it will be sufficient to argue that $\frac{ p_{i,a}(1-p_{i,a})}{q_{i,a}(1-q_{i,a}) } = O\!\left(1\right)$ with high probability (as then the RHS will be bounded as $O\!\left(\frac{2^d n}{m}\right)$, and thus taking $m=O(2^d n /\eps^2)$ will be sufficient). {It is enough to show that, with high probability, $q_{i,a}(1-q_{i,a})$ is within a factor $2$ of $p_{i,a}(1-p_{i,a})$} simultaneously for all $(i,a)$; but this is exactly the guarantee we had due to the ``heaviness'' of the configuration. Therefore, the contribution to the squared $L_1$ error from the heavy configuration is $O\!\left(\frac{2^d n}{m}\right) = O\!\left(\frac{\eps^2}{d + \log n}\right) \ll \eps^2$. Gathering the two sources of error, we get an $L_1$-error at most $\eps$ with high probability, as claimed.

%\cnote{There is [was] a gap in the [previous] argument above~--~$O\!\left(\frac{p_{i,a}(1-p_{i,a})}{q_{i,a}(1-q_{i,a})} \right) = O(1)$?}
%\inote{If this require non-trivial work, just use covers; look at Dasgupta's paper. This might lose a $\log n$ factor, but who cares?}
%\cnote{Made an attempt to fix it~--~someone definitely should check that proof}

\paragraph{Unknown Structure}
For this upper bound, we reduce to the previous case of known structure. 
Indeed, there are only $N=n^{O(dn)}$ possible max-degree-$d$ candidate structures: 
one can therefore run the above algorithm (with probability of failure $9/10$) 
for each candidate structure on a common sample from $P$ of size $O(2^d n /\eps^2)$, 
before running a tournament (cf.~\cite{DDS12stoc, DDS15}) to select a hypothesis with error $O(\eps)$ 
(which is guaranteed to exist as, with probability at least $9/10$, the ``right'' candidate structure 
will generate a good hypothesis with error at most $\eps$). The total sample complexity will then be
\[
O\!\left( \frac{2^d n}{\eps^2} + \frac{\log N}{\eps^2} \right) = O\!\left( \frac{2^d n}{\eps^2} + \frac{dn \log n}{\eps^2} \right) \;,
\]
which is $O\!\left( {2^d n}/{\eps^2}\right)$ for $d=\Omega(\log n)$.

\subsection{Sample Complexity Lower Bound}  \label{sec:learn-lower}

Our lower bound will be derived from families of Bayes nets with the following structure:
The first $d$ nodes are all independent (and will in fact have marginal probability $1/2$ each), 
and will form in some sense a ``pointer'' to one of $2^d$ arbitrary product distributions. 
The remaining $n-d$ nodes will each depend on all of the first $d$. 
The resulting distribution is now an (evenly weighted) disjoint mixture of $2^d$ product distributions 
on the $(n-d)$-dimensional hypercube. In other words, there are $2^d$ product distributions 
$p_1,\dots,p_{2^d}$, and our distribution returns a random $i$ (encoded in binary) 
followed by a random sample form $p_i$. Note that the $p_i$ can be arbitrary product distributions.

We show a lower bound of $\Omega(2^{d} n/\eps^2)$ lower bound for learning, 
whenever $d < (1-\Omega(1))n$.

Let $\mathcal{C}_\eps$ be a family of product distributions over $\{0,1\}^{n-d}$ which is hard to learn, 
i.e., such that any algorithm learning $\mathcal{C}_\eps$ to accuracy $4\eps$ 
which succeeds with probability greater than $1/2$ must have sample complexity 
$\Omega(n/\eps^2)$. We will choose each $p_i$ independently and uniformly at random from $\mathcal{C}_\eps$.

Assume for the sake of contradiction that there exists an algorithm $\mathcal{A}$ to learn the resulting 
disjoint mixture of $p_i$'s to error $\eps$ with $o(2^d (n-d)/\eps^2)$ samples. 
Without loss of generality, this algorithm can be thought of returning as hypothesis a disjoint 
union of some other distributions $q_1,\dots,q_{2^d}$, in which case the error incurred is 
$\normone{p-q} = \frac{1}{2^d}\sum_{i=1}^{2^d} \normone{p_i-q_i}$.

By assumption on its sample complexity, for at least half of the indices $1\leq i \leq 2^d$ 
the algorithm obtains $o(n/\eps^2)$ samples from $p_i$. This implies that in expectation, 
by the fact that $\mathcal{C}_\eps$ was chosen hard to learn, for at least half of these indices 
it will be the case that $\normone{p_i-q_i} > 16\eps$ (as each of these $i$'s is such 
that $\normone{p_i-q_i} > 16\eps$ with probability at least $1/2$). 
This in turn shows that in expectation, $\normone{p-q} = \frac{1}{2^d}\sum_{i=1}^{2^d} \normone{p_i-q_i} > \frac{4\eps}{4} = \eps$, 
leading to a contradiction.

%%%%%%%%%%%%%%%%%%%%%%%%%%%%%%%%%%%%%%%%%%%%%%%%%%%%%%%%%%%%%%%%%%%%%%%%%%%%%%%%%%%%%%%%%%%%%%%%%%%%%%%%%%%%%%%%%%%%%%%%%%%%%
\section{Omitted Proofs \texorpdfstring{from~\cref{ssec:product-lower}}{from~\cref{ssec:product-lower}}} \label{ssec:product-lower:proofs}
We give in this section the proofs of the corresponding lemmas from~\cref{{ssec:product-lower}}.

\begin{proof}[Proof of Lemma~\ref{lemma:noinstances:far}]
By symmetry, it is sufficient to consider the distribution $P\eqdef \bigotimes_{j=1}^n \bernoulli{\frac{1}{2} + \frac{\eps}{\sqrt{n}}}$. 
We explicitly bound from below the expression of $\normone{P-U}$:%\cnote{Hellinger-based arguments would only give a lower bound of $\eps^2$ or so.}
\begin{align*}
  \normone{P-U}  
  &= \sum_{x\in\{0,1\}^n} \bigabs{ \left( \frac{1}{2} + \frac{\eps}{\sqrt{n}} \right)^{\abs{x}}\left( \frac{1}{2} - \frac{\eps}{\sqrt{n}} \right)^{n-\abs{x}} - \frac{1}{2^n} } \\
  &= \frac{1}{2^n}\sum_{k=0}^n \binom{n}{k}\bigabs{ \left( 1 + \frac{2\eps}{\sqrt{n}} \right)^{k}\left( 1 - \frac{2\eps}{\sqrt{n}} \right)^{n-k} - 1 } \\
  &\geq \frac{1}{2^n}\sum_{k=\frac{n}{2}+\sqrt{n}}^{\frac{n}{2}+2\sqrt{n}} \binom{n}{k}\bigabs{ \left( 1 + \frac{2\eps}{\sqrt{n}} \right)^{k}\left( 1 - \frac{2\eps}{\sqrt{n}} \right)^{n-k} - 1 } \\
  &\geq \frac{C}{\sqrt{n}}\sum_{k=\frac{n}{2}+\sqrt{n}}^{\frac{n}{2}+2\sqrt{n}} \bigabs{ \left( 1 + \frac{2\eps}{\sqrt{n}} \right)^{k}\left( 1 - \frac{2\eps}{\sqrt{n}} \right)^{n-k} - 1 } \;,
\end{align*}
where $C>0$ is an absolute constant. We handle each summand separately: fixing $k$, and writing $\ell \eqdef k-\frac{n}{2} \in [\sqrt{n},2\sqrt{n}]$,
\begin{align*}
  \left( 1 + \frac{2\eps}{\sqrt{n}} \right)^{k}\left( 1 - \frac{2\eps}{\sqrt{n}} \right)^{n-k}
  &=  \left( 1 - \frac{4\eps^2}{n} \right)^{n/2}\left( \frac{1 + \frac{2\eps}{\sqrt{n}}}{1 - \frac{2\eps}{\sqrt{n}}}\right)^\ell \\
  &\geq \left( 1 - \frac{4\eps^2}{n} \right)^{n/2}\left( \frac{1 + \frac{2\eps}{\sqrt{n}}}{1 - \frac{2\eps}{\sqrt{n}}}\right)^{\sqrt{n}} \\
  &\xrightarrow[n\to\infty]{} e^{4\eps-2\eps^2} \;,
\end{align*}
so that each summand is bounded by a quantity that converges (when $n\to \infty$) to $e^{4\eps-2\eps^2}-1 > 4\eps-2\eps^2 > 2\eps$, 
implying that each is $\Omega(\eps)$. 
{Combining the above gives}
\begin{align*}
  \normone{P-U} &\geq \frac{C}{\sqrt{n}}\sum_{k=\frac{n}{2}+\sqrt{n}}^{\frac{n}{2}+2\sqrt{n}} \Omega(\eps) = \Omega(\eps)
\end{align*}
as claimed.
\end{proof}

\begin{proof}[Proof of Lemma~\ref{lemma:lb:key}]
By symmetry it suffices to consider only the case of $i=1$, so that we let $A=N_1$. The first step is to bound from above $\mutualinfo{X}{A}$ by a more manageable quantity: 
\begin{fact}\label{fact:mutualinfo:expansion}
We have that
  \begin{equation}
  \mutualinfo{X}{A} \leq \sum_{a=0}^\infty \probaOf{ A=a }\mleft(1-\frac{ \probaCond{A=a }{ X=1}  }{ \probaCond{A=a }{ X=0} }\mright)^2.
  \end{equation}
\end{fact}
\noindent The proof of this fact is given in~\cref{sec:misc:proofs}. It will then be sufficient to bound the RHS, which we do next. Since $A\sim \poisson{kp_1}$ with $p_1 = 1/2$ if $X=0$ and uniformly $\frac{1}{2}\pm\frac{\eps}{\sqrt{n}}$ if $X=1$, a simple computation yields that
\begin{align*}
\probaCond{A=\ell }{ X=0 } &= e^{-k/2}\frac{(k/2)^\ell}{\ell!} \\
\probaCond{A=\ell }{ X=1 } 
  &= e^{-k/2}\frac{(k/2)^\ell}{\ell!}  \!\left(\frac{e^{-k\eps/\sqrt{n}}(1+2\frac{\eps}{\sqrt{n}})^\ell+e^{k\eps/\sqrt{n}}(1-2\frac{\eps}{\sqrt{n}})^\ell}{2} \right).
\end{align*}
Writing out $\varphi(\eps,\ell) = \frac{ \probaCond{ A=\ell }{ X=1 } }{\probaCond{A=\ell }{ X=0 }}$ 
as a function of $\eps/\sqrt{n}$, we see that it is even. Thus, expanding it as 
a Taylor series in $\alpha\eqdef\eps/\sqrt{n}$, 
the odd degree terms will cancel. Moreover, we can write
\begin{align*}
\sum_{\ell=0}^\infty \probaOf{ A=\ell }\left(1-\varphi(\eps,A)\right)^2 
&= \shortexpect_A\left[ \left(1-\varphi(\eps,A)\right)^2 \right] \\
&= \frac{1}{2}\shortexpect_{A\sim \poisson{k/2}}\left[ \left(1-\varphi(\eps,A)\right)^2 \right]  \\
&\quad+ \frac{1}{4}\shortexpect_{A\sim \poisson{k(1/2+\alpha)}}\left[ \left(1-\varphi(\eps,A)\right)^2 \right]  \\
&\quad+ \frac{1}{4}\shortexpect_{A\sim \poisson{k(1/2-\alpha)}}\left[ \left(1-\varphi(\eps,A)\right)^2 \right] \;.
\end{align*}
Now, we can rewrite
\begin{align*}
(1-\varphi(\eps,A))^2 
&= \left( 1 - \frac{e^{-k\alpha}(1+2\alpha)^\ell+e^{k\alpha}(1-2\alpha)^\ell}{2} \right)^2 \\
&= 1 - \left( e^{-k\alpha}(1+2\alpha)^\ell +e^{k\alpha}(1-2\alpha)^\ell \right) + \frac{e^{-2k\alpha}(1+2\alpha)^{2\ell}+ 2(1-4\alpha^2)^\ell +e^{2k\alpha}(1-2\alpha)^{2\ell}}{4} \;.
\end{align*}
For $b \in \{-1,0,1\}$, we have $\shortexpect_{A\sim \poisson{k(1/2+b\alpha)}}\left[ 1 \right] = 1$ (!), and (from the MGF of a Poisson distribution)
\begin{align*}
  e^{-k\alpha}\shortexpect_{A\sim \poisson{k(1/2+b\alpha)}}\left[ (1+2\alpha)^A \right]
%  &= e^{-k\alpha} e^{k(1/2+b\alpha)\cdot 2\alpha} \\
  &= e^{b\cdot 2\alpha^2 k} \\
  e^{k\alpha}\shortexpect_{A\sim \poisson{k(1/2+b\alpha)}}\left[ (1-2\alpha)^A \right]
%  &= e^{k\alpha} e^{k(1/2+b\alpha)\cdot -2\alpha} \\
  &= e^{-b\cdot 2\alpha^2 k} \;,
\end{align*}
as well as
\begin{align*}
  e^{-2k\alpha}\shortexpect_{A\sim \poisson{k(1/2+b\alpha)}}\left[ (1+2\alpha)^{2A} \right] 
%  &= e^{-2k\alpha} e^{k(1/2+b\alpha)\cdot (4\alpha+4\alpha^2)} \\
  &= e^{2k\alpha^2(1+2b+2b\alpha)} \\
  e^{2k\alpha}\shortexpect_{A\sim \poisson{k(1/2+b\alpha)}}\left[ (1-2\alpha)^{2A} \right] 
%  &= e^{2k\alpha} e^{k(1/2+b\alpha)\cdot (-4\alpha+4\alpha^2)} \\
  &= e^{2k\alpha^2(1-2b+2b\alpha)} \\
  2\shortexpect_{A\sim \poisson{k(1/2+b\alpha)}}\left[ (1-4\alpha^2)^{A} \right] 
%  &= 2 e^{k(1/2+b\alpha)\cdot (-4\alpha^2)} \\
  &= 2e^{-2k\alpha^2 - 4kb\alpha^3} \;.
\end{align*}
Gathering the terms, we get
\begin{align*}
\shortexpect_A\left[ \left(1-\varphi(\eps,A)\right)^2 \right]  
% &= \frac{1}{4}\Bigg( 2\left( 1 - 2 + \frac{e^{2k\alpha^2}+e^{-2k\alpha^2}}{2} \right) \\
% &\quad+ \Bigg( 1 - (e^{2k\alpha^2}+e^{-2k\alpha^2}) \\&\quad\quad+ \frac{ e^{2k\alpha^2(3+2\alpha)}+e^{2k\alpha^2(-1+2\alpha)} + 2e^{-2k\alpha^2 - 4k\alpha^3} }{4} \Bigg) \\
% &\quad+ \Bigg( 1 - (e^{-2k\alpha^2}+e^{2k\alpha^2})  \\&\quad\quad+ \frac{ e^{-2k\alpha^2(1+2\alpha)}+e^{2k\alpha^2(3-2\alpha)} + 2e^{-2k\alpha^2 + 4k\alpha^3} }{4} \Bigg)
%  \Bigg)\\
&= \frac{1}{16}\Big( -4(e^{2k\alpha^2}+e^{-2k\alpha^2}) + e^{2k\alpha^2(3+2\alpha)}\\
&\quad+e^{2k\alpha^2(-1+2\alpha)} + 2e^{-2k\alpha^2 - 4k\alpha^3}  \\
&\quad+ e^{-2k\alpha^2(1+2\alpha)}+e^{2k\alpha^2(3-2\alpha)} + 2e^{-2k\alpha^2 + 4k\alpha^3}
 \Big) \\
 &= O( k^2 \alpha^4 ) \tag{Taylor series expansion in $\alpha$} \;,
\end{align*}
giving that indeed
%% Series[-4 (Exp[2 k x^2] + Exp[-2 k x^2] ) + Exp[2 k x^2 (3 + 2 x)] +   Exp[2 k x^2 (-1 + 2 x)] + 2 Exp[-2 k x^2 - 4 k x^3] +   Exp[-2 k x^2 (1 + 2 x)] + Exp[2 k x^2 (3 - 2 x)] +   2 Exp[-2 k x^2 + 4 k x^3], {x, 0, 4}]
%% -> 32 k^2 x^4 + O(x^5)
$
\sum_{\ell=0}^\infty \probaOf{ A=\ell }\left(1-\varphi(\eps,A)\right)^2
 = O\left( \frac{\eps^4 k^2}{n^2} \right). 
$
This completes the proof.
\end{proof}

\begin{proof}[Proof of Lemma~\ref{lemma:noinstances:far:unbalanced}]
Using Hellinger distance as a proxy will only result in an $\Omega(\eps^2)$ lower bound on the distance, 
so we compute it explicitly instead: in what follows, $e^{(j)}\in\{0,1\}^n$ denotes the basis vector with $e^{(j)}_i = \indic{\{i=j\}}$. 
Fix any vector $b=(b_1,\dots, b_n)\in\{0,1\}^n$ such that $\abs{b}\in[n/3, 2n/3]$, 
and let $P$ be the corresponding distribution from the support of $\dno$.
\begin{align*}
  \normone{P-P^\ast}  
  &\geq \sum_{j=1}^n \abs{ P(e^{(j)}) - P^\ast(e^{(j)}) } \\
  &= \sum_{j=1}^n \Big\lvert \frac{1+(-1)^{b_j}\eps}{n}\prod_{i\neq j}\left(1-\frac{1+(-1)^{b_i}\eps}{n}\right)  - \frac{1}{n}\left(1-\frac{1}{n}\right)^{n-1} \Big\rvert \\
  &= \frac{1}{n}\left(1-\frac{1}{n}\right)^{n-1}\sum_{j=1}^n \Big\lvert ( 1+(-1)^{b_j}\eps )\prod_{i\neq j}\left(1-\frac{(-1)^{b_i}\eps}{n-1}\right)  - 1 \Big\rvert \;.
\end{align*}
Each summand can be bounded from above as follows:
\begin{align*}
  \left( 1 - \frac{\eps}{n-1}\right)^{2n/3} \leq \prod_{i\neq j}\left(1-\frac{(-1)^{b_i}\eps}{n-1}\right) \leq \left( 1+ \frac{\eps}{n-1}\right)^{2n/3} \;,
\end{align*}
where the last inequality follows from our assumption on $\abs{b}$. In turn, this gives that
\begin{itemize}
  \item If $b_j = 0$, 
  \begin{align*}
  (1+(-1)^{b_j}\eps )\prod_{i\neq j}\left(1-\frac{(-1)^{b_i}\eps}{n-1}\right)  - 1 
    &\geq (1+\eps)\left( 1 - \frac{\eps}{n-1}\right)^{2n/3} - 1 =  \Omega\left(\eps\right) \;.
  \end{align*}
  \item If $b_j = 1$, 
  \begin{align*}
  1 - (1+(-1)^{b_j}\eps )\prod_{i\neq j}\left(1-\frac{(-1)^{b_i}\eps}{n-1}\right) 
  &\geq 1- (1-\eps)\left( 1 + \frac{\eps}{n-1}\right)^{2n/3} = \Omega\left(\eps\right) \;.
  \end{align*}
\end{itemize}
Since $\frac{1}{n}\left(1-\frac{1}{n}\right)^{n-1} = \frac{e^{-1}+o(1)}{n}$, we get $\normone{P-P^\ast} = \Omega(\eps)$. 
The lemma now follows from the fact that a uniformly random $b\in\{0,1\}^n$ satisfies $\abs{b}\in[n/3, 2n/3]$ with probability $1-2^{-\Omega(n)}$.
\end{proof}

\section{Omitted Proofs \texorpdfstring{from~\cref{sec:prelim}}{from~\cref{sec:prelim}}} \label{sec:prelim:distances:proofs}
We provide in this section the proofs of the inequalities stated in the preliminaries (\cref{sec:prelim}).

\begin{proof}[Proof of Lemma~\ref{lem:prod:kl}]
Using properties of the KL-divergence:
\begin{align*}
  \dkl{P}{Q} &= \dkl{ P_1\otimes\dots\otimes P_n }{ Q_1\otimes\dots\otimes Q_n } = \sum_{i=1}^n \dkl{P_i}{Q_i}
\end{align*}
so it suffices to show that $2(p_i-q_i)^2 \leq \dkl{P_i}{Q_i} \leq \frac{(p_i-q_i)^2}{q_i(1-q_i)}$ for all $i\in[n]$.
Since $P_i = \bernoulli{p_i}$ and $Q_i = \bernoulli{q_i}$, we can write
\begin{align*}
    \dkl{P_i}{Q_i} &= p_i\ln \frac{p_i}{q_i} + (1-p_i)\ln \frac{1-p_i}{1-q_i}
\end{align*}
Defining $f\colon(0,1)^2\to\R$ as $f(x,y) \eqdef x\ln \frac{x}{y} + (1-x)\ln \frac{1-x}{1-y}$, we thus have to show that
$
      2(x-y)^2 \leq f(x,y) \leq \frac{(x-y)^2}{y(1-y)}
$
for all $x,y\in(0,1)$.

We begin by the upper bound: fixing any $x,y\in (0,1)$ and setting $\delta\eqdef x-y$, we have
\begin{align*}
    f(x,y) &= x\ln \left(1+\frac{\delta}{y}\right) + (1-x)\ln \left(1-\frac{\delta}{1-y}\right) \\
    &\leq \frac{x}{y} \delta - \frac{1-x}{1-y} \delta = \frac{\delta^2}{y(1-y)}
\end{align*}
from $\ln(1+u)\leq u$ for all $u\in(-1,\infty)$.

Turning to the lower bound, we fix any $y\in(0,1)$ and consider the auxiliary function $h_y\colon(0,1)\to \R$ defined by
$
h_y(x) = f(x,y) - 2(x-y)^2.
$
From $h_y''(x) = \frac{(1-2x)^2}{x(1-x)} \geq 0$, we get that $h_y$ is convex, i.e. $h_y'$ is non-decreasing. Since $h_y'(x) = \ln \frac{x(1-y)}{(1-x)y}-4(x-y)$, we have $h'_y(y) = 0$, and in turn we get that $h_y$ is non-increasing on $(0,y]$ and non-decreasing on $[y,1)$. Since $h_y(y) = 0$, this leads to $h_y(x) \geq 0$ for all $x\in(0,1)$, i.e. $f(x,y) \geq 2(x-y)^2$.
\end{proof}

\begin{proof}[Proof of Lemma~\ref{lem:prod:dtv:hellinger}]
Recall that for any pair of distributions, 
we have that $\hellinger{P}{Q}^2 \leq \frac{1}{2}\normone{P-Q} \leq \sqrt{2} \hellinger{P}{Q},$ 
where $\hellinger{P}{Q}$ denotes the Hellinger distance between $P,Q$. Therefore, it is enough to show that
\begin{equation}
  \min\left(c', \frac{1}{4} \sum_{i=1}^n (p_i-q_i)^2 \right) \leq \hellinger{P}{Q}^2 \leq \sum_{i=1}^n \frac{(p_i-q_i)^2}{q_i(1-q_i)}
\end{equation}
for some absolute constant $c' >0$. (We will show $c' = 1-e^{-3/2} \simeq 0.78$.)

Since $P,Q$ are product measures,
\begin{align*}
\hellinger{P}{Q}^2 &= 1-\prod_{i=1}^n (1 - \hellinger{P_i}{Q_i}^2) = 1- \prod_{i=1}^n (\sqrt{p_i q_i} + \sqrt{(1-p_i)(1-q_i)}) \;.
\end{align*}

We start with the lower bound. Noting that for any $x,y\in(0,1)$ it holds that
\[
\sqrt{xy}+\sqrt{(1-x)(1-y)} \leq 1-\frac{1}{2}(x-y)^2
\]
(e.g., by observing that the function $x\in(0,1)\mapsto 1-\frac{1}{2}(x-y)^2 - (\sqrt{xy}+\sqrt{(1-x)(1-y)})$ is minimized at $y$, where it takes value $0$), we get
\begin{align*}
\hellinger{P}{Q}^2 
&\geq 1- \prod_{i=1}^n \left( 1-\frac{1}{2}(p_i-q_i)^2 \right) \\
&\geq \min( 1-e^{-3/2}, \frac{1}{4} \sum_{i=1}^n (p_i-q_i)^2 ) \\
&= \min( 1-e^{-3/2}, \frac{1}{4}\normtwo{p-q}^2
)
\end{align*}
where we relied on the inequality $1 - \prod_{i=1}^n (1-x_i) \geq \frac{1}{2}\sum_{i=1}^n x_i$ for $(x_1,\dots,x_n)\in [0,1]$:
\begin{align*}
1 - \prod_{i=1}^n (1-x_i) 
&= 1 - e^{\sum_{i=1}^n \ln(1-x_i)} \geq 1 - e^{-\sum_{i=1}^n x_i} 
\geq 1 - \left(1 - \frac{1}{2}\sum_{i=1}^n x_i \right).
\end{align*}
(the last inequality being true for $\sum_{i=1}^n x_i \leq \frac{3}{2}$, i.e. $\normtwo{p-q}^2 \leq 3$).

Turning to the upper bound, the elementary inequality $2\sqrt{xy} = x+y - (\sqrt{x}-\sqrt{y})^2,$ $x,y>0,$ gives that
\begin{align*}
\sqrt{p_i q_i} + \sqrt{(1-p_i)(1-q_i)} &\geq 1-  \frac{(p_i-q_i)^2}{(p_i+q_i)(2-p_i-q_i)} = 1-z_i \;.
\end{align*}
Therefore,
\begin{align}\label{eq:hellinger:intermediate:ub}
\hellinger{P}{Q}^2 &\leq (1- \prod_{i=1}^n (1-z_i)) \leq \sum_{i=1}^n z_i \notag\\
&= \sum_{i=1}^n \frac{(p_i-q_i)^2}{(p_i+q_i)(2-p_i-q_i)} \notag\\
&\leq \sum_{i=1}^n \frac{(p_i-q_i)^2}{q_i(1-q_i)} \;,
\end{align}
where the third-to-last inequality follows from the union bound, and the last from the simple fact that $(p_i+q_i)(2-p_i-q_i) \geq q_i(1-q_i)$.
\end{proof}

\begin{proof}[Proof of Lemma~\ref{lemma:kl:bn}]
Let $A$ and $B$ be two distributions on $\{0,1\}^d$. Then we have:
\begin{equation} \label{eq:kl}
\dkl{A}{B} = \sum_{x \in \{0,1\}^d} \proba_A[x] \ln  \frac{ \proba_A[x] }{ \proba_B[x] } ;.
\end{equation}
For a fixed $i \in [d]$, the events $\Pi_{i,a}$ form a partition of $\{0,1\}^d$.
Dividing the sum above into this partition, we obtain
\begin{align*}
\dkl{A}{B} 
&= \sum_a \sum_{x \in \Pi_{i,a}} \proba_A[x] \ln  \frac{ \proba_A[x] }{ \proba_B[x] } \\
&= \sum_a \proba_A[\Pi_{i,a}] \ \cdot\!\!\!\sum_{x \in \{0,1\}^d}\!\!\! \proba_{A \mid \Pi_{i,a}}[x] \left(\ln \frac{ \proba_A[\Pi_{i,a}] }{ \proba_B[\Pi_{i,a}] } + \ln \frac{ \proba_{A \mid \Pi_{i,a}}[x] }{ \proba_{B \mid \Pi_{i,a}}[x] } \right)\\
& = \sum_a \proba_A[\Pi_{i,a}] (\ln \frac{ \proba_A[\Pi_{i,a}] }{ \proba_B[\Pi_{i,a}] } + \dkl{ A \mid \Pi_{i,a}  }{  B \mid \Pi_{i,a} }\;.
\end{align*}

Let $P_{\leq i}$ be the distribution over the first $i$ coordinates of $P$
and define $Q_{\leq i}$ similarly for $Q$. Let $P_i$ and $Q_i$ be the distribution
of the $i$-th coordinate of $P$ and $Q$ respectively.
We will apply the above to $P_{\leq i-1}$ and $P_{\leq i}$.
First note that the $i$-th coordinate of $P_{\leq i} \mid \Pi_{i,a}$
and $Q_{\leq i}\mid\Pi_{i,a}$ is independent of the others,
thus we have (which likely follows from standard results):
\begin{align*}
\dkl{ P_{\leq i}\mid\Pi_{i,a} }{  Q_{\leq i}\mid\Pi_{i,a} } 
&= \sum_{x \in \{0,1\}^i} \proba_{P_{\leq i}\mid\Pi_{i,a}}[x] \ln  \frac{ \proba_{P_{\leq i}\mid\Pi_{i,a}}[x] }{ \proba_{Q_{\leq i}\mid\Pi_{i,a}}[x] }  \\
&= \sum_{x \in \{0,1\}^i} \proba_{P_{\leq i-1}\mid\Pi_{i,a}}[x_{\leq i-1}]\proba_{P_{i}\mid\Pi_{i,a}}[x_i] \cdot\left(\ln  \frac{ \proba_{P_{\leq i-1}\mid\Pi_{i,a}}[x_{\leq i-1}] }{ \proba_{Q_{\leq i-1}\mid\Pi_{i,a}}[x_{\leq i-1}]} + \ln \frac{ \proba_{P_i\mid\Pi_{i,a}}[x_i] }{ \proba_{Q_{\leq i}\mid\Pi_{i,a}}[x_i] }\right) \\
&= \sum_{x \in \{0,1\}^{i-1}} \proba_{P_{\leq i-1}\mid\Pi_{i,a}}[x] \ln \frac{ \proba_{P_{\leq i-1}\mid\Pi_{i,a}}[x] }{ \proba_{Q_{\leq i-1}\mid\Pi_{i,a}}[x] } 
+ \sum_{y \in \{0,1\}} \proba_{P_{i}\mid\Pi_{i,a}}[y] \ln \frac{ \proba_{P_i\mid\Pi_{i,a}}[y] }{ \proba_{Q_{\leq i}\mid\Pi_{i,a}}[y] }\\
&= \dkl{ P_{\leq i-1}\mid\Pi_{i,a} }{  Q_{\leq i-1}\mid\Pi_{i,a} } + \dkl{ P_{i}\mid\Pi_{i,a} }{  Q_{i}\mid\Pi_{i,a} }
\end{align*}

Thus, we have:
\begin{align*}
  \dkl{ P_{\leq i}  }{   Q_{\leq i} } 
  &= \sum_a \proba_{P_{\leq i}} [\Pi_{i,a}] (\ln \frac{ \proba_{P_{\leq i}}[\Pi_{i,a}] }{ \proba_{Q_{\leq i}}[\Pi_{i,a}] }  + \dkl{ P_{\leq i} \mid \Pi_{i,a}  }{  Q_{\leq i} \mid \Pi_{i,a} } \\
  &= \sum_a \proba_{P_{\leq i}} [\Pi_{i,a}] \Big(\ln \frac{ \proba_{P_{\leq i}}[\Pi_{i,a}] }{ \proba_{Q_{\leq i}}[\Pi_{i,a}] }  
  + \dkl{ P_{\leq i-1} \mid \Pi_{i,a}  }{  Q_{\leq i-1} \mid \Pi_{i,a} } 
  + \dkl{ P_{i} \mid \Pi_{i,a}  }{  Q_{i} \mid \Pi_{i,a} } \Big) \\
  &=  \dkl{ P_{\leq i-1}  }{  Q_{\leq i -1} } + \sum_a \proba_P[\Pi_{i,a}] \dkl{ P_{i} \mid \Pi_{i,a}  }{  Q_{i} \mid \Pi_{i,a} }  
 \end{align*}

By induction on $i$, we get
\[
\dkl{P}{Q} = \sum_{(i,a) \in S} \proba_P[\Pi_{i,a}]  \dkl{ P_{i}\mid\Pi_{i,a}  }{   Q_{i}\mid\Pi_{i,a}}.
\]
Now the distributions $P_{i}\mid\Pi_{i,a}$ and $Q_{i}\mid\Pi_{i,a}$ are Bernoulli random variables
with means $p_{i,a}$ and $q_{i,a}$. For $p,q \in [0,1]$, we have:
\begin{align*}
\dkl{ \bernoulli{p}  }{  \bernoulli{q} } 
= p \ln \frac{p}{q} + (1-p) \ln\frac{1-p}{1-q} \leq \frac{(p-q)^2}{q(1-q)}
\end{align*}
as in the proof of~\cref{claim:toy:2:distance:means}. On the other hand, studying for instance the function $f_q\colon (0,1) \to \R$ defined by $f_q(p) = 
\frac{ p \ln \frac{p}{q} + (1-p) \ln\frac{1-p}{1-q}}{(p-q)^2}$ (extended by continuity at $q$), we get
\[
f_q(p) \geq f_q(1-q) \geq 2
\]
for all $p,q\in(0,1)^2$. This shows the lower bound.
\end{proof}

\section{Omitted Proofs \texorpdfstring{from~\cref{ssec:product-lower,ssec:identity-unknown-upper}}{from Sections~\ref{ssec:product-lower} and~\ref{ssec:identity-unknown-upper}}}\label{sec:misc:proofs}

\begin{unnumberedfact}[\cref{fact:mutualinfo:expansion}]
  \begin{equation}
  \mutualinfo{X}{A} \leq \sum_{a=0}^\infty \probaOf{ A=a }\left(1-\frac{ \probaCond{A=a }{ X=1}  }{ \probaCond{A=a }{ X=0} }\right)^2.
  \end{equation}
\end{unnumberedfact}
\begin{proof}
  For $a\in\mathbb{N}$, we write $p_a = \probaCond{X=0}{A=a}$ and $q_a = \probaCond{X=1}{A=a}$, so that $p_a+q_a=1$.  By definition,
    \begin{align*}
      \mutualinfo{X}{A} 
      &= \sum_{a=0}^\infty \probaOf{A=a}  \cdot\sum_{x\in\{0,1\}} \probaCond{X=x}{A=a} \log \frac{\probaCond{X=x}{A=a}}{\probaOf{X=x}} \\
      &= \sum_{a=0}^\infty \probaOf{A=a} \cdot\left( 
              p_a \log \frac{p_a}{\probaOf{X=1}} + q_a \log \frac{q_a}{\probaOf{X=0}} 
        \right) \\
      &= \sum_{a=0}^\infty \probaOf{A=a} \left( 
              p_a \log \left( 2p_a \right) + q_a \log \left( 2q_a \right)
        \right) \\
      &= \sum_{a=0}^\infty \probaOf{A=a} \left( 
              (1-q_a) \log \left(1-q_a\right) + q_a \log \left(q_a \right) + 1
        \right) \\
      &\leq \sum_{a=0}^\infty \probaOf{A=a} \left( 1-\frac{q_a}{1-q_a} \right)^2 \\
      &= \sum_{a=0}^\infty \probaOf{A=a} \left( 1-\frac{q_a}{p_a} \right)^2
    \end{align*}
where for the last inequality we rely on the fact that the binary entropy satisfies $h(x) \geq 1 - \left(1-\frac{x}{1-x}\right)^2$ for all $x\in [0,1)$.
\end{proof}

\begin{proof}[Alternative proof of~\cref{lemma:lb:key:unbalanced}]
We proceed as in the proof of~\cref{lemma:lb:key}, first writing
\begin{align*}
\mutualinfo{X}{A} &= \sum_{\ell=0}^\infty \probaOf{ A=a }\left(1-\frac{ \probaCond{A=a }{ X=1}  }{ \probaCond{A=a }{ X=0} }\right)^2 \\
&= \shortexpect_A\left[ \left(1-\frac{ \probaCond{A=a }{ X=1}  }{ \probaCond{A=a }{ X=0} }\right)^2 \right] 
\end{align*}
and noticing that
$
z_a \eqdef \frac{ \probaCond{A=a }{ X=1}  }{ \probaCond{A=a }{ X=0} } = \frac{e^{-k\eps/n}(1+\eps)^a+e^{k\eps/n}(1-\eps)^a}{2}
$. This leads to 
\begin{align}\label{eq:expansion:mutualinf:bound}
4(1-z_a)^2 
&= 4 - 4\left( e^{-\frac{k}{n}\eps} (1+\eps)^a + e^{\frac{k}{n}\eps} (1-\eps)^a \right)  
+ \left( e^{-2\frac{k}{n}\eps} (1+\eps)^{2a} + e^{2\frac{k}{n}\eps} (1-\eps)^{2a} + 2(1-\eps^2)^a \right)
\end{align}
We also have, by definition of $A$, that for any $z\in\mathbb{R}$
\begin{align}\label{eq:pgf:a}
  \shortexpect_A[z^A] &= \frac{1}{2}\shortexpect_{A\sim\poisson{\frac{k}{n}}}[z^A]
  + \frac{1}{4}\shortexpect_{A\sim\poisson{\frac{k}{n}(1+\eps)}}[z^A] + \frac{1}{4}\shortexpect_{A\sim\poisson{\frac{k}{n}(1-\eps)}}[z^A] \notag\\
  &= 
  \frac{1}{2}e^{\frac{k}{n}(z-1)}
  + \frac{1}{4}e^{\frac{k}{n}(1+\eps)(z-1)}
  + \frac{1}{4}e^{\frac{k}{n}(1-\eps)(z-1)}
\end{align}
from the expression of the probability generating function of a Poisson random variable. For any $\beta\in\{-1,1\}$, we therefore have
\begin{align*}
  4e^{-\frac{k}{n}\beta\eps}\shortexpect_A[(1+\beta\eps)^A]  
  % &= 2e^{\frac{k}{n}\beta\eps}e^{-\frac{k}{n}\beta\eps} + e^{\frac{k}{n}(1+\eps)\beta\eps}e^{-\frac{k}{n}\beta\eps} \\
  % &\qquad+ e^{\frac{k}{n}(1-\eps)\beta\eps}e^{-\frac{k}{n}\beta\eps} \\
  &= 2 + e^{\frac{k}{n}\beta\eps^2} + e^{-\frac{k}{n}\beta\eps^2}
  \\
  4e^{-2\frac{k}{n}\beta\eps}\shortexpect_A[(1+\beta\eps)^{2A}] 
  %  &= 2e^{\frac{k}{n}(2\beta\eps+\eps^2)}e^{-2\frac{k}{n}\beta\eps} \\
  % &\qquad+ e^{\frac{k}{n}(1+\eps)(2\beta\eps+\eps^2)}e^{-2\frac{k}{n}\beta\eps} \\
  % &\qquad+ e^{\frac{k}{n}(1-\eps)(2\beta\eps+\eps^2)}e^{-2\frac{k}{n}\beta\eps} \\
  &=  2e^{\frac{k}{n}\eps^2} + e^{\frac{k}{n}((1+2\beta)\eps^2+\eps^3)} + e^{\frac{k}{n}((1-2\beta)\eps^2-\eps^3)}
\end{align*}
and
\begin{align*}
  4\cdot 2\shortexpect_A[(1-\eps^2)^{A}]
%      &= 4e^{\frac{k}{n}\eps^2} + 2e^{\frac{k}{n}(1+\eps)\eps^2} + 2e^{\frac{k}{n}(1-\eps)\eps^2}\\
      &= 4e^{\frac{k}{n}\eps^2} + 2e^{\frac{k}{n}(\eps^2+\eps^3)} + 2e^{\frac{k}{n}(\eps^2-\eps^3)}
\end{align*}
Combining~\cref{eq:expansion:mutualinf:bound} and the above, we obtain
\begin{align*}
  16\shortexpect_A\left[ \left(1-z_A\right)^2 \right]  
  &= 16 - 4\left( 4e^{-\frac{k}{n}\eps} \shortexpect_A[(1+\eps)^A] + 4e^{\frac{k}{n}\eps} \shortexpect_A[(1-\eps)^A] \right) \\
      &\qquad + \big( 4e^{-2\frac{k}{n}\eps} \shortexpect_A[(1+\eps)^{2A}] + 4e^{2\frac{k}{n}\eps} \shortexpect_A[(1-\eps)^{2A}] \\
      &\qquad + 8\shortexpect_A[(1-\eps^2)^A] \big) \\
%   &= 16 - 8 \left( 2 + e^{\frac{k}{n}\eps^2} + e^{-\frac{k}{n}\eps^2} \right) \\
%         &\qquad+ \Big( 4e^{\frac{k}{n}\eps^2} 
%                 + e^{\frac{k}{n}(3\eps^2+\eps^3)} \\
%                 &\qquad+ e^{\frac{k}{n}(-\eps^2-\eps^3)} 
%                 + e^{\frac{k}{n}(-\eps^2+\eps^3)} + e^{\frac{k}{n}(3\eps^2-\eps^3)} \\
%         &\qquad + 4e^{\frac{k}{n}\eps^2} + 2e^{\frac{k}{n}(\eps^2+\eps^3)} + 2e^{\frac{k}{n}(\eps^2-\eps^3)} \Big) \\
  &= 8 - 8 \left(e^{\frac{k}{n}\eps^2} + e^{-\frac{k}{n}\eps^2} \right) 
        + \left( e^{\frac{k}{n}3\eps^2}+  e^{-\frac{k}{n}\eps^2} + 2e^{\frac{k}{n}\eps^2} \right) \left(e^{\frac{k}{n}\eps^3}+e^{-\frac{k}{n}\eps^3}\right).
\end{align*}
% g[x_] := 8 - 8 (  Exp[m x^2] + Exp[-m x^2] ) + (Exp[m 3 x^2] + Exp[-m x^2] + 2 Exp[m x^2]) (Exp[m x^3] + Exp[-m x^3])
% Series[ g[x], {x, 0, 5}]
A Taylor expansion of this expression (in $\frac{k\eps^2}{n}$ for the first two parentheses, and $\frac{k\eps^3}{n}$ for the last) shows that
\begin{align*}
  \shortexpect_A\left[ \left(1-z_A\right)^2 \right] = O\left( \frac{k^2\eps^4}{n^2} \right)
\end{align*}
as claimed.
\end{proof}

%%%%%%%%%%%%%%%%%%%%%%%%%%%%%%%%%%%%%

\begin{proof}[Proof of~\cref{lemma:unknown:tree:structure:lipschitz:mutualinfo}]
From the definition of $\mutualinfo{X}{Y} = \sum_{(x,y)\in\{0,1\}} \probaOf{X=x,Y=y} \log\frac{\probaOf{X=x,Y=y}}{\probaOf{X=x}\probaOf{Y=y}}$, it is straightforward to check that for $X,Y$ taking values in $\{0,1\}$
\begin{align*}
  \mutualinfo{X}{Y} 
  &= \expect{XY}\log\frac{\expect{XY}}{\expect{X}\expect{Y}}  \\
    &\qquad+ (\expect{X}-\expect{XY})\log\frac{\expect{X}-\expect{XY}}{\expect{X}(1-\expect{Y})}\\
    &\qquad+ (\expect{Y}-\expect{XY})\log\frac{\expect{Y}-\expect{XY}}{(1-\expect{X})\expect{Y}}\\
    &\qquad+ (1-\expect{X}-\expect{Y}+\expect{XY})\cdot\log\frac{1-\expect{X}-\expect{Y}+\expect{XY}}{(1-\expect{X})(1-\expect{Y})} \\
    &= f( \expect{X},\expect{Y},\expect{XY}  )
\end{align*}
for $f$ defined by $f(x,y,z) \eqdef z \log\frac{z}{xy}+ (x-z) \log\frac{x-z}{x(1-y)}+ (y-z) \log\frac{y-z}{(1-x)y} + (1-x-y+z)\log\frac{1-x-y+z}{(1-x)(1-y)}$.

The domain of definition of $f$ is the subset $\Omega\subseteq [0,1]^3$ defined 
by (recalling that $x,y,z$ correspond to $\expect{X},\expect{Y},\expect{XY}$ for ${X},Y\in\{0,1\}$)
\begin{align*}
  0&\leq x,y\leq 1\\
  0&\leq z\leq \min(x,y)\\
  0&\leq z \leq \sqrt{xy} \tag{Cauchy--Schwarz}\\
  0&\leq 1+z-x-y
\end{align*}
Given the $c$-balancedness assumption on $P$, $\Omega_c\subseteq \Omega$ satisfies the further following constraints:
\begin{align*}
  c&\leq x,y\leq 1-c\\
  c^2&\leq \frac{z}{x},\frac{z}{y}\leq 1-c^2\\
  c^2&\leq \frac{1+z-x-y}{1-x},\frac{1+z-x-y}{1-y}\leq 1-c^2
\end{align*}
by Baye's rule and recalling that $1+z-x-y$ corresponds to the three (equal) quantities $\probaOf{X=0,Y=0}$, $\probaCond{X=0}{Y=0}\probaOf{Y=0}$, $\probaCond{Y=0}{X=0}\probaOf{X=0}$; while $1-x,1-y$ correspond to
$\probaOf{X=0},\probaOf{Y=0}$ respectively.%\cnote{Check that... too sketchy? Need more details?}

\noindent One can then check that
\begin{align*}
    \frac{\partial{f}}{\partial{x}}(x,y,z) &= \log\frac{(1-x)(x-z)}{x(1+z-x-y)},\\
    \frac{\partial{f}}{\partial{y}}(x,y,z) &= \log\frac{(1-y)(y-z)}{y(1+z-x-y)},\\
    \frac{\partial{f}}{\partial{z}}(x,y,z) &= \log\frac{z(1+z-x-y)}{(x-z)(y-z)}
\end{align*}
and therefore, on $\Omega_c$, that
\begin{align*}
    \norminf{\frac{\partial{f}}{\partial{x}}}
    &=\norminf{\frac{\partial{f}}{\partial{y}}}\\
    &\leq \sup_{(x,y,z)\in\Omega_c} \abs{\log\frac{x-z}{x}}+\abs{\log\frac{1+z-x-y}{1-x}}\\
    &\leq 2\log\frac{1}{c^2} 
    = 4\log\frac{1}{c}.
\end{align*}
Similarly,
\begin{align*}
    \norminf{\frac{\partial{f}}{\partial{z}}}
    &\leq \sup_{(x,y,z)\in\Omega_c} \abs{\log\frac{(x-z)(y-z)}{z}}+\abs{\log(1+z-x-y)}\\
    &\leq \sup_{(x,y,z)\in\Omega_c}\abs{ \log\frac{(1-c^2)y}{c^4xy}} +\abs{\log\frac{1}{c^2x}} \\
    &\leq \log\frac{1}{c^5} +\log\frac{1}{c^3} = 8\log\frac{1}{c} \;.
\end{align*}
So overall, $f$ is $\lambda$-Lipschitz (with regard to the $\norminf{\cdot}$ norm) on $\Omega_c$, for $\lambda = \norminf{\frac{\partial{f}}{\partial{x}}}+\norminf{\frac{\partial{f}}{\partial{y}}}+\norminf{\frac{\partial{f}}{\partial{z}}} \leq 16\log\frac{1}{c}$.
\end{proof}
\end{document}